\pdfoutput=1
\documentclass[11pt]{extarticle}

\usepackage[thmstyles]{andrews}
\usepackage{andrews_bib}

\usepackage[style=alphabetic,
backend=biber,
backref,
url=false,
doi=true,
useprefix=true,
maxbibnames=99,
maxcitenames=5,
mincitenames=3,
maxalphanames=5,
minalphanames=3]{biblatex}
\DeclareMathOperator{\Newton}{Newton}
\DeclareMathOperator{\res}{res}
\DeclareMathOperator{\Syl}{Syl}
\DeclareMathOperator{\Bez}{Bez}
\DeclareMathOperator{\adj}{adj}
\DeclareMathOperator{\disc}{disc}
\DeclareMathOperator{\rev}{rev}
\DeclareMathOperator{\lcm}{lcm}

\DeclareMathOperator{\Thr}{Thr}

\renewcommand{\K}{\mathbb{K}}

\newcommand{\sel}{\mathsf{select}}

\title{Constant-Depth Arithmetic Circuits for Linear Algebra Problems}
\date{September 22, 2025}
\author{Robert Andrews\thanks{Cheriton School of Computer Science, University of Waterloo. Part of this work was done as a member in the School of Mathematics at the Institute for Advanced Study, and was supported by NSF grant CCF-1900460 and by the Erik Ellentuck Endowed Fellowship Fund. Email: randrews@uwaterloo.ca.} \and Avi Wigderson\thanks{School of Mathematics, Institute for Advanced Study. Supported by NSF grant CCF-1900460. Email: avi@ias.edu.}}

\begin{document}

\maketitle

\begin{abstract}
	We design polynomial size, constant depth (namely, $\AC^0_\F$) arithmetic formulae for the greatest common divisor (GCD) of two polynomials, as well as the related problems of the discriminant, resultant, B\'{e}zout coefficients, squarefree decomposition, and the inversion of structured matrices like Sylvester and B\'{e}zout matrices. 
	Our GCD algorithm extends to any number of polynomials. 
	Previously, the best known arithmetic formulae for these problems required super-polynomial size, regardless of depth.

	These results are based on new algorithmic techniques to compute various symmetric functions in the roots of polynomials, as well as manipulate the multiplicities of these roots, without having access to them. 
	These techniques allow $\AC^0_\F$ computation of a large class of linear and polynomial algebra problems, which include the above as special cases.

	We extend these techniques to problems whose inputs are \emph{multivariate} polynomials, which are represented by constant-depth arithmetic circuits. 
	Here too we solve problems such as computing the GCD and squarefree decomposition in $\AC^0_\F$. 
\end{abstract}

\tableofcontents

\section{Introduction}

\subsection{Background}

Arithmetic complexity theory studies the computation of polynomials (and rational functions) using the basic arithmetic operations over a field. 
Like Boolean complexity theory, it is a vast, developed, and very active field, with its own computational models, complexity classes, algorithms, lower bounds, reductions, and complete problems.
Extensive texts on different parts of the field include \textcite{vzGG13, Burgisser2000, SY10}.
Our main model of computation is arithmetic circuits, focusing on their size (up to polynomial factors), and in more detail on their depth (up to constant factors).

Arithmetic algorithms were devised by mathematicians for centuries, long before computers ushered in applications in symbolic computation and computer algebra, which have invigorated the theoretical study of arithmetic complexity.
The Euclidean algorithm for computing the GCD (originally described for integers, but works equally well for polynomials) is perhaps the oldest nontrivial algorithm to appear in print.
Another is Gaussian elimination (whose origins also appear in ancient texts \cite{Grcar11b,Hart11}) for computing the determinant of a matrix, which is essential for linear algebra.
In a sense, the story below concerns the relative difficulty of these two basic problems, determinant and GCD, from the viewpoint of parallel computation.

The development of parallel computing architectures invited the design of fast parallel algorithms, in both the Boolean and arithmetic settings.
The GCD and determinant algorithms above, while polynomial time, are described in a way which looks ``inherently sequential,'' and parallelizing them presented a nontrivial challenge.
The breakthrough of Csanky in 1976 \cite{Csanky76}, providing an arithmetic $\NC^2$ (polynomial size, $O(\log^2 n)$ depth) algorithm for the determinant, was followed by a decade of intense activity in which many more parallel algorithms were developed.
One natural consequence was that practically all linear algebra problems, such as inverting matrices, solving systems of linear equations, and computing the rank of a matrix, were also in arithmetic $\NC^2$, which we will denote by $\NC^2_\F$.
Many other problems in polynomial algebra were known or were found to have linear algebraic formulations, and so could be reduced to the determinant.
Examples include polynomial division with remainder, decoding of (some) linear codes, and the GCD problem for (any number of) univariate polynomials~\cite{vonzurGathen84}.\footnote{A minor but important point which should be mentioned is that arithmetic circuits formally cannot compute discontinuous functions like GCD, and one has to add to them (as is standard in the field) the ability of branching on testing of a given field element is zero or not.}

A more general understanding of the parallel complexity of arithmetic computation followed another breakthrough, this time of Hyafil in 1979 \cite{Hyafil79}, later sharpened by \textcite{VSBR83}, who gave a generic \emph{depth reduction} of arithmetic circuits.
Using as input size both the number of input variables as well as the degree of the computed polynomials, they showed that {\em any} polynomial-size circuit can be parallelized, namely simulated in similar size and $O(\log^2 n)$ depth! 
In other words, in the arithmetic setting, \emph{every} sequential algorithm can be efficiently parallelized: $\P = \NC^2$.
In particular, fast parallel algorithms for determinant and GCD could be derived directly from their sequential analogues via this simulation.
(Note that such a collapse is believed \emph{not} to hold in the Boolean setting, and in particular \emph{Integer} GCD is one canonical example of a problem believed not to have a shallow Boolean circuit.)

While very satisfying, the above general results above get stuck at $O(\log^2 n)$ depth, and though the circuits are of polynomial size, the resulting formulae have quasi-polynomial size $n^{O(\log n)}$.
Which of these problems have polynomial-size formulae, namely are in the class $\NC^1_\F$ of logarithmic depth (with bounded fan-in gates)? 
How about constant parallel time? 
There are clearly natural linear and polynomial algebra problems that can be performed even in $\AC^0_\F$, i.e., have polynomial size formulae of \emph{constant depth}, allowing gates of unbounded fan-in.
Easy examples include polynomial addition, multiplication, univariate polynomial evaluation, and interpolation (obtaining the coefficients of a polynomial from point evaluations). 

The same golden decade (from the mid 1970's to mid 1980's, summarized beautifully in von zur Gathen's 1986 survey \cite{vonzurGathen86survey}) provided further nontrivial algorithms putting certain polynomial and linear algebra problems in $\AC^0_\F$.
(In some cases they were stated as $\NC^1_\F$ algorithms, but it is easy to see that they can also be implemented in $\AC^0_\F$.)
\textcite{Bini84} devised an ingenious $\AC^0_\F$ algorithm to invert \emph{triangular} Toeplitz matrices\footnote{Toeplitz matrices have constant diagonals.}.
This was used by \textcite{BP85} to give an $\AC^0_\F$ algorithm for polynomial division with remainder.
(We will give very different $\AC^0_\F$ algorithms for these problems.)
Another important example was Ben-Or's $\AC^0_\F$ circuit to compute the elementary symmetric polynomials.\footnote{The reader unfamiliar with this gem is encouraged to find any efficient algorithm for them.}
As these polynomials are natural arithmetic analogues of Boolean threshold functions, this reveals the surprising power of arithmetic circuits over Boolean circuits in the bounded-depth regime (in all other regimes, arithmetic circuits are considered ``weaker''), as the majority function is well-known not to have Boolean $\AC^0$ circuits.

Indeed, while $\AC^0$ lower bounds in the Boolean setting has been known since the 1980's, it took over 30 years to obtain any analogous lower bound in the arithmetic one.
This was finally done by Limaye, Srinivasan, and Tavenas in 2021 \cite{LST21a}, later generalized to all fields by \textcite{Forbes24}.
They proved that any constant-depth circuit for the product of $n$ $2 \times 2$ matrices (a problem in $\NC^1_\F$) must have super-polynomial size!\footnote{This result is not made explicit in \cite{LST21a}, but it is an easy corollary of their depth hierarchy theorem together with the completeness of $2 \times 2$ iterated matrix product for $\NC^1_\F$ \cite{BC92, BIZ18}.}
As the determinant can efficiently simulate any formula (by Valiant's completeness result for the class $\VF$ \cite{Valiant79}), the same lower bound follows for the determinant.\footnote{Of course, this can be seen by a simple direct reduction as well.}
Moreover, following \cite{CKLMSS23}, we know that, in contrast to families like the elementary symmetric and power sum polynomials, which are in $\AC^0_\F$, some natural families of symmetric polynomials such as the Schur polynomials are as hard as the determinant, and thus cannot be in $\AC^0_\F$.

What about the GCD? 
Is it in $\AC^0_\F$? 
As hard as the determinant? 
Or somewhere in the middle?

\subsection{Our Results}

We give $\AC^0_\F$ algorithms (more precisely, constant depth and polynomial size formulae) for a host of problems from linear and polynomial algebra.
We note that in all cases these problems were not known to have polynomial size formulae, regardless of depth. 

While arithmetic circuits are a non-uniform model of computation, and in particular allow access to arbitrary constants in the underlying field, all our algorithms are uniform, and may alternatively be viewed as arithmetic PRAM algorithms with a polynomial number of processors running in constant parallel time.
This view may be advantageous, as PRAMs allow more basic operations, like branching on a zero-test; such a test must be added to the model of arithmetic circuits for discontinuous functions like GCD.

All our results hold over every field of characteristic 0, and every field of large enough\footnote{Only polynomial in the input size.} positive characteristic.
Input and output polynomials, which will typically be univariate,\footnote{In some cases they will have a constant number of variables.} are always described by their coefficients, and so their degrees will be implicitly counted in the input (and output) size.

We state our results informally here, and defer formal results to the technical sections, after formally defining our computational model, which is essentially arithmetic circuits over a field.
As is standard in the field (and many of the results in the previous section), when computing a discontinuous function like the GCD, a circuit can also branch by testing if a field element equals zero, which we refer to as piecewise computation in the theorem statements below.

Our first main theorem resolves the complexity of the GCD.

\begin{theorem}[see \cref{thm:gcd} and \cref{cor:lcm}]
	Given two polynomials $f,g \in \F[x]$, their GCD and LCM can be computed piecewise in $\AC^0_\F$.
\end{theorem}

We discuss some concrete and abstract extensions of this algorithm.
The concrete ones involve some polynomials and matrices related to the GCD, which we now define. 

Assume that the polynomials $f$ and $g$ are monic and of degrees $n$ and $m$, respectively, and factor completely over the algebraic closure $\overline{\F}$ of $\F$ as $f=\prod_i (x-\alpha_i)$ and $g=\prod_i (x-\beta_i)$, with possible repetitions in case of multiplicities.
The \emph{resultant} $\res(f,g) = \prod_{i,j} (\alpha_i - \beta_j) \in \F$ is nonzero if and only if $\gcd(f,g)=1$.
A well known special case of the resultant is the \emph{discriminant} of a polynomial, defined as $\disc(f) \coloneqq \res(f,f')$ where $f'$ is the first derivative of $f$, which is zero precisely when $f$ has a double root.
For quadratic polynomials, the discriminant takes the familiar form $\disc(ax^2 + bx + c) = b^2 - 4ac$.

When $\gcd(f,g)=1$, it is well-known that there are (unique) polynomials $a$ of degree $< m$ and $b$ of degree $< n$ such that $af+bg = 1$.
(A similar formula holds when the GCD has positive degree.)
The polynomials $a$ and $b$ are known as the \emph{B\'{e}zout coefficients} of $f$ and $g$. 
The equation $af+bg = 1$ is actually a linear system whose variables are the coefficients of the unknown polynomials $a$ and $b$, given by an $(n+m) \times (n+m)$ matrix called the \emph{Sylvester matrix} $\Syl(f,g)$ (which can be seen in \cref{def:sylvester matrix}) whose entries are the given coefficients of $f$ and $g$.
In fact, the resultant $\res(f,g)$ is precisely the determinant of the Sylvester matrix $\Syl(f,g)$.
A related matrix, important for rational interpolation and Pad\'{e} approximation of polynomials is the B\'{e}zout matrix of $f$ and $g$, denoted $\Bez(f,g)$. 

The natural problems associated with these objects have fast parallel algorithms as well.

\begin{theorem}[see \cref{thm:resultant ac0,cor:bezout general,thm:sylvester adjugate,thm:bezout inverse}]
	Given two polynomials $f,g \in \F[x]$, their resultant and B\'{e}zout coefficients, as well as the inverses of their Sylvester and B\'{e}zout matrices, can be computed in $\AC^0_\F$.
\end{theorem}

We note that the previously-mentioned problems of inverting triangular Toeplitz matrices and polynomial division with remainder are in fact corollaries of the theorem above through simple reductions.
Although these problems were already known to have $\AC^0_\F$ algorithms by \textcite{Bini84} and \textcite{BP85}, our algorithms are of a very different flavor!

For the more abstract extensions of our GCD algorithm, let us rewrite the factorization of the polynomials $f$ and $g$ slightly differently.
Let $\gamma_1, \ldots, \gamma_k$ be the union of their roots over the algebraic closure $\overline{\F}$ (namely, the union of the $\alpha_i$ and $\beta_i$ above).
We can write $f = \prod_{i=1}^k (x-\gamma_i)^{a_i}$ and $g = \prod_{i=1}^k (x-\gamma_i)^{b_i}$, where the $a_i$ and $b_i$ are the multiplicities of the root $\gamma_i$ in $f$ and $g$, respectively.
Clearly, their GCD is given by
\[
	\gcd(f,g) = \prod_{i=1}^k (x-\gamma_i)^{\min(a_i, b_i)}.
\]
Similarly, their LCM is given by 
\[
	\lcm(f,g) = \prod_{i=1}^k (x-\gamma_i)^{\max(a_i, b_i)}.
\]
Consider now the product of $f$ and $g$, which also has a trivial $\AC^0_\F$ algorithm.
In this notation we have
\[
	f \cdot g =  \prod_{i=1}^k (x-\gamma_i)^{a_i+b_i}.
\]

One is naturally led to consider what other functions of the exponents are computable in $\AC^0_\F$.
How about product of the exponents, namely
\[
	f \diamond g \coloneqq \prod_{i=1}^k (x-\gamma_i)^{a_i b_i}?
\]
It is not even clear at first sight that this is an algebraic operation computable by arithmetic circuits at all.
However, with the same techniques we use to compute the GCD, we can actually prove that it is.
Indeed, \emph{any} function is.

\begin{theorem}[see \cref{thm:arbitrary function}]
	Let $P: \naturals \times \naturals \to \naturals$ be any integer function.
	Given two polynomials $f,g \in \F[x]$ as above, the polynomial  
	\[
		f \diamond_P g \coloneqq \prod_{i=1}^k (x-\gamma_i)^{P(a_i,b_i)}
	\]
	can be computed piecewise in $\AC^0_\F$.
\end{theorem}

Note that the size of the circuit computing $f \diamond_P g$ is polynomial in the degrees of the input \emph{and} output of the problem, the latter being $\max \set{P(i,j) \, : \,  i,j \in [k]}$, which is of course necessary. 

The same theorem extends verbatim for any \emph{constant} number $c$ of polynomials, and for any $c$-ary integer function $P$.
As an example, for $c = 5$, we can take $P$ to return the median of its inputs and compute the corresponding polynomial piecewise in $\AC^0_\F$.
The intuitive meaning of some operations we can perform so efficiently is far from obvious, and it would be interesting to find a real application for some such function $P$.

Another natural extension is to consider an arbitrary number of input polynomials.
This is not obvious, even for the GCD function.
Indeed, for three polynomials $f$, $g$, and $h$, there is no analogue of the resultant, i.e., a single polynomial function of the coefficients of $f$, $g$, and $h$ which is nonzero if and only if $\gcd(f,g,h) = 1$.
Nevertheless, we can compute the GCD of an arbitrary number of polynomials, again piecewise in $\AC^0_\F$.

\begin{theorem}[see \cref{thm:multigcd,cor:multilcm}]
	Given any number of polynomials $f_1, f_2, \ldots , f_m \in \F[x]$, $\gcd(f_1, f_2, \ldots , f_m)$ and $\lcm(f_1, f_2, \ldots, f_m)$ can be computed piecewise in $\AC^0_\F$.
\end{theorem}

Again, one can note that the GCD and LCM functions perform $\min$ and $\max$ operations, respectively, on the exponents of each linear factor of the $f_j$.
Moreover, the product of $m$ polynomials performs addition on the exponents, and is also in $\AC^0_\F$.
So, we are lead to ask: what other functions of the exponents give rise to such fast parallel algorithms?
Here we also have a fairly general result.
Keeping with the notation above, let $\gamma_1, \ldots, \gamma_k \in \overline{\F}$ be the union of the roots of $m$ polynomials $f_1,\ldots,f_m \in \F[x]$ and write
\[
	f_i(x) = \prod_{j=1}^k (x - \gamma_j)^{a_{i,j}}.
\]
As in the case of two polynomials, we can apply \emph{any} function to the exponents of each factor.
This holds when the function $P : \naturals^m \to \naturals$ is given in the dense representation (so $P$ is specified by a list of roughly $d^m$ numbers, where $d$ bounds the degree of the $f_i$), and also when $P$ is described more succinctly as a sort of circuit over the integers (see \cref{def:tropical threshold circuit} for the precise definition).

\begin{theorem}[see \cref{thm:multi arbitrary function,thm:ckt multi arbitrary function}]
	Let $P : \naturals^m \to \naturals$ be any integer function.
	Given $m$ polynomials $f_1, \ldots, f_m \in \F[x]$ as above, the polynomial 
	\[
		\diamond_P(f_1,\ldots, f_m) \coloneqq \prod_{i=1}^k (x - \gamma_i)^{P(a_{1,i}, \ldots, a_{m,i})}
	\]
	can be computed piecewise in $\AC^0_\F$.
\end{theorem}

Our results also extend to the multivariate setting.
Here, the input polynomials can be themselves be given as arithmetic circuits.
Using standard tools to reduce questions about multivariate factorization (like the GCD) to their univariate counterparts, our earlier algorithm for the univariate GCD can be used to compute the GCD of multivariate polynomials.

\begin{theorem}[see \cref{thm:multivariate gcd and lcm}]
	Given multivariate polynomials $f_1,\ldots, f_m \in \F[\vec{x}]$, if $f_1, \ldots, f_m$ can be computed in $\AC^0_\F$, then $\gcd(f_1,\ldots,f_m)$ and $\lcm(f_1,\ldots,f_m)$ can be computed piecewise in $\AC^0_\F$.
\end{theorem}

\subsection{Our Techniques and Their Origins}

To summarize this section in a sentence, everything is about efficiently computing interesting and useful \emph{symmetric} functions over the roots of given polynomials using only their coefficients as input.
We will now give a taste of the main ones.

Fix an integer $n$.
Assume again that a degree $n$ polynomial $f = \sum_{j=0}^n a_j x^j \in \F[x]$ factors completely over $\overline{\F}$ as 
\[
	f=\prod_{i=1}^n (x-\alpha_i),
\]
where some of the $\alpha_i$ may repeat due to multiplicity.
We denote by $\vec{\alpha}$ the vector of the $\alpha_i$ in some arbitrary order (it will not matter which).
We stress that our algorithms have no access to these roots (indeed, the roots may not even be in $\F$), only to the coefficients $a_j \in \F$ of $f$.
However, it is well known that these coefficients are (up to sign) the \emph{elementary symmetric polynomials} of the roots.
More precisely,
\[
	a_j = (-1)^{n-j} e_{n-j}(\vec{\alpha}),
\]
where $e_d(\vec{z}) \coloneqq \sum_{S\subset [n],\,|S|=d} \prod_{j\in S} z_j$.

The famous Girard--Newton identities relate the elementary symmetric polynomials to another important family of symmetric polynomials, the \emph{power sum polynomials}, defined by $p_k(\vec{z}) \coloneqq \sum_{j=1}^n z_j^k$.
More precisely, for any integer $m$, the first $m$ power sums are polynomial expressions in the first $m$  elementary symmetric polynomials, and vice versa.\footnote{Note that while $e_m(\vec{z})=0$ for $m>n$, this is not the case for $p_m(z)$. Still, the statement above holds for every $m$.}
Thus, the coefficients of a polynomial give us access to new symmetric functions of its roots: the power sums.

There are several different ways to express these relations, each with their own uses.
For example, noticing that the relations above are \emph{triangular} was a key fact used in Csanky's original $\NC^2_\F$ algorithm for the determinant \cite{Csanky76}.
For us, an exponential form encompassing the relations between the generating functions $\sum_m e_m t^m$ and $\sum_m p_m t^m$ (see \cref{sec:newton series}) will be key.
Simple use of interpolation gives rise to our first tool (observed by many), namely that this conversion between the two families can be performed in $\AC^0_\F$.

\begin{theorem}[Folklore]
	For every $m$, there are $\AC^0_\F$ circuits that given the first $m$ elementary symmetric polynomials $e_j(\vec{z})$ as inputs, output the first $m$ power sums $p_j(\vec{z})$, and vice versa.
\end{theorem}

For example, this fact (in \emph{one} direction) has been used in \cite{SW01} to obtain smaller depth-4 and depth-6 formulas for the elementary symmetric polynomials than Ben-Or's depth-3 formulas mentioned above. 

But for us, a central inspiration came from the paper of \textcite{BFSS06}.
While they only discuss sequential algorithms, they use this conversion (in \emph{both} directions) to manipulate roots of given polynomials in very interesting ways that look similar to polynomials we want to compute, such as the resultant.
In particular, given polynomials $f$ and $g$ having roots $\alpha$'s and $\beta$'s as above, they compute the polynomials $f \oplus g \coloneqq \prod_{i,j} (x-(\alpha_i + \beta_j))$ and  $f \otimes g \coloneqq \prod_{i,j} (x-\alpha_i \beta_j)$.
Indeed, they can replace sum and product with any fixed bivariate polynomial! 

Observing that actually $(f(x) \oplus g(-x))|_{x=0} = \res(f,g)$, and inspecting their sequential algorithm to verify that it can be parallelized using the folklore theorem above, we already get an $\AC^0_\F$ circuit for the resultant! 
This indeed was our starting point.

To get our results, we will need to obtain more symmetric functions of the roots of a given polynomial $f$.
The \emph{fundamental theorem of symmetric polynomials} states that \emph{any} $n$-variate symmetric polynomial can be written as a fixed polynomial in the $n$ elementary symmetric ones (and hence, by the Girard--Newton identities, also in terms of the first $n$ power sums).
The whole question for us is which of these conversions can be performed in $\AC^0_\F$.

A first step is the simple observation that for any given polynomial $g$, we can compute in $\AC^0_\F$ the sum $\sum_{i=1}^n g(\alpha_i)$ over the roots of $f$, by using the power sums $p_k(\vec{\alpha})$.
It would actually be useful to do the same for the product $\prod_{i=1}^n g(\alpha_i)$, as this is yet another expression for the resultant $\res(f,g)$.
We prove a more general result, computing in $\AC^0_\F$ every elementary symmetric polynomial over the values $r(\alpha_i)$ where $r$ is any given rational function.

\begin{theorem}[see \cref{lem:esym over roots}]
	Given $f, g \in \F[x]$ and any integer $d$, we can compute $e_d(g(\alpha_1),\ldots,g(\alpha_n))$ in $\AC^0$, where $\alpha_1, \ldots, \alpha_n \in \overline{\F}$ are the roots of $f$ in the algebraic closure of $\F$.
	Moreover given another polynomial $h$ which has no roots in common with $f$ over the algebraic closure $\overline{\F}$, we can compute 
	\[
		e_d\del{\frac{g(\alpha_1)}{h(\alpha_1)}, \ldots, \frac{g(\alpha_n)}{h(\alpha_n)}}
	\]
	in $\AC^0_\F$.
\end{theorem}

Towards getting a handle over the multiplicities of roots of $f$, we first note that all derivatives $f^{(r)}(x)$ of $f$ can be easily computed in $\AC^0_\F$.
A root $\alpha$ has multiplicity at least $r$ if and only if $f$ and its first $r-1$ derivatives vanish at this root.
This is used to construct a polynomial $g$ (with a constant number of variables besides $x$), to which the theorem above can be applied (see details in \cref{sec:operations on roots}).
This allows us to \emph{filter} out the roots of $f$ with multiplicities above (or below) a given threshold $r$, and hence obtain those of multiplicity precisely $r$.
An important consequence is the ability to compute the \emph{squarefree decomposition} of $f$.
We conclude the techniques section with this consequence.

\begin{theorem}[see \cref{lem:squarefree decomp}]
	Given $f \in \F[x]$, we can compute piecewise in $\AC^0_\F$ the (unique) sequence of polynomials $f_1, f_2, \ldots, f_n \in \F[x]$ such that no $f_r$ has a double root, $\gcd(f_i,f_j) = 1$, and $f=\prod_{r=1}^n f_r^r$.
\end{theorem}

This ability to filter out the roots by multiplicity is key to most of our results, and the reader is invited to see how compute $\gcd(f,g)$ using it. 

\subsection{Organization}

The rest of this paper is organized as follows.
We start with preliminary material in \cref{sec:prelim}.
In \cref{sec:newton series}, we review folklore $\AC^0_\F$ implementations of Newton's identities, which will be an essential tool for all of our results.
\cref{sec:exact division} is a warm-up, where we show how Newton's identities can be used to solve some interesting toy problems in $\AC^0_\F$.
Our work starts in earnest in \cref{sec:symmetric functions of roots}, where we develop tools to evaluate symmetric functions of the roots of a given polynomial.
We then apply these tools in \cref{sec:sylvester and bezout}, where we compute the determinants and inverses of the Sylvester and B\'{e}zout matrices in $\AC^0_\F$.

\cref{sec:operations on roots} is the main technical section in this work, where we introduce filtering and thresholding, two techniques that let us design piecewise $\AC^0_\F$ algorithms that manipulate the factorization pattern of a polynomial without having explicit access to its roots.
\cref{sec:gcd and lcm} applies the results of \cref{sec:operations on roots} to compute the GCD and LCM of many polynomials piecewise in $\AC^0_\F$.
We generalize this result in \cref{sec:arbitrary functions} to arbitrary functions of root multiplicities.
In \cref{sec:multivariate}, we extend our univariate algorithms to the multivariate setting.
Finally, \cref{sec:conclusion} concludes with some open problems.

\section{Preliminaries} \label{sec:prelim}

To keep this work self-contained, this section includes a number of well-known results from arithmetic complexity and computer algebra.

\subsection{Notation}

We work over a field $\F$ of characteristic zero or of polynomially-large characteristic.
For example, to compute the GCD of two polynomials of degree $d$ over a field of positive characteristic, we require $\ch(\F) \ge 2d+1$.
The precise requirements on the characteristic of $\F$ will be specified in the statements of our results.

When we analyze our algorithms, it will often be convenient to work with the factorization of a polynomial $f(x)$ into linear factors $f(x) = \prod_i (x - \alpha_i)$.
In general, such a factorization only exists over the algebraic closure $\overline{\F}$ of the base field.
This factorization is only used in the analysis of our algorithms; in particular, we do not assume that the field $\F$ is algebraically closed.

We abbreviate a vector $(x_1,\ldots,x_n)$ as $\vec{x}$.
We denote by $\F[\vec{x}]$ the polynomial ring in the variables $x_1,\ldots,x_n$.
For a vector $\vec{a} \in \naturals^n$, we abbreviate the monomial $\prod_{i=1}^n x_i^{a_i}$ as $\vec{x}^{\vec{a}}$.
We let $\norm{\vec{a}}_1 \coloneqq \sum_{i=1}^n a_i$ denote the $\ell_1$ norm of $\vec{a}$.
We use $\F(\vec{x})$ and $\F \llb \vec{x} \rrb$ to denote the field of rational functions and ring of formal power series, respectively, in the variables $x_1,\ldots,x_n$.
Given a polynomial $f \in \F[x]$ and a natural number $r \in \naturals$, we write $f^{(r)}(x)$ for the $r$\ts{th} derivative of $f$.
For two polynomials $f, g \in \F[x]$, we write $f \mid g$ to denote that $f$ divides $g$.

Throughout, if the input to an algorithmic problem consists of univariate polynomials $f_1,\ldots,f_m \in \F[x]$, we assume that the polynomials $f_1,\ldots,f_m$ are monic, i.e., that the leading coefficient of each $f_i$ is 1.
This is done for the sake of notational convenience; all of our algorithms easily extend to handle non-monic inputs.
As our inputs are assumed to be monic, we likewise adopt the convention that the GCD and LCM are defined to be monic polynomials.
From here on, all univariate polynomials are assumed to be monic unless specified otherwise.

For a natural number $d \in \naturals$, we write
\begin{align*}
	p_d(\vec{x}) &\coloneqq \sum_{i=1}^n x_i^d \\
	e_d(\vec{x}) &\coloneqq \sum_{\substack{S \subseteq [n] \\ |S| = d}} \prod_{i \in S} x_i
\end{align*}
for the degree-$d$ power sum and elementary symmetric polynomials, respectively.
These two families of polynomials play an essential role throughout our work.
We adopt the conventions that $p_0(\vec{x}) = n$, $e_0(\vec{x}) = 1$, and $e_d(\vec{x}) = 0$ for $d > n$.

For a matrix $A \in \F^{n \times n}$, we denote by $\adj(A) \in \F^{n \times n}$ the \emph{adjugate} of $A$, defined as
\[
	\adj(A)_{i,j} \coloneqq (-1)^{i+j} \det(A_{-j, -i}),
\]
where $A_{-j, -i}$ is the submatrix of $A$ obtained by deleting the $j$\ts{th} row and $i$\ts{th} column.
The adjugate satisfies the identity $\adj(A) A = \det(A) I_n$, where $I_n$ is the $n \times n$ identity matrix.
In particular, when $A$ is invertible, the inverse of $A$ is given by $\frac{1}{\det(A)} \adj(A)$.

\subsection{Arithmetic Circuits}

\subsubsection{Circuits and Complexity Classes}

We use arithmetic circuits as our basic model of computation.

\begin{definition} \label{def:arith ckt}
	Let $\F$ be a field and let $\F(\vec{x})$ be the field of rational functions in the variables $x_1,\ldots,x_n$.
	An \emph{arithmetic circuit over $\F$} is a directed acyclic graph.
	Vertices of in-degree zero are called \emph{input gates} and are each labeled by a variable $x_i$ or a field element $\alpha \in \F$.
	Vertices of positive in-degree are called \emph{internal gates} and are labeled by an element of $\set{+, \times, \div}$.
	Vertices of out-degree zero are called \emph{output gates}.
	Each gate of the circuit computes a rational function in $\F(\vec{x})$ in the natural way.
	We require that no division by the identically zero polynomial takes place in the circuit.
	If $\set{f_1,\ldots,f_m}$ are the functions computed by the output gates of the circuit, we say that the circuit computes $\set{f_1,\ldots,f_m}$.
	The \emph{size} of the circuit is the number of wires in the circuit.
	The \emph{depth} of the circuit is the length of the longest path from an input gate to an output gate.
\end{definition}

As is standard in arithmetic complexity, when we say that an arithmetic circuit computes a rational function $f \in \F(\vec{x})$, we require the circuit to formally compute the same rational function as $f$.
This distinction is important over finite fields, where two unequal polynomials can induce the same function, as in the case of $x^2 - x$ and $0$ over the finite field $\F_2$.

Naturally, one can define complexity classes of (families of) rational functions in terms of their arithmetic circuit complexity.
The following classes capture efficient low-depth computation.
Needless to say, these are separate classes for every field $\F$.

\begin{definition} \label{def:circuit classes}
	Let $\F$ be a field and let $f = (f_1,f_2,\ldots)$ be a family of rational functions.
	We say that $f \in \NC^i_\F$ if $f_n$ can be computed by an arithmetic circuit of fan-in two, size $n^{O(1)}$, and depth $O(\log^i n)$.
	We say that $f \in \AC^i_\F$ if $f_n$ can be computed by an arithmetic circuit of \emph{unbounded} fan-in, size $n^{O(1)}$, and depth $O(\log^i n)$.
\end{definition}

In arithmetic circuit complexity, it is standard to restrict attention to families of polynomials $(f_1,f_2,\ldots)$ with the additional restriction that $\deg(f_n) \le n^{O(1)}$.
Under this restriction, the complexity classes corresponding to $\AC^i_\F$ and $\NC^i_\F$ are denoted by $\VAC^i$ and $\VNC^i$, respectively.
A well-known result of \textcite{VSBR83} shows that any circuit of size $n^{O(1)}$ can be converted to one of size $n^{O(1)}$ and depth $O(\log^2 n)$.
A careful reading of their proof shows that depth $O(\log n)$ suffices if unbounded fan-in is allowed.
Thus, when we restrict attention to polynomials of degree $n^{O(1)}$, we have the collapse
\[
	\VAC^0 \subsetneq \VNC^1 \subseteq \VAC^1 = \VNC^2 = \VAC^2 = \VNC^3 = \cdots.
\]
Further inspection of the proof of \cite{VSBR83} shows that depth $O(\log n)$ can be attained using addition gates of unbounded fan-in and multiplication gates of fan-in two, corresponding to the class $\SAC^1$ (for \emph{semi-unbounded} $\AC^1$).
The strict inclusion $\VAC^0 \subsetneq \VNC^1$ is a straightforward corollary of the depth hierarchy theorem of \textcite{LST21a}.

We also mention the Boolean complexity class $\DET$ \cite{Cook85}, which consists of all problems that are logspace-reducible to the determinant of an integer matrix.
This includes many familiar linear-algebraic problems, including matrix powering, matrix inverse, and computing the characteristic polynomial of a matrix.
In the Boolean setting, it is known that $\NL \subseteq \DET \subseteq \NC^2$.
It is conjectured that $\DET \not \subseteq \NC^1$.
If $\DET \subseteq \NC^1$ were true, then we would have the chain of inclusions
\[
	\NL \subseteq \DET \subseteq \NC^1 \subseteq \L,
\]
implying the unlikely collapse $\L = \NC^1 = \NL$.

The algebraic analogue of $\DET$, often denoted $\VBP$, is the class of polynomial families $(f_1,f_2,\ldots)$ computable by arithmetic branching programs of size $n^{O(1)}$.
We will not give a precise definition of arithmetic branching programs here, as we will not make use of them.
As in the Boolean setting, it is conjectured that the determinant and matrix inverse are \emph{not} computable by arithmetic circuits of fan-in two, size $n^{O(1)}$, and depth $O(\log n)$.
The recent work of \textcite{LST21a} shows unconditionally that the determinant and related problems are not computable by arithmetic circuits of size $n^{O(1)}$ and depth $O(1)$, i.e., that the determinant is not computable in $\AC^0_\F$.

\subsubsection{Piecewise Arithmetic Circuits}

By definition, arithmetic circuits can only represent rational functions.
However, natural functions of interest, like the GCD, are not rational.
For example, it is easy to see that
\[
	\gcd(x - \alpha, x - \beta) = 
	\begin{cases}
		x - \alpha & \text{if $\alpha = \beta$,} \\
		1 & \text{otherwise.}
	\end{cases}
\]
This is not a continuous function of $\alpha$ and $\beta$, so we cannot hope to compute the GCD using an arithmetic circuit.

To compute the GCD, we have to add a branching instruction to our model of computation.
In arithmetic complexity, this is typically done by allowing an algorithm to test if a computed quantity equals zero and branch accordingly.
This sort of operation is standard: even the Euclidean algorithm uses zero-testing to detect the degree of a remainder.
We formalize this via arithmetic circuits that define a function piecewise. 
This is a natural and simple extension of arithmetic circuits that allows them to compute functions like the GCD.
Although we provide a precise definition here for completeness, we encourage the reader to keep in mind standard arithmetic circuits as the model of computation.

We now define piecewise arithmetic circuits.
For our purposes, it will be enough to include an additional gate type that implements the selection function $\sel : \F^m \to \F$ defined by
\[
	\sel(y_1,\ldots,y_m) \coloneqq \begin{cases}
		y_1 & \text{if $y_1 \neq 0$,} \\
		y_2 & \text{if $y_1 = 0$ and $y_2 \neq 0$,} \\
		& \vdots \\
		y_m & \text{if $y_1 = \cdots = y_{m-1} = 0$ and $y_m \neq 0$,} \\
		0 & \text{otherwise.}
	\end{cases}
\]
That is, $\sel(y_1,\ldots,y_m)$ outputs the first nonzero element of the vector $(y_1,\ldots,y_m)$, falling back to zero otherwise.

\begin{definition} \label{def:piecewise circuits}
	Let $\F$ be a field.
	A \emph{piecewise arithmetic circuit over $\F$} is a directed acyclic graph.
	Vertices of in-degree zero are called \emph{input gates} and are each labeled by a variable $x_i$ or a field element $\alpha \in \F$.
	Vertices of positive in-degree are called \emph{internal gates} and are labeled by an element of $\set{+, \times, \div, \sel}$.
	Vertices of out-degree zero are called \emph{output gates}.
	Each gate of the circuit computes a piecewise rational function $\F^n \to \F$ in the natural way.
	We require that no division by zero takes place on any input to the circuit.
	If $\set{f_1, \ldots, f_m}$ are the functions computed by the output gates of the circuit, we say that the circuit computes $\set{f_1,\ldots,f_m}$ piecewise.
	The \emph{size} of the circuit is the number of wires in the circuit.
	The \emph{depth} of the circuit is the length of the longest path from an input gate to an output gate.
\end{definition}

Of course, one can generalize \cref{def:piecewise circuits} to allow for more complex branching logic.
The arithmetic networks of \textcite{vonzurGathen86survey} are one such generalization, where one can test a value for equality with zero, perform boolean computation on the results of these tests, and branch accordingly.
These zero tests and branching operations can clearly simulate the $\sel$ function we included in \cref{def:piecewise circuits}, but that level of generality will not be necessary for our results.

We note that simple modifications of the parallel algorithms of \textcite{BGH82,vonzurGathen84} for the GCD and Extended Euclidean algorithm, respectively, can be implemented in our model of piecewise arithmetic computation.
We are not aware of a natural problem that requires branching beyond what $\sel$ provides, but factorization of univariate polynomials over finite fields may be such an example \cite{vonzurGathen84}.

Continuing the example of the GCD of two linear polynomials, we can rewrite $\gcd(x-\alpha, x-\beta)$ as
\[
	\gcd(x-\alpha,x-\beta) = \sel((x - \beta) - (x - \alpha), x - \alpha) = \sel(\alpha - \beta, x - \alpha),
\]
which matches the form of \cref{def:piecewise circuits}.
Functions computed piecewise by low-depth arithmetic circuits can be evaluated quickly in parallel, as the relevant branching logic can be implemented efficiently in parallel.

We now define complexity classes that correspond to functions computed piecewise by small, low-depth arithmetic circuits.

\begin{definition} \label{def:piecewise circuit classes}
	Let $f = (f_1,f_2,\ldots)$ be a family of piecewise rational functions.
	We say that $f \in \NC^i_\F$ if $f_n$ can be computed piecewise by arithmetic circuits of fan-in two, size $n^{O(1)}$, and depth $O(\log^i n)$.
	We say that $f \in \AC^i_\F$ if $f_n$ can be computed piecewise by arithmetic circuits of \emph{unbounded} fan-in, size $n^{O(1)}$, and depth $O(\log^i n)$.
\end{definition}

To justify the use of the same notation for the classes defined in \cref{def:circuit classes,def:piecewise circuit classes}, the following lemma shows that piecewise arithmetic circuits provide no advantage over standard arithmetic circuits when computing a family of rational functions.
The proof makes use of basic properties of the Zariski topology over an algebraically closed field.
In particular, we use the fact that a finite intersection of nonempty open sets is nonempty and open, and that rational functions are determined identically by their evaluations on a nonempty open set \cite[Chapter 1, Section 3.2]{Shafarevich1}.

\begin{lemma} \label{lem:remove select}
	Let $\F$ be a field and $f \in \F(\vec{x})$ be a rational function.
	Suppose there is a piecewise arithmetic circuit of size $s$, depth $\Delta$, and fan-in $c$ that computes $f$.
	Then $f$ can be computed by an arithmetic circuit of the same size, depth, and fan-in.
\end{lemma}

\begin{proof}
	Because a piecewise arithmetic circuit correctly evaluates $f$ on any field extension $\mathbb{K} \supseteq \F$, we may replace $\F$ with its algebraic closure without any loss of generality.
	Let $\Phi$ be a piecewise arithmetic circuit of size $s$, depth $\Delta$, and fan-in $c$ that computes $f$.
	Our goal will be to remove the selection gates of $\Phi$ in a careful way, resulting in an arithmetic circuit that computes a rational function $g$ which agrees with $f$ on a nonempty Zariski-open subset $U \subseteq \F^n$.
	Because $f$ and $g$ agree on a nonempty Zariski-open subset, it must be the case that $f$ and $g$ are identically equal, i.e., that $f = g$ as elements of $\F(\vec{x})$.

	To remove the selection gates, let $v_1, \ldots, v_t$ be a topological ordering of the gates of $\Phi$ that are labeled by $\sel$.
	Let $w_1, \ldots, w_t$ be the children of $v_1$.
	Because $v_1$ precedes all other selection gates in topological order, its children $w_1, \ldots, w_t$ compute rational functions $h_1, \ldots, h_t \in \F(\vec{x})$.
	By removing children that compute the identically zero function, we may assume without loss of generality that $h_1(\vec{x}) \neq 0$.
	Let $U_1 \subseteq \F^n$ be the set of points $\vec{\alpha} \in \F^n$ at which $h_1(\vec{\alpha})$ is defined and is nonzero.
	Because $h_1$ is not identically zero, the set $U_1$ is a nonempty Zariski-open subset of $\F^n$.
	Replacing the gate $v_1$ with $w_1$ results in an arithmetic circuit with $t-1$ selection gates that agrees with $\Phi$ on $U_1$.

	Iterating this argument, we obtain an arithmetic circuit $\Psi$ with no selection gates and a sequence $U_1, \ldots, U_t \subseteq \F^n$ of nonempty Zariski-open sets such that $\Psi$ agrees with $\Phi$ on $U \coloneqq \bigcap_{i=1}^t U_i$.
	Because $U$ is a finite intersection of nonempty Zariski-open sets, $U$ itself is nonempty and Zariski-open.
	Since the rational function computed by $\Psi$ agrees with $f$ on a nonempty Zariski-open set, it follows that $\Psi$ computes $f$ correctly on all inputs for which $\Psi$ is defined.
	As $\Psi$ was constructed from $\Phi$ only by replacing gates, it is clear that the size, depth, and fan-in of the circuit did not increase during this process.
\end{proof}

\subsection{Known $\AC^0_\F$ Algorithms}

In this subsection, we collect previously-known algorithmic results that can be implemented in $\AC^0_\F$.
We start with basic operations on univariate polynomials: addition, multiplication, and derivatives.
Note that the inputs to these problems are the coefficients of two univariate polynomials $f, g \in \F[x]$ and the outputs are polynomial functions of the coefficients of $f$ and $g$.

\begin{lemma}
	Let $f, g \in \F[x]$ be univariate polynomials given by their coefficients.
	Then the coefficients of $f(x) + g(x)$, $f(x) \cdot g(x)$, and $f^{(r)}(x)$ can all be computed in $\AC^0_\F$.
\end{lemma}

Our next tool is polynomial interpolation, which we use extensively.
Let $f \in \F[\vec{x}, y]$ be a polynomial of degree $d$.
We can write $f$ as a polynomial in $y$ whose coefficients are polynomials in $\vec{x}$, i.e., there are polynomials $f_0,\ldots,f_d \in \F[\vec{x}]$ such that
\[
	f(\vec{x}, y) = \sum_{i=0}^d f_i(\vec{x}) y^i.
\]
Given evaluations $f(\vec{x}, \alpha_1), \ldots, f(\vec{x}, \alpha_{d+1})$ at $d+1$ distinct values for $y$, each of the $f_i$ can be expressed as a linear combination of these evaluations.
In particular, if $f$ can be computed by a circuit of size $s$ and depth $\Delta$, then the coefficients $f_0,\ldots,f_d$ can each be computed by a circuit of size $O(s d)$ and depth $\Delta + 1$.
We record this observation in the following lemma.

\begin{lemma} \label{lem:polynomial interpolation}
	Let $\F$ be a field.
	Let $f \in \F[\vec{x}, y]$ be a polynomial of degree $d$.
	Let $f_0, \ldots, f_d \in \F[\vec{x}]$ be polynomials such that
	\[
		f(\vec{x}, y) = \sum_{i=0}^d f_i(\vec{x}) y^i.
	\]
	Suppose that $f$ can be computed by an arithmetic circuit of size $s$ and depth $\Delta$.
	Then for each $i \in \set{0,1,\ldots,d}$, there is a circuit of size $O(s d)$ and depth $\Delta + 1$ that computes $f_i$.
	If $|\F| \le d$, then the circuit computing $f_i$ is defined over an extension $\mathbb{K} \supseteq \F$ such that $|\mathbb{K}| \ge d+1$.
\end{lemma}

An immediate consequence of \cref{lem:polynomial interpolation} is that if a polynomial $f(\vec{x}, y)$ can be computed by a small, low-depth circuit, then there is a circuit of comparable complexity that piecewise computes the \emph{leading coefficient} of $f$.
This follows from the fact that the leading coefficient of $f$ (with respect to $y$) is precisely $\sel(f_d(\vec{x}), f_{d-1}(\vec{x}), \ldots, f_0(\vec{x}))$.

\begin{lemma} \label{lem:leading coefficient}
	Let $\F$ be a field.
	Let $f \in \F[\vec{x}, y]$ be a polynomial of degree at most $d$.
	Let $f_0, \ldots, f_d \in \F[\vec{x}]$ be polynomials such that
	\[
		f(\vec{x}, y) = \sum_{i=0}^d f_i(\vec{x}) y^i.
	\]
	Suppose that $f$ can be computed by an arithmetic circuit of size $s$ and depth $\Delta$.
	Then there is a circuit of size $O(sd + d^2)$ and depth $\Delta + 2$ that piecewise computes the leading coefficient of $f$.
	If $|\F| \le d$, then the circuit computing the leading coefficient is defined over an extension $\mathbb{K} \supseteq \F$ such that $|\mathbb{K}| \ge d+1$.
\end{lemma}

A surprising application of interpolation, discovered by Ben-Or, shows that the elementary symmetric polynomials $e_d(x_1,\ldots,x_n)$ can be computed by arithmetic circuits of size $O(n^2)$ and depth 3.
This is done by applying interpolation to the polynomial
\[
	\prod_{i=1}^n (1 + y x_i) = \sum_{i=0}^n e_d(\vec{x}) y^i,
\]
which clearly can be computed by a circuit of size $O(n)$ depth 2.

\begin{theorem}[{[Ben-Or]}]
	The elementary symmetric polynomials $e_d(x_1,\ldots,x_n)$ can be computed in $\AC^0_\F$.
\end{theorem}

The next algorithms we quote perform linear algebra with structured matrices.
They take as input Toeplitz matrices (or equivalently, Hankel matrices), which we now define.

\begin{definition}
	Let $A \in \F^{n \times n}$ be an $n \times n$ matrix.
	We say that $A$ is a \emph{Toeplitz} matrix if there are field elements $\alpha_{-n+1},\ldots,\alpha_{n-1} \in \F$ such that $A_{i,j} = \alpha_{i-j}$.
	We say that $A$ is a \emph{Hankel} matrix if there are field elements $\alpha_1,\ldots,\alpha_{2n-1} \in \F$ such that $A_{i,j} = \alpha_{i+j-1}$.
\end{definition}

As an example, the matrices
\[
	T = \begin{pmatrix}
		x_3 & x_4 & x_5 \\
		x_2 & x_3 & x_4 \\
		x_1 & x_2 & x_3
	\end{pmatrix}
	\qquad 
	H = \begin{pmatrix}
		x_1 & x_2 & x_3 \\
		x_2 & x_3 & x_4 \\
		x_3 & x_4 & x_5
	\end{pmatrix}
\]
are $3 \times 3$ Toeplitz and Hankel matrices, respectively.
It is easy to see that by reversing the order of the rows, a Toeplitz matrix becomes Hankel and vice-versa.
A beautiful algorithm of \textcite{Bini84} shows that triangular Toeplitz matrices can be inverted in $\AC^0_\F$.
(The result in \cite{Bini84} is stated as an $\NC^1_\F$ algorithm, but it is clear that $\AC^0_\F$ suffices.)

\begin{theorem}[\cite{Bini84}]
	Let $X \in \F^{n \times n}$ be a triangular Toeplitz matrix.
	Then the inverse $X^{-1}$ can be computed in $\AC^0_\F$.
\end{theorem}

As an application of this algorithm, \textcite{BP85} showed that polynomial division with remainder can be performed in $\AC^0_\F$.
(Once again, the result in \textcite{BP85} is stated as an $\NC^1_\F$ algorithm, but $\AC^0_\F$ suffices.)

\begin{theorem}[\cite{BP85}]
	Let $f, g \in \F[x]$ be univariate polynomials given by their coefficients.
	Let $q, r \in \F[x]$ be the unique polynomials such that $f = q g + r$ and $\deg(r) < \deg(g)$.
	Then the coefficients of $q$ and $r$ can be computed in $\AC^0_\F$.
\end{theorem}

Later in \cref{sec:sylvester and bezout}, we will also describe $\AC^0_\F$ algorithms for inverting triangular Toeplitz matrices and for polynomial division with remainder.
Our algorithms employ a different approach than that used by \textcite{Bini84} and \textcite{BP85}.

The last algorithm we quote is Strassen's theorem on division elimination in arithmetic circuits.
Although this result applies to general arithmetic circuits, we cite it in a form that is useful for the low-depth regime.
For completeness, we provide a proof that Strassen's result preserves the depth of the circuit.

\begin{theorem}[\cite{Strassen73b}] \label{thm:ac0 division elimination}
	Let $\F$ be an infinite field and let $f(\vec{x}) \in \F[\vec{x}]$ be a multivariate polynomial of degree $n^{O(1)}$.
	Suppose $f(\vec{x})$ can be computed in $\AC^0_\F$.
	Then there is a division-free $\AC^0_\F$ circuit that computes $f(\vec{x})$.
\end{theorem}

\begin{proof}
	We first push the division gate to the top of the circuit.
	We do this by splitting each gate $v$ into two gates $(v, \text{num})$ and $(v, \text{den})$, with the intended meaning that $(v, \text{num})$ and $(v, \text{den})$ compute polynomials $g, h \in \F[\vec{x}]$ such that the gate $v$ computes the rational function $g/h$.
	This is straightforward to do for all input gates of the circuit.
	If $v = \prod_{i=1}^m u_i$ is a product gate with children $u_1, \ldots, u_m$, then we set
	\begin{align*}
		(v, \text{num}) &= \prod_{i=1}^m (u_i, \text{num}) \\
		(v, \text{den}) &= \prod_{i=1}^m (u_i, \text{den}).
	\end{align*}
	If $v = \sum_{i=1}^m u_i$ is a sum gate, then we instead set
	\begin{align*}
		(v, \text{num}) &= \sum_{i=1}^m (u_i, \text{num}) \prod_{j \neq i} (u_j, \text{den}) \\
		(v, \text{den}) &= \prod_{i=1}^m (u_i, \text{den}).
	\end{align*}
	Most of this rewiring can be done in a straightforward manner that only increases the size of the circuit by a polynomial factor and the depth by a constant factor.
	The only nontrivial case is to handle $(v, \text{num})$ when $v = \sum_{i=1}^m u_i$ is an addition gate.
	In this case, observe that the desired value of $(v, \text{num})$ is given by the coefficient of $t$ in the polynomial
	\[
		\prod_{i=1}^m \del{t \cdot (u_i, \text{num}) + (u_i, \text{den})},
	\]
	which can be computed by a circuit of size $m^{O(1)}$ and constant depth using \cref{lem:polynomial interpolation}.

	Now we have a division-free $\AC^0_\F$ circuit that computes two polynomials $g, h \in \F[\vec{x}]$ such that $f = g/h$ and $h$ is nonzero.
	We now eliminate this final division gate.
	By translating the variables and rescaling $h$ if necessary, we may assume that $h(\vec{0}) = 1$.
	Then $1/h(\vec{x})$ admits the power series expansion
	\[
		\frac{1}{h(\vec{x})} = \frac{1}{1 - (1 - h(\vec{x}))} = \sum_{i=0}^\infty \del{1 - h(\vec{x})}^i,
	\]
	so we have the equality of power series
	\[
		f(\vec{x}) = g(\vec{x}) \sum_{i=0}^\infty (1 - h(\vec{x}))^i.
	\]
	Because $1 - h(\vec{x})$ is a polynomial with no constant term, the polynomial $(1 - h(\vec{x}))^i$ contains only monomials of degree $i$ or higher.
	Letting $d \coloneqq \deg(f)$, truncating the above power series to $d$ terms yields the equality
	\[
		g(\vec{x}) \sum_{i=0}^d (1 - h(\vec{x}))^i = f(\vec{x}) + q(\vec{x})
	\]
	for some polynomial $q(\vec{x})$ that only contains monomials of degree $d+1$ or higher.
	This implies that $f(\vec{x})$ corresponds to the homogeneous components of $g(\vec{x}) \sum_{i=0}^d (1 - h(\vec{x}))$ of degree up to and including $d$. 

	To compute $f(\vec{x})$, we first construct a circuit that computes $g(\vec{x}) \sum_{i=0}^d (1 - h(\vec{x}))^i$, which can be done by adding $O(d)$ gates and constant depth to the circuit that computes $g$ and $h$.
	The homogeneous components of this polynomial can then computed by interpolating the coefficients of $t$ in the polynomial $g(t x_1, \ldots, t x_n) \sum_{i=0}^d (1 - h(t x_1, \ldots, t x_n))^i$, which can be done in $\AC^0_\F$ using \cref{lem:polynomial interpolation}.
	Here, we implicitly use the fact that the circuit computing $g$ and $h$ has size $s^{O(1)}$ and depth $O(1)$, where $s$ is the size of the circuit with divisions that computes $f$.
	This implies that $g$ and $h$ have degree bounded by $s^{O(1)}$, so the polynomial $g(\vec{x}) \sum_{i=0}^d (1 - h(\vec{x}))^i$ is of degree at most $d s^{O(1)}$, so the interpolation of \cref{lem:polynomial interpolation} can be performed efficiently.
\end{proof}

\begin{remark} \label{remark: division elimination}
	Although the statement of \cref{thm:ac0 division elimination} assumes the underlying field is infinite, divisions can likewise be eliminated over sufficiently-large finite fields.
	In the proof of \cref{thm:ac0 division elimination}, after pushing the division to the top of the circuit, we have a division-free $\AC^0_\F$ circuit that computes polynomials $g, h \in \F[\vec{x}]$ such that $f = g / h$ and $h$ is nonzero.
	To eliminate this final division, it was enough to find a point $\vec{\alpha} \in \F^n$ where $h(\vec{\alpha}) \neq 0$.
	Such a point always exists when $\F$ is infinite or a sufficiently-large finite field.
	Over small finite fields, it may not be possible to find a point where $h$ is nonzero.
	However, when we later apply \cref{thm:ac0 division elimination}, we will know precisely what the polynomial $h$ is.
	We will use this information to explicitly exhibit a point where $h$ is nonzero, permitting us to apply \cref{thm:ac0 division elimination} even when working over a small finite field.
\end{remark}

\subsection{The Euclidean Algorithm and the Resultant}

This subsection recalls basic definitions and results related to greatest common divisors and the Euclidean algorithm.
A more thorough treatment of this material can be found in the delightful book of \textcite{vzGG13}.
We start with the definition of the B\'{e}zout coefficients of a pair of polynomials.

\begin{definition} \label{def:bezout coefficients}
	Let $f(x) = \sum_{i=0}^n f_i x^i$ and $g(x) = \sum_{i=0}^m g_i x^i$ be univariate polynomials of degrees $n$ and $m$, respectively.
	Let $d \coloneqq \deg(\gcd(f,g))$.
	The \emph{B\'{e}zout coefficients} of $f$ and $g$ are the unique polynomials $a, b \in \F[x]$ that satisfy
	\begin{enumerate}
		\item
			$\deg(a) < \deg(g) - d$,
		\item
			$\deg(b) < \deg(f) - d$, and
		\item
			$a(x)f(x) + b(x) g(x) = \gcd(f,g)$. \qedhere
	\end{enumerate}
\end{definition}

Note that B\'{e}zout coefficients are unchanged by multiplying or dividing $f$ and $g$ by a common factor.
Next, we recall the definition of the Sylvester matrix of a pair of polynomials.

\begin{definition} \label{def:sylvester matrix}
	Let $f(x) = \sum_{i=0}^n f_i x^i$ and $g(x) = \sum_{i=0}^m g_i x^i$ be univariate polynomials of degrees $n$ and $m$, respectively.
	The \emph{Sylvester matrix} of $f$ and $g$ is the $(n + m) \times (n + m)$ matrix given by
	\[
		\Syl(f,g) \coloneqq \begin{pmatrix}
			f_n & & & & g_m & & & & & \\
			f_{n-1} & f_n & & & g_{m-1} & g_m & & & & \\
			\vdots & \vdots & \ddots & & \vdots & \vdots & \ddots & & & \\
			\vdots & \vdots & & f_n & g_1 & \vdots & & \ddots & & \\
			\vdots & \vdots & & f_{n-1} & g_0 & \vdots & & & \ddots & \\
			\vdots & \vdots & & \vdots & & g_0 & & & & g_m \\
			f_0 & \vdots & & \vdots & & & \ddots & & & \vdots \\
			& f_0 & & \vdots & & & & \ddots & & \vdots \\
			& & \ddots & \vdots & & & & & \ddots & \vdots \\
			& & & f_0 & & & & & & g_0
		\end{pmatrix}.
	\]
	The \emph{resultant} of $f$ and $g$ is $\res(f,g) \coloneqq \det \Syl(f,g)$.
\end{definition}

Let $f, g \in \F[x]$ be polynomials of degrees $n$ and $m$, respectively.
For a natural number $d \in \naturals$, denote by $\F[x]_{< d}$ the space of univariate polynomials of degree less than $d$.
The Sylvester matrix $\Syl(f,g)$ corresponds to the linear map
\begin{align*}
	\F[x]_{<m} \times \F[x]_{<n} & \to \F[x]_{<n+m} \\
	(a, b) & \mapsto af+bg
\end{align*}
written in the monomial basis $\set{1, x, x^2, \ldots}$.
One can show that this map is an isomorphism of vector spaces if and only if $\gcd(f,g) = 1$ (see, e.g., \cite[Section 6.3]{vzGG13}).
In particular, the resultant $\res(f,g) = \det \Syl(f,g)$ characterizes when two polynomials share a common factor.

\begin{lemma}[{see, e.g., \cite[Corollary 6.17]{vzGG13}}] \label{lem:resultant properties}
	Let $f, g \in \F[x]$.
	Then $\gcd(f, g) = 1$ if and only if $\res(f, g) \neq 0$.
	Moreover, $\gcd(f,g) = 1$ if and only if $\Syl(f,g)$ is invertible.
\end{lemma}

From the description of the Sylvester matrix as the linear map $(a, b) \mapsto af + bg$, it is easy to see that the inverse of $\Syl(f,g)$ allows us to recover the B\'{e}zout coefficients of $f$ and $g$ when $\gcd(f,g) = 1$.
More generally, for any $\ell \in [n+m]$, the entries of the $\ell$\ts{th} column of $\Syl(f,g)^{-1}$ correspond to the coefficients of polynomials $a_\ell, b_\ell \in \F[x]$ such that $a_\ell f + b_\ell g = x^{n+m-\ell}$.
We record this observation in the following lemma.

\begin{lemma} \label{lem:sylvester inverse}
	Let $f, g \in \F[x]$ be univariate polynomials of degrees $n$ and $m$, respectively.
	Assume that $\gcd(f,g) = 1$, so $\Syl(f,g)^{-1}$ is invertible.
	For $\ell \in [n+m]$, the $\ell$\ts{th} column of $\Syl(f,g)^{-1}$ corresponds to the coefficients of polynomials $a_\ell, b_\ell \in \F[x]$ such that
	\begin{enumerate}
		\item
			$\deg(a_\ell) < \deg(g)$, 
		\item
			$\deg(b_\ell) < \deg(f)$, and 
		\item
			$a_\ell f + b_\ell g = x^{n+m-\ell}$.
	\end{enumerate}
\end{lemma}

If the polynomials $f$ and $g$ happen to split into linear factors (as they do over an algebraically closed field), then we can write the resultant $\res(f,g)$ as a simple function of the roots of $f$ and $g$.
Although we will not directly compute the roots of $f$ or $g$, the following identity will be crucial in designing an $\AC^0$ circuit to compute the resultant.

\begin{lemma}[{see, e.g., \cite[Section 3.1]{CLO05}}] \label{lem:resultant as product}
	Let $f(x) = \prod_{i=1}^n (x - \alpha_i)$ and $g(x) = \prod_{i=1}^m (x - \beta_i)$ be univariate polynomials.
	Then the resultant $\res(f,g)$ is given by 
	\[
		\res(f,g) = \prod_{i=1}^n \prod_{j=1}^m (\alpha_i - \beta_j) = \prod_{i=1}^n g(\alpha_i) = (-1)^{nm} \prod_{i=1}^m f(\beta_i). \qedhere
	\]
\end{lemma}

The resultant also allows us to detect when a polynomial has a double root.
Recall that a \emph{double root} of a polynomial $f \in \F[x]$ is a point $\alpha \in \F$ such that $f(\alpha) = f'(\alpha) = 0$.
That is, a double root corresponds to a shared root of $f$ and its derivative.
This leads to the discriminant of a polynomial, a specialization of the resultant.

\begin{definition}
	Let $f(x) \in \F[x]$ be a monic univariate polynomial.
	The \emph{discriminant} of $f$, denoted $\disc(f)$, is defined as $\disc(f) \coloneqq (-1)^{\binom{n}{2}} \res(f,f')$.
\end{definition}

For quadratic polynomials, the discriminant is given by the familiar formula
\[
	\disc(a x^2 + b x + c) = b^2 - 4 a c.
\]
As remarked above, a double root is precisely a common root of $f$ and its derivative $f'$.
The discriminant detects when a polynomial has a double root.

\begin{lemma}
	Let $f \in \F[x]$.
	Then $f$ has a double root if and only if $\disc(f) = 0$.
\end{lemma}

Just as we can express the resultant $\res(f,g)$ as a simple function of the roots of $f$ and $g$, we can likewise express the discriminant $\disc(f)$ as a simple function of the roots of $f$.

\begin{lemma} \label{lem:discriminant as product}
	Let $f(x) = \prod_{i=1}^n (x - \alpha_i)$ be a univariate polynomial.
	Then
	\[
		\disc(f) = (-1)^{\binom{n}{2}} \prod_{\substack{i, j \in [n] \\ i \neq j}} (\alpha_i - \alpha_j) = \prod_{i < j} (\alpha_i - \alpha_j)^2.
	\]
\end{lemma}

The resultant was originally found by \textcite{Bezout1764} as the determinant of what is now called the B\'{e}zout matrix of two polynomials, which we define below.

\begin{definition} \label{def:bezout matrix}
	Let $f, g \in \F[x]$ be univariate polynomials of degree at most $n$.
	The \emph{B\'{e}zout matrix of order $n$} associated with the polynomials $f$ and $g$, denoted by $\Bez_n(f, g)$, is the $n \times n$ matrix that satisfies the identity
	\[
		\frac{f(x) g(y) - f(y) g(x)}{x - y} = \sum_{i, j = 0}^{n-1} \Bez_n(f, g)_{i,j} x^i y^j. \qedhere
	\]
\end{definition}

For more on B\'{e}zout matrices and their applications to control theory and elimination theory, we refer to \cite[Chapter 2, Section 9]{BP94} and references therein. 
In this work, we restrict our attention to computing the determinant and inverse of B\'{e}zout matrices.

The determinant of the B\'{e}zout matrix provides an alternate way to compute the resultant of two polynomials.

\begin{lemma} \label{lem:bezout determinant}
	Let $f, g \in \F[x]$ be univariate polynomials and let $d = \max(\deg(f), \deg(g))$.
	Then $\det \Bez_d(f, g) = \res(f, g)$.
\end{lemma}

The inverse of the B\'{e}zout matrix $\Bez_n(f,g)$ is easily described in terms of the polynomials $f$ and $g$.
The following proposition shows that the inverse of a B\'{e}zout matrix is a Hankel matrix whose entries can be efficiently computed from the coefficients of $f$ and $g$.
Conversely, the inverse of any Hankel matrix is a B\'{e}zout matrix of some pair of polynomials.

\begin{proposition}[{\cite[Chapter 2, Proposition 9.3]{BP94}}]\label{prop:bezout inverse}
	Let $f \in \F[x]$ be a monic polynomial of degree $n$ and let $g \in \F[x]$ be a polynomial of degree at most $n$.
	If $\res(f,g) \neq 0$, then $\Bez_n(f,g)$ is invertible.
	Moreover, the inverse $\Bez_n(f,g)^{-1}$ can be described as follows.
	Let $p \in \F[x]$ be a polynomial of degree at most $n-1$ that satisfies the congruence
	\[
		p(x) g(x) \equiv 1 \pmod{f(x)}.
	\]
	Let $h(x) = \sum_{i=0}^{\infty} h_i x^i$ be the power series expansion of $\frac{x^n p(1/x)}{x^n f(1/x)}$.
	Then the inverse of $\Bez_n(f,g)$ is given by the Hankel matrix
	\[
		\Bez_n(f,g)^{-1} = \begin{pmatrix}
			h_1 & h_2 & h_3 & \cdots & h_n \\
			h_2 & h_3 & h_4 & \cdots & h_{n+1} \\
			h_3 & h_4 & h_5 & \cdots & h_{n+2} \\
			\vdots & \vdots & \vdots & \ddots & \vdots \\
			h_n & h_{n+1} & h_{n+2} & \cdots & h_{2n-1}
		\end{pmatrix}. \qedhere
	\]
\end{proposition}

\subsection{Squarefree Decomposition}

This subsection recalls the \emph{squarefree decomposition} of a polynomial, a structured partial factorization.
We first recall the notion of a squarefree polynomial.

\begin{definition}
	Let $f \in \F[x]$ be a univariate polynomial.
	We say that $f$ is \emph{squarefree} if there is no non-constant polynomial $g \in \F[x]$ such that $g^2$ divides $f$.
\end{definition}

As an example, the polynomial $(x - 1)(x - 2)$ is squarefree, but $(x-1)^2 (x-2)$ is not.
By factoring a polynomial $f$ as a product of irreducible polynomials and erasing the exponents, one obtains the squarefree part of $f$.

\begin{definition}
	Let $f \in \F[x]$ be a univariate polynomial.
	Let $f = \prod_{i=1}^m f_i^{d_i}$ be the factorization of $f$ into irreducible polynomials, where $f_1,\ldots,f_m \in \F[x]$ are irreducible in $\F[x]$ and are pairwise coprime.
	The \emph{squarefree part} of $f$ is given by $\prod_{i=1}^m f_i$.
\end{definition}

Finally, we define the squarefree decomposition of a polynomial $f \in \F[x]$.
The squarefree decomposition is a partial factorization of $f$ into a structured product of squarefree polynomials.
The task of computing the squarefree decomposition is referred to as \emph{squarefree factorization} and is a basic step in algorithms for factoring polynomials; see \cite[Chapter 14]{vzGG13} for more.

\begin{definition}
	Let $f \in \F[x]$ be a univariate polynomial.
	The \emph{squarefree decomposition} of $f(x)$ is the unique sequence of monic squarefree pairwise coprime polynomials $(f_1,\ldots,f_m)$ such that $f = \prod_{i=1}^m f_i^i$ and $f_m \neq 1$.
\end{definition}

\section{Symmetric Polynomials and Newton's Identities} \label{sec:newton series}

In this section, we study two important families of symmetric polynomials, the elementary symmetric and power sum polynomials.
Applying these functions to the roots of a univariate polynomial $f$ give different representations of $f$, and these representations are useful for different algorithmic tasks.
As we will see, the representation using the elementary symmetric polynomials works nicely with additive operations, while the power sum polynomials are better suited for multiplicative operations.
It is extremely convenient that one can pass between these representations in $\AC^0_\F$---we exposit these reductions in this section.

Suppose we are given a polynomial $f \in \F[x]$ by its list of coefficients.
What do the coefficients tell us about the polynomial?
Recall that if $f$ factorizes as $f(x) = \prod_{i=1}^n (x - \alpha_i)$, then $f$ can be written as
\[
	f(x) = \sum_{i=0}^n (-1)^{n-i} e_{n - i}(\vec{\alpha})\, x^i,
\]
where the polynomial $e_d(x_1,\ldots,x_n)$ is given by
\[
	e_d(x_1,\ldots,x_n) = \sum_{\substack{S \subseteq [n] \\ |S| = d}} \prod_{i \in S} x_i.
\]
That is, the degree-$i$ coefficient of $f$ is, up to a sign, the elementary symmetric polynomial of degree $n-i$ evaluated at $\vec{\alpha}$.

The elementary symmetric polynomials are essential in the study of symmetric polynomials, i.e., polynomials that are invariant under permutations of the variables.
The \emph{fundamental theorem of symmetric polynomials} says that for any symmetric polynomial $f(\vec{x})$, there exists a (not necessarily symmetric) polynomial $g(\vec{y})$ such that
\[
	f(\vec{x}) = g(e_1(\vec{x}), \ldots, e_n(\vec{x})),
\]
and moreover this polynomial $g$ is unique.
That is, any symmetric polynomial can be written as a polynomial combination of the elementary symmetric polynomials.
This means that from the coefficients of $f$, we can compute \emph{any} symmetric function of the roots $\alpha_1,\ldots,\alpha_n$.
The central theme of this paper is to understand which of these functions can be computed \emph{efficiently}.

In this section, we study an important family of symmetric polynomials, the power sum polynomials.
Recall that the degree-$d$ power sum polynomial $p_d(x_1,\ldots,x_n)$ is given by
\[
	p_d(x_1,\ldots,x_n) = \sum_{i=1}^n x_i^d.
\]
An analogue of the fundamental theorem of symmetric polynomials holds for the power sum polynomials when $\ch(\F) = 0$ or $\ch(\F) > n$.
Given a symmetric polynomial $f(\vec{x})$, there is a (not necessarily symmetric) polynomial $h(\vec{y})$ such that
\[
	f(\vec{x}) = h(p_1(\vec{x}), \ldots, p_n(\vec{x})),
\]
and like before, this $h$ is unique.

The two versions of the fundamental theorem of symmetric polynomials mentioned above imply that if we are given the coefficients $e_1(\vec{\alpha}),\ldots,e_n(\vec{\alpha})$ of a polynomial $f$, we can compute the power sums $p_1(\vec{\alpha}), \ldots, p_n(\vec{\alpha})$ of its roots, and vice-versa.
This computation can be made efficient by making use of explicit identities that relate the elementary symmetric and power sum polynomials.
Such identities are well-known, and go by the name of Newton's identities (or the Girard--Newton identities).
For $1 \le d \le n$, we have 
\[
	p_d(\vec{\alpha}) = (-1)^{d-1} d \cdot e_{d}(\vec{\alpha}) + \sum_{i=1}^{d-1} (-1)^{d-1+i} e_{d-i}(\vec{\alpha}) p_i(\vec{\alpha}),
\]
and for $d > n$, we instead have
\[
	0 = \sum_{i=d-n}^d (-1)^{i-1} e_{d-i}(\vec{\alpha}) p_i(\vec{\alpha}).
\]
If we are given the values of $e_1(\vec{\alpha}),\ldots,e_n(\vec{\alpha})$, we can compute the values of $p_d(\vec{\alpha})$ iteratively using dynamic programming.
Conversely, when $\ch(\F) = 0$ or $\ch(\F) > n$, Newton's identities show that the elementary symmetric polynomials can be computed from the power sum polynomials.
The requirement that the field has large characteristic stems from the need to invert $d$ when computing the value of $e_d(\vec{\alpha})$.
As before, given the values of $p_1(\vec{\alpha}),\ldots,p_n(\vec{\alpha})$, we can iteratively compute the $e_d(\vec{\alpha})$ via dynamic programming.

Although these conversions between the elementary symmetric and power sum polynomials are efficient, their natural implementations are iterative, which would yield circuits of large depth.
To obtain $\AC^0_\F$ algorithms, we make use of well-known alternate forms of Newton's identities as identities of formal power series.
These identities and algorithms are by no means new.
For example, \textcite{SW01} made use of these identities to compute the elementary symmetric polynomials using circuits of depth 4 and depth 6 that are smaller than the circuits of size $O(n^2)$ and depth 3 constructed by Ben-Or.
Because a low-depth implementation of Newton's identities plays such a fundamental role in our work, we include a description of it for the sake of completeness.

We now define the \emph{Newton series} of a polynomial $f(x)$.
This is the formal power series whose coefficients are the power sum polynomials evaluated at the roots of $f$.

\begin{definition}[\cite{BFSS06}]
	Let $f(x) \in \F[x]$ be a univariate polynomial of degree $n$.
	Let $\alpha_1,\ldots,\alpha_n \in \F$ be the roots of $f$, counted with multiplicity.
	The \emph{Newton series} of $f$, denoted $\Newton(f)$, is the formal power series in $\F \llb t \rrb$ given by
	\[
		\Newton(f) \coloneqq \sum_{k=0}^\infty p_k(\vec{\alpha})\, t^k. \qedhere
	\]
\end{definition}

\subsection{From Coefficients to Power Sums}

We will now see how to efficiently convert between the values of $e_1(\vec{\alpha}),\ldots,e_n(\vec{\alpha})$ and $p_1(\vec{\alpha}), \ldots, p_n(\vec{\alpha})$. 
That is, we will convert between the representations of a polynomial $f(x) = \prod_{i=1}^n (x - \alpha_i)$ by its coefficients and by its Newton series $\Newton(f)$.

We first consider the task of computing $\Newton(f)$ up to a specified degree, given the coefficients of $f$.
To do this, we make use of the fact that $\Newton(f)$ can be expressed as a rational function in $t$, where the numerator and denominator can be easily computed from $f$.
Below, we recall the notion of the \emph{reversal} of a polynomial.

\begin{definition}
	Let $f(x) = \sum_{i=0}^n a_i x^i$ be a univariate polynomial of degree $n$ with $a_n \neq 0$.
	The \emph{reversal} of $f$, denoted $\rev(f)$, is the univariate polynomial
	\[
		\rev(f)(x) \coloneqq \sum_{i=0}^n a_{n-i} x^i. \qedhere
	\]
\end{definition}

The Newton series $\Newton(f)$ of a polynomial $f$ can be written as a rational function in terms of the reversal of $f$.

\begin{lemma}[{\cite[Lemma 1]{BFSS06}}] \label{lem:p_k from e_k identity}
	Let $f(x) = \sum_{i=0}^n a_i x^i$ be a univariate polynomial of degree $n$ with $a_n \neq 0$.
	Then $\Newton(f)$ is a rational function in $t$ given by
	\[
		\Newton(f) = \frac{\rev(f')(t)}{\rev(f)(t)}. \qedhere
	\]
\end{lemma}

Using the above expression for $\Newton(f)$ as a rational function, we can compute $\Newton(f)$ to degree $d$ by expanding $\frac{1}{\rev(f)(t)}$ as a power series in $t$ up to degree $d$, multiplying by $\rev(f')(t)$, and using polynomial interpolation to recover the coefficients of $\Newton(f)$.

\begin{lemma} \label{lem:p_k from e_k algorithm}
	Let $f \in \F[x]$ be a univariate polynomial of degree $n$ and let $d \in \naturals$.
	Then the coefficients of $\Newton(f)$ up to degree $d$ can be computed in $\AC^0_\F$.
\end{lemma}

\begin{proof}
	Write $f(x) = \sum_{i=0}^n f_i x^i$.
	Let $g(x) \coloneqq \rev(f')(x)$ and let $h(x) \coloneqq \rev(f)(x) - f_n$.
	Note that $h(0) = 0$.
	Let $p_k$ be the degree-$k$ coefficient of $\Newton(f)$.
	\cref{lem:p_k from e_k identity} implies
	\[
		\sum_{k=0}^\infty p_k x^k = \frac{g(x)}{f_n + h(x)}.
	\]
	Because $h(0) = 0$, we can invert $f_n + h(x)$ in $\F \llb x \rrb$, giving us
	\[
		\sum_{k=0}^\infty p_k x^k = g(x) \cdot \frac{1}{f_n} \sum_{k=0}^\infty \del{\frac{-h(x)}{f_n}}^k.
	\]
	This implies
	\[
		\frac{g(x)}{f_n} \sum_{k=0}^d \del{\frac{-h(x)}{f_n}}^k = \sum_{k=0}^d p_k x^k + x^{d+1} r(x)
	\]
	for some polynomial $r(x) \in \F[x]$.
	Let $s(x) \coloneqq \frac{g(x)}{f_n} \sum_{k=0}^d \del{\frac{-h(x)}{f_n}}^k$.
	The above equality implies that we can recover the desired coefficients $p_0,\ldots,p_d$ of $\Newton(f)$ by interpolating the coefficients of $s(x)$.

	Observe that the coefficients of $g(x)$ and $h(x)$ can be computed from the coefficients of $f(x)$ in a straightforward manner by a circuit of depth 1 and size $O(n)$.
	This, together with the definition of $s(x)$, yields a circuit of size $O(n + d)$ and depth 4 that computes $s(x)$.
	Note that $\deg(s) \le d (n+1)$.
	By \cref{lem:polynomial interpolation}, we can compute the coefficients of $s(x)$ using a circuit of size $O(n d^2)$ and depth 5 as desired.
\end{proof}

\subsection{From Power Sums to Coefficients}

To recover a polynomial from its Newton series, we make use of the following identity.
\textcite{SW01} used this identity to construct smaller low-depth circuits for the elementary symmetric polynomials.

\begin{lemma} \label{lem:e_k from p_k identity}
	Fix $n \in \naturals$.
	Let $\exp(t) \coloneqq \sum_{k=0}^\infty \frac{1}{k!} t^k \in \F \llb t \rrb$ denote the exponential power series.
	Then we have the equality of formal power series
	\[
		\sum_{k=0}^n e_k(\vec{x}) t^k = \exp\del{\sum_{k=1}^\infty \frac{(-1)^{k+1}}{k} p_k(\vec{x}) t^k}. \qedhere
	\]
\end{lemma}

Note that in a sense, the power sum polynomials correspond to the logarithm of the elementary symmetric polynomials.
We will witness the power of this in the next section.

We now show how to recover the coefficients of a polynomial from its Newton series.
Note that $\Newton(f) = \Newton(\alpha f)$ for any polynomial $f$ and nonzero $\alpha \in \F$, so we cannot hope to recover the leading coefficient of $f$.
This will not be an issue for us, as we restrict our attention to monic polynomials throughout this work.

\begin{lemma} \label{lem:e_k from p_k algorithm}
	Let $\F$ be a field of characteristic zero or characteristic greater than $n$.
	Let $f \in \F[x]$ be a univariate polynomial of degree $n$.
	Given the coefficients of $\Newton(f)$ of degree up to $n$ as input, the coefficients of $f(x)$ can be computed in $\AC^0_\F$.
\end{lemma}

\begin{proof}
	Let $\alpha_1,\ldots,\alpha_n \in \F$ denote the roots of $f(x)$.
	Recall that $e_k(\vec{\alpha})$ is the coefficient of $x^{n-k}$ in $f(x)$ and $p_k(\vec{\alpha})$ is the coefficient of $t^k$ in $\Newton(f)$.
	\cref{lem:e_k from p_k identity} implies
	\[
		\sum_{k=0}^n e_k(\vec{\alpha}) x^k = \exp\del{\sum_{k=1}^\infty \frac{(-1)^{k+1}}{k} p_k(\vec{\alpha}) x^k}.
	\]
	Let
	\[
		g(x) \coloneqq \sum_{k=1}^n \frac{(-1)^{k+1}}{k} p_k(\vec{\alpha}) x^k.
	\]
	Note that because $\ch(\F) = 0$ or $\ch(\F) > n$, the polynomial $g$ is well-defined.
	The preceding identity implies that
	\[
		\sum_{k=0}^n \frac{g(x)^k}{k!} = \sum_{k=0}^n e_k(\vec{\alpha}) x^k + x^{n+1} h(x),
	\]
	where $h(x) \in \F[x]$ is some polynomial in $x$.
	Let $r(x) \coloneqq \sum_{k=0}^n \frac{g(x)^k}{k!}$.
	(As with $g$, the polynomial $r$ is well-defined because $\ch(\F) = 0$ or $\ch(\F) > n$.)
	By interpolating the coefficients of $r(x)$, we can recover the values $e_1(\vec{\alpha}),\ldots,e_n(\vec{\alpha})$.

	Observe that the coefficients of $g(x)$ are easily determined from the coefficients of $\Newton(f)$.
	This yields an arithmetic circuit of size $O(n)$ and depth 1 that computes $g(x)$.
	The truncated exponential $r(x)$ can then be computed by an arithmetic circuit of size $O(n)$ and depth 3.
	Note that $\deg(r) \le n^2$, so \cref{lem:polynomial interpolation} implies that we can compute the coefficients of $r$ using an arithmetic circuit of size $O(n^3)$ and depth 4.
	In particular, this circuit computes the values $e_1(\vec{\alpha}),\ldots,e_n(\vec{\alpha})$.
	The polynomial $f(x)$ is given by $f(x) = \sum_{i=0}^n (-1)^{n-i} e_{n-i}(\vec{\alpha}) x^i$, so multiplying each $e_i(\vec{\alpha})$ by $(-1)^i$ results in the coefficients of $f(x)$.
\end{proof}

\section{Exact Division and Roots of Perfect Powers} \label{sec:exact division}

In this section, we warm up with two applications of the algorithms described in \cref{sec:newton series}.
We design $\AC^0_\F$ algorithms to compute the quotient $f/g$ when $g$ is promised to divide $f$, and to compute $f^{1/r}$ when $f$ is promised to be an $r$\ts{th} power of a polynomial.
These are toy problems, and the simple algorithms we design are meant to illustrate the utility of representing a polynomial by its Newton series.
The Newton series plays the role of a logarithm, converting multiplication to addition and division to subtraction, which is easy to see from the fact that $\Newton(fg) = \Newton(f) + \Newton(g)$.

We now describe an $\AC^0_\F$ algorithm to compute the quotient $f/g$ when $g$ is promised to divide $f$.

\begin{lemma} \label{lem:exact division}
	Let $\F$ be a field of characteristic zero or characteristic greater than $d$.
	Let $f, g \in \F[x]$ be monic univariate polynomials of degree at most $d$ given by their coefficients.
	Suppose that $g$ divides $f$.
	Then the coefficients of $f/g$ can be computed in $\AC^0_\F$.
\end{lemma}

\begin{proof}
	Suppose $f(x)$ and $g(x)$ factor as
	\begin{align*}
		f(x) &= \prod_{i=1}^n (x - \alpha_i) \prod_{i=1}^{m} (x - \beta_i) \\
		g(x) &= \prod_{i=1}^n (x - \alpha_i),
	\end{align*}
	where $\alpha_1, \ldots, \alpha_n, \beta_1, \ldots, \beta_m \in \F$ and the $\alpha_i$ and $\beta_i$ are not necessarily distinct.
	Applying \cref{lem:p_k from e_k algorithm} to $f(x)$ and $g(x)$, we can compute the power sums $p_k(\vec{\alpha}) + p_k(\vec{\beta})$ and $p_k(\vec{\alpha})$, respectively, for all $k \in [m]$ in $\AC^0_\F$.
	Taking differences, we obtain the power sums $p_k(\vec{\beta})$ for all $k \in [m]$.
	Applying \cref{lem:e_k from p_k algorithm} to the power sums $p_k(\vec{\beta})$ yields the coefficients of $\prod_{i=1}^{m} (x - \beta_i)$ in $\AC^0_\F$ as desired.
\end{proof}

Next, we describe an $\AC^0_\F$ algorithm to compute the $r$\ts{th} root $f^{1/r}$ when $f$ is promised to be an $r$\ts{th} power of a polynomial.

\begin{lemma} \label{lem:perfect power}
	Let $\F$ be a field of characteristic zero or characteristic greater than $d$.
	Let $f \in \F[x]$ be a univariate polynomial of degree at most $d$ given by its coefficients and let $r \in \naturals$.
	Suppose that $f = g^r$ for some $g \in \F[x]$.
	Then the coefficients of $g$ can be computed in $\AC^0_\F$.
\end{lemma}

\begin{proof}
	Suppose that $g$ factors as 
	\[
		g(x) = \prod_{i=1}^n (x - \alpha_i),
	\]
	where $\alpha_1,\ldots,\alpha_n \in \F$ and the $\alpha_i$ are not necessarily distinct.
	Because $f = g^r$, this implies that $f$ factors as
	\[
		f(x) = \prod_{i=1}^n (x - \alpha_i)^r.
	\]
	Applying \cref{lem:p_k from e_k algorithm} to $f$, we compute in $\AC^0_\F$ the power sums
	\[
		\sum_{i=1}^n r \cdot \alpha_i^k = r \cdot p_k(\vec{\alpha})
	\]
	for each $k \in [n]$.
	Dividing by $r$ yields $p_k(\vec{\alpha})$ for all $k \in [n]$.
	Applying \cref{lem:e_k from p_k algorithm} to the power sums $p_k(\vec{\alpha})$ produces the coefficients of $g$, as desired.
\end{proof}

\section{Computing Symmetric Functions of the Roots of a Polynomial} \label{sec:symmetric functions of roots}

In this section, we return to the topic of computing symmetric functions of the roots of a polynomial $f \in \F[x]$ when $f$ is given by its coefficients.
The coefficients of $f$ are the elementary symmetric functions of its roots. 
As we saw in \cref{sec:newton series}, we can compute the power sums of the roots in $\AC^0_\F$.
In this section, we expand our toolbox, finding more symmetric functions that can be evaluated at the roots of a given polynomial $f$ in $\AC^0_\F$.

\subsection{Polynomial Functions}

Let $f, g \in \F[x]$ be univariate polynomials and let $\alpha_1, \ldots, \alpha_n \in \F$ be the roots of $f$, counted with multiplicity.
In this subsection, we evaluate the elementary symmetric polynomials at the values $g(\alpha_1),\ldots,g(\alpha_n)$.

We start by computing the sum $\sum_{i=1}^n g(\alpha_i)$.

\begin{lemma} \label{lem:sum over roots}
	Let $f, g \in \F[x]$ be univariate polynomials given by their coefficients.
	Suppose that $\alpha_1,\ldots,\alpha_n \in \F$ are the roots of $f$, counted with multiplicity.
	Then the sum $\sum_{i=1}^n g(\alpha_i)$ can be computed in $\AC^0_\F$.
\end{lemma}

\begin{proof}
	Letting $g(x) = \sum_{i=0}^m g_i x^i$, we can rewrite the sum $\sum_{i=1}^n g(\alpha_i)$ as 
	\[
		\sum_{i=1}^n g(\alpha_i) = \sum_{i=1}^n \sum_{j=0}^m g_j \alpha_i^j = \sum_{j=0}^m g_j \sum_{i=1}^n \alpha_i^j = \sum_{j=0}^m g_j p_j(\vec{\alpha}).
	\]
	We can compute the power sums $p_k(\vec{\alpha})$ for $0 \le k \le m$ in $\AC^0_\F$ by applying \cref{lem:p_k from e_k algorithm} to $f(x)$.
	The sum $\sum_{j=0}^m g_j p_j(\vec{\alpha})$ can then be computed by a subcircuit of size $O(m)$ and depth 2.
\end{proof}

This shows that the sum $\sum_{i=1}^n g(\alpha_i)$ is easy to compute.
It would be interesting to find an $\AC^0_\F$ algorithm to compute the product $\prod_{i=1}^n g(\alpha_i)$, since this equals $\res(f,g)$ by \cref{lem:resultant as product}.
We will not only compute the product $\prod_{i=1}^n g(\alpha_i)$, but we will in fact compute all the elementary symmetric polynomials evaluated at $g(\alpha_1),\ldots,g(\alpha_n)$.
To do this, it suffices (by Newton's identities) to compute the power sums $\sum_{i=1}^n g(\alpha_i)^k$, which can be done by a straightforward application of \cref{lem:sum over roots}.

\begin{lemma} \label{lem:esym over roots}
	Let $\F$ be a field of characteristic zero or characteristic greater than $n$.
	Let $f, g \in \F[x]$ be univariate polynomials given by their coefficients.
	Suppose that $\alpha_1,\ldots,\alpha_n \in \F$ are the roots of $f$, counted with multiplicity.
	Then for any $d \in [n]$, the elementary symmetric polynomial $e_d(g(\alpha_1),\ldots,g(\alpha_n))$ can be computed in $\AC^0_\F$.
\end{lemma}

\begin{proof}
	For any $k \in [n]$, the coefficients of $g(x)^k$ can be computed in $\AC^0_\F$ using \cref{lem:polynomial interpolation}.
	By \cref{lem:sum over roots}, we can compute the sum $\sum_{i=1}^n g(\alpha_i)^k$ in $\AC^0_\F$ for each $k \in [n]$.
	Applying \cref{lem:e_k from p_k algorithm} to the power sums $\set{p_k(g(\alpha_1),\ldots,g(\alpha_n)) : k \in [n]}$ yields the elementary symmetric polynomials $\set{e_k(g(\alpha_1),\ldots,g(\alpha_n)) : k \in [n]}$, again in $\AC^0_\F$.
\end{proof}

Later in \cref{sec:operations on roots}, it will be useful for us to compute the product of only the nonzero values of $g(\alpha_i)$.
We can do this by computing all elementary symmetric polynomials $e_k(g(\alpha_1),\ldots,g(\alpha_n))$ and then (piecewise) selecting the largest index $k$ such that $e_k(g(\alpha_1),\ldots,g(\alpha_n)) \neq 0$.

\begin{lemma} \label{lem:esym as multiplication}
	Let $\F$ be a field of characteristic zero or characteristic greater than $n$.
	Let $f, g \in \F[x]$ be univariate polynomials given by their coefficients.
	Suppose that $\alpha_1,\ldots,\alpha_n \in \F$ are the roots of $f$, counted with multiplicity.
	Let $S = \set{i \in [n] : g(\alpha_i) \neq 0}$.
	Then the product $\prod_{i \in S} g(\alpha_i)$ can be computed piecewise in $\AC^0_\F$.
\end{lemma}

\begin{proof}
	By relabeling the roots of $f$, we may assume without loss of generality that $S = \set{1,\ldots,d}$.
	Let
	\begin{align*}
		\lambda_1 &\coloneqq g(\alpha_1) \neq 0 \\
		&\vdots \\
		\lambda_d &\coloneqq g(\alpha_d) \neq 0 \\
		\lambda_{d+1} &\coloneqq g(\alpha_{d+1}) = 0 \\
		&\vdots \\
		\lambda_n &\coloneqq g(\alpha_n) = 0,
	\end{align*}
	where $d \in \naturals$ is unknown to us.
	By \cref{lem:esym over roots}, we can compute $e_k(\vec{\lambda})$ for all $k \in [n]$ in $\AC^0$, and we will leverage this to compute $\prod_{i=1}^d \lambda_i$.

	When $k > d$, the degree-$k$ elementary symmetric polynomial evaluated at $\vec{\lambda}$ vanishes.
	To see this, expand the elementary symmetric polynomial as
	\[
		e_k(\lambda_1,\ldots,\lambda_n) = \sum_{\substack{S \subseteq [d] \\ T \subseteq \set{d+1,\ldots,n} \\ |S| + |T| = k}} \prod_{i \in S} \lambda_i \prod_{j \in T} \lambda_j.
	\]
	For $k > d$, every term in the above sum corresponds to a choice of $T \subseteq \set{d+1,\ldots,n}$ that is nonempty.
	Each such term incurs a factor of $\lambda_j = 0$, so the sum simplifies to zero.

	On the other hand, the degree-$d$ elementary symmetric polynomial evaluated at $\vec{\lambda}$ equals $\prod_{i=1}^d \lambda_i$.
	This follows from the fact that every term in the expansion of $e_d(\vec{\lambda})$ corresponding to a nonempty choice of $T \subseteq \set{d+1,\ldots,n}$ is zero, so the sum simplifies to the single nonzero term corresponding to $T = \varnothing$.

	This allows us to compute the product $\prod_{i=1}^d \lambda_i$ if we know the value of $d$.
	We can recover $d$ by finding the largest index $\hat{d}$ at which the degree-$\hat{d}$ elementary symmetric polynomial is nonzero.
	This $\hat{d}$ is precisely $d$, the number of nonzero $\lambda_i$.

	To formalize this algorithm as a piecewise $\AC^0$ circuit, we use a selection gate to output $\sel(e_n(\vec{\lambda}), e_{n-1}(\vec{\lambda}), \ldots, e_1(\vec{\lambda}))$.
\end{proof}

\begin{remark}
	Note that \cref{lem:sum over roots,lem:esym over roots,lem:esym as multiplication} extend to the setting where $g \in \F[x, \vec{y}]$ is a polynomial in many variables.
	In this variant, we want to compute the sum
	\[
		\sum_{i=1}^n g(\alpha_i, \vec{y}),
	\]
	and more generally the elementary symmetric polynomial
	\[
		e_d( g(\alpha_1, \vec{y}), \ldots, g(\alpha_n, \vec{y}) ).
	\]
	To do this, regard $g \in \F[\vec{y}][x]$ as a univariate polynomial in $x$ whose coefficients are polynomials in $\vec{y}$.
	When $g$ is given by its coefficients as a polynomial in $\F[x,\vec{y}]$, it is straightforward to form the coefficients of $g \in \F[\vec{y}][x]$ in $\AC^0_\F$.
	If $g$ is instead given by an $\AC^0_\F$ circuit, then we can obtain $\AC^0_\F$ circuits that compute the coefficients of $g \in \F[\vec{y}][x]$ using \cref{lem:polynomial interpolation}.
	With this modification, the proofs of \cref{lem:sum over roots,lem:esym over roots,lem:esym as multiplication} go through without modification.
	The ability to apply \cref{lem:sum over roots,lem:esym over roots,lem:esym as multiplication} in this setting will be an essential tool that we use throughout our work.
\end{remark}

\subsection{Rational Functions}

\cref{lem:sum over roots,lem:esym over roots} generalize to the setting where we want to evaluate a rational function $g/h$ at the roots of a polynomial $f$.
We start by summing a rational function $g/h$ over the roots of $f$.
Of course, this requires that $h$ is nonzero at the roots of $f$.

\begin{lemma} \label{lem:rational sum over roots}
	Let $\F$ be a field of characteristic zero or characteristic greater than $n$.
	Let $f, g, h \in \F[x]$ be univariate polynomials given by their coefficients.
	Suppose that $\alpha_1,\ldots,\alpha_n \in \F$ are the roots of $f$, counted with multiplicity.
	Assume that $h(\alpha_i) \neq 0$ for all $i \in [n]$.
	Then the sum $\sum_{i=1}^n \frac{g(\alpha_i)}{h(\alpha_i)}$ can be computed in $\AC^0_\F$.
\end{lemma}

\begin{proof}
	Writing the sum $\sum_{i=1}^n \frac{g(\alpha_i)}{h(\alpha_i)}$ over a common denominator, we have
	\[
		\sum_{i=1}^n \frac{g(\alpha_i)}{h(\alpha_i)} = \frac{\sum_{i=1}^n g(\alpha_i) \prod_{j \neq i} h(\alpha_j)}{\prod_{i=1}^n h(\alpha_i)}.
	\]
	Applying \cref{lem:esym over roots} to $f$ and $h$ allows us to compute $\prod_{i=1}^n h(\alpha_i)$ in $\AC^0_\F$, so we are left with the task of computing $\sum_{i=1}^n g(\alpha_i) \prod_{j \neq i} h(\alpha_j)$.

	Let $y$ be a fresh variable. 
	Observe that when we expand the polynomial
	\[
		r(y) \coloneqq \prod_{i=1}^n (g(\alpha_i) y + h(\alpha_i)),
	\]
	the coefficient of the degree-$1$ term is precisely $\sum_{i=1}^n g(\alpha_i) \prod_{j \neq i} h(\alpha_j)$.
	By applying \cref{lem:esym over roots} to the polynomials $f(x)$ and $g(x) y + h(x)$, we obtain a circuit of constant depth and polynomial size that computes $r(y)$.
	As $\deg(r) = n$, \cref{lem:polynomial interpolation} implies that we can compute the coefficients of $r(y)$ in $\AC^0_\F$. 
	This yields an $\AC^0_\F$ algorithm to compute the sum $\sum_{i=1}^n g(\alpha_i) \prod_{j \neq i} h(\alpha_j)$, and hence an $\AC^0_\F$ algorithm to compute $\sum_{i=1}^n \frac{g(\alpha_i)}{h(\alpha_i)}$.
\end{proof}

We now use \cref{lem:rational sum over roots} to compute any elementary symmetric function of the values $\frac{g(\alpha_1)}{h(\alpha_1)},\ldots,\frac{g(\alpha_n)}{h(\alpha_n)}$.

\begin{lemma} \label{lem:rational esym over roots}
	Let $\F$ be a field of characteristic zero or characteristic greater than $n$.
	Let $f, g, h \in \F[x]$ be univariate polynomials given by their coefficients.
	Suppose that $\alpha_1,\ldots,\alpha_n \in \F$ are the roots of $f$, counted with multiplicity.
	Assume that $h(\alpha_i) \neq 0$ for all $i \in [n]$.
	Then for any $d \in [n]$, the elementary symmetric function
	\[
		e_d \del{\frac{g(\alpha_1)}{h(\alpha_1)}, \ldots, \frac{g(\alpha_n)}{h(\alpha_n)}}
	\]
	can be computed in $\AC^0_\F$.
\end{lemma}

\begin{proof}
	For any $k \in [n]$, we can compute the coefficients of $g(x)^k$ and $h(x)^k$ in $\AC^0_\F$ using \cref{lem:polynomial interpolation}.
	Applying \cref{lem:rational sum over roots} to $f(x)$, $g(x)^k$, and $h(x)^k$ computes the sum $\sum_{i=1}^n \frac{g(\alpha_i)^k}{h(\alpha_i)^k}$ in $\AC^0_\F$.
	We then compute the desired elementary symmetric polynomial by invoking \cref{lem:e_k from p_k algorithm} on the power sums $\sum_{i=1}^n \frac{g(\alpha_i)^k}{h(\alpha_i)^k}$.
\end{proof}

\section{The Sylvester and B\'{e}zout Matrices} \label{sec:sylvester and bezout}

In this section, we apply the results of \cref{sec:symmetric functions of roots} to compute the determinant and inverse of the Sylvester and B\'{e}zout matrices of a pair of polynomials in $\AC^0_\F$.
We also extend our division algorithm from \cref{lem:exact division} to handle polynomial division with remainder.

\subsection{The Resultant and Discriminant} \label{subsec:resultant}

We start by designing a constant-depth circuit to compute the resultant of two polynomials.
As the resultant is precisely the determinant of the Sylvester (\cref{def:sylvester matrix}) and B\'{e}zout matrices (\cref{lem:bezout determinant}), this provides an $\AC^0_\F$ algorithm to compute the determinants of matrices of these forms.

To compute the resultant, let $f, g \in \F[x]$ be monic polynomials and let $\alpha_1,\ldots,\alpha_n \in \F$ be the roots of $f$, counted with multiplicity.
By \cref{lem:resultant as product}, we know that
\[
	\res(f,g) = \prod_{i=1}^n g(\alpha_i).
\]
This is precisely the $n$\ts{th} elementary symmetric polynomial evaluated at $(g(\alpha_1),\ldots,g(\alpha_n))$.
Thus, we can compute $\res(f,g)$ using a direct application of \cref{lem:esym over roots}.

\begin{theorem} \label{thm:resultant ac0}
	Let $f, g \in \F[x]$ be univariate polynomials given by their coefficients.
	Then the resultant $\res(f,g)$ can be computed in $\AC^0_\F$.
\end{theorem}

\begin{proof}
	This is an immediate consequence of \cref{lem:resultant as product} and \cref{lem:esym over roots}.
\end{proof}

As an immediate corollary, we obtain an $\AC^0_\F$ algorithm to compute the discriminant of a single polynomial.

\begin{corollary} \label{cor:discriminant ac0}
	Let $f \in \F[x]$ be a univariate polynomial given by its coefficients.
	Then the discriminant $\disc(f)$ can be computed in $\AC^0_\F$.
\end{corollary}

\begin{proof}
	Because $\disc(f) = (-1)^{\binom{n}{2}} \res(f,f')$, this immediately follows from \cref{thm:resultant ac0}.
\end{proof}

The resultant and discriminant are extremely useful tools in algebra and number theory.
Using well-known applications of the resultant and discriminant, we present three corollaries of \cref{thm:resultant ac0} and \cref{cor:discriminant ac0}.

Our first corollary deals with computing implicit equations of rational plane curves.
A \emph{rational plane curve} is the (Zariski closure of the) image of a rational map 
\begin{align*}
	\gamma : \F &\to \F^2 \\
	t &\mapsto \del{\frac{f(t)}{h(t)}, \frac{g(t)}{h(t)}}
\end{align*}
where $f, g, h \in \F[t]$ are univariate polynomials.
For example, the unit circle has a rational parameterization given by
\[
	t \mapsto \del{\frac{1-t^2}{1+t^2}, \frac{2t}{1+t^2}}.
\]
The careful reader will notice that $(-1, 0)$ is on the unit circle, but is not in the image of this map.
This is addressed by taking the closure of the image in the \emph{Zariski topology}, the standard topology used in algebraic geometry.
We will not define the Zariski topology here, and instead refer the reader interested in the precise details to \cite{CLO15}.

Given a rational parameterization of a plane curve $\Gamma$, it is often useful to find an implicit equation for $\Gamma$.
An \emph{implicit equation} is a polynomial $r \in \F[x,y]$ such that $r(a,b) = 0$ if and only if $(a,b) \in \Gamma$.
Clearly, an implicit equation gives rise to an algorithm that decides if a given point $(a,b)$ lies on the curve $\Gamma$: compute $r(a,b)$ and check if this value equals zero.
In the case of the unit circle, one implicit equation is given by $x^2 + y^2 - 1$.
(Other implicit equations for the unit circle are $(x^2+y^2-1)^n$ for $n \in \naturals$, and these are essentially all possible equations for the unit circle.)

Resultants provide a straightforward method to compute an implicit equation of a rational plane curve.
Given a curve parameterized by
\[
	t \mapsto \del{\frac{f(t)}{h(t)}, \frac{g(t)}{h(t)}},
\]
one can show that the polynomial
\[
	r(x,y) \coloneqq \res_t(x \cdot h(t) - f(t), y \cdot h(t) - g(t)),
\]
where $x \cdot h(t) - f(t)$ and $y \cdot h(t) - g(t)$ are regarded as polynomials in $t$ with coefficients in $\F[x,y]$, is an implicit equation of the plane curve.
By computing this resultant with \cref{thm:resultant ac0} and then interpolating the coefficients of $r(x,y)$ using \cref{lem:polynomial interpolation}, we obtain an $\AC^0_\F$ algorithm to compute an implicit equation of a given rational plane curve.

\begin{corollary} \label{cor:implicit equation}
	Let $f, g, h \in \F[t]$ be univariate polynomials given by their coefficients.
	Let $\Gamma \subseteq \F^2$ be the plane curve corresponding to the map $t \mapsto (f(t)/h(t), g(t)/h(t))$.
	Then the coefficients of an implicit equation $r \in \F[x,y]$ for $\Gamma$ can be computed in $\AC^0_\F$.
\end{corollary}

The second application of \cref{thm:resultant ac0} is again to geometry.
Given the implicit equations of two plane curves $\Gamma_1, \Gamma_2 \in \F^2$, we would like to compute their intersection $\Gamma_1 \cap \Gamma_2$.
Equivalently, we want to solve the system of polynomial equations $f(x,y) = g(x,y) = 0$, where $f$ and $g$ are the implicit equations of $\Gamma_1$ and $\Gamma_2$, respectively.
Using resultants, we can reduce this problem to solving polynomial equations in one variable.
View $f$ and $g$ as elements of $\F[x][y]$ and let
\[
	h(x) \coloneqq \res_y(f(x,y), g(x,y))
\]
be their resultant.
If $h(x) = 0$, then $f$ and $g$ share a common factor, so the intersection $\Gamma_1 \cap \Gamma_2$ is infinite and corresponds to this common factor.

Suppose instead that $h(x) \neq 0$, so that $\Gamma_1 \cap \Gamma_2$ is finite.
Standard properties of resultants imply that if $(a,b) \in \Gamma_1 \cap \Gamma_2$, then $h(a) = 0$.
Thus, to compute $\Gamma_1 \cap \Gamma_2$, it suffices to first compute the roots of $h$, and for each such root $a \in \F$, find common roots of the univariate polynomials $f(a,y)$ and $g(a,y)$.
For more on resultants in elimination theory, see \cite[Chapter 3]{CLO15}.

Our third and final application is to combining the roots of polynomials in nontrivial ways.
Suppose we are given the coefficients of $f(x) = \prod_{i=1}^n (x - \alpha_i)$ and $g(x) = \prod_{i=1}^m (x - \beta_i)$.
Consider the polynomials
\begin{align*}
	(f \oplus g)(x) &\coloneqq \prod_{i, j} (x - (\alpha_i + \beta_j)) \\
	(f \otimes g)(x) &\coloneqq \prod_{i,j} (x - \alpha_i \beta_j),
\end{align*}
which are called the \emph{composed sum} and \emph{composed product} of $f$ and $g$, respectively.

The composed sum and composed product are useful in implementing arithmetic for algebraic numbers.
Recall that a number $\alpha \in \complexes$ is \emph{algebraic} if there is a nonzero polynomial $f \in \rationals[x]$ such that $f(\alpha) = 0$.
For every algebraic number $\alpha \in \complexes$, there is a nonzero polynomial $f \in \rationals[x]$ of minimal degree that vanishes at $\alpha$; such a polynomial is the \emph{minimal polynomial} of $\alpha$.
A natural way to represent algebraic numbers is by their minimal polynomial.
If $\alpha, \beta \in \complexes$ are represented by polynomials $f$ and $g$, respectively, then the sum $\alpha + \beta$ is a root of the composed sum $f \oplus g$.
This implies that the minimal polynomial of $\alpha + \beta$ is an irreducible factor (over $\rationals[x]$) of $f \oplus g$, which we can find by factoring $f \oplus g$.
Likewise, the minimal polynomial of $\alpha \beta$ is an irreducible factor of $f \otimes g$.

Fast algorithms to compute the composed sum and composed product were given by \textcite{BFSS06}, where the key tool was a fast algorithm to convert between the coefficient representation of a polynomial and its Newton series.
Using properties of resultants, it is a straightforward exercise (see, e.g., \cite[Section 3.6, Exercise 19]{CLO15}) to show that
\begin{align*}
	(f \oplus g)(x) &= \res_y(f(y), g(x - y)) \\
	(f \otimes g)(x) &= \res_y(f(y), y^m g(x/y)),
\end{align*}
where $f(y)$, $g(x-y)$, and $y^m g(x/y)$ are viewed as polynomials in $y$ with coefficients in $\F[x]$.
As a corollary of \cref{thm:resultant ac0}, we conclude $\AC^0_\F$ algorithms to compute the composed sum and composed product.
For more applications of the composed sum and composed product, see \cite[Section 5]{BFSS06}.

\begin{corollary} 
	Let $f, g \in \F[x]$ be univariate polynomials given by their coefficients.
	Suppose that $f$ and $g$ factor as $f(x) = \prod_{i=1}^n (x - \alpha_i)$ and $g(x) = \prod_{i=1}^m (x - \beta_i)$, where the $\alpha_i, \beta_i \in \F$ are not necessarily distinct, i.e., $f$ and $g$ are not necessarily squarefree.
	Then the coefficients of the polynomials
	\begin{align*}
		(f \oplus g)(x) &\coloneqq \prod_{i, j} (x - (\alpha_i + \beta_j)) \\
		(f \otimes g)(x) &\coloneqq \prod_{i,j} (x - \alpha_i \beta_j)
	\end{align*}
	can be computed in $\AC^0_\F$.
\end{corollary}

\subsection{Division with Remainder}

In this subsection, we take a brief detour from Sylvester and B\'{e}zout matrices to extend our division algorithm from \cref{lem:exact division} to handle polynomial division with remainder.
Recall that if $f, g \in \F[x]$ are univariate polynomials, then there are unique polynomials $q, r \in \F[x]$ such that $f = qg + r$ and $\deg(r) < \deg(g)$.
The polynomial $r$ is the \emph{remainder} of $f$ divided by $g$.
We will describe an $\AC^0_\F$ algorithm to compute the remainder $r$, which leads to an algorithm for division with remainder by dividing $f - r$ by $g$ using \cref{lem:exact division}.
The technique we use to compute the remainder $r$ will be useful later in \cref{subsec:sylvester inverse}, where we design an $\AC^0_\F$ algorithm to invert the Sylvester matrix $\Syl(f,g)$.

For the moment, suppose that $g$ is squarefree, i.e., that $g$ has $m$ distinct roots $\beta_1,\ldots,\beta_m \in \overline{\F}$ in the algebraic closure of $\F$, all of multiplicity 1.
This is a mild assumption, and we will sketch how to remove it below.
Assuming $g$ is squarefree, we can explicitly write the remainder $r$ using polynomial interpolation.
Consider the polynomial $\hat{r}(x)$ obtained by interpolating the values of $f$ at the points $\beta_1,\ldots,\beta_m$.
The Lagrange interpolation formula lets us write $\hat{r}$ as
\[
	\hat{r}(x) = \sum_{i=1}^m f(\beta_i) \prod_{j \neq i} \frac{x - \beta_j}{\beta_i - \beta_j}.
\]
By construction, we have $\deg(\hat{r}) \le m-1 < \deg(g)$.
The polynomial $f - \hat{r}$ is zero at each $\beta_i$.
Because $g$ is squarefree, it follows that $g$ divides $f - \hat{r}$.
That is, there is some polynomial $\hat{q}$ such that $f - \hat{r} = \hat{q} g$.
Rearranging, we have $f = \hat{q} g + \hat{r}$, so by uniqueness of the quotient and remainder, we have $q = \hat{q}$ and $r = \hat{r}$.
Thus, to compute the remainder $r$, it suffices to interpolate a polynomial that agrees with $f$ at the roots of $g$, at least when $g$ is squarefree.
We will perform this interpolation by implementing the Lagrange interpolation formula without explicitly computing the roots of $g$, making use of the tools developed in \cref{sec:symmetric functions of roots}.

When $g$ is not squarefree, we can still perform the above interpolation, but $g$ may not divide $f - \hat{r}$.
If $\beta_i \in \F$ is a multiple root of $g$, we would need to interpolate $\hat{r}$ so that $f - \hat{r}$ vanishes at $\beta_i$ to the same order as $g$.
However, we do not know the multiset of multiplicities of the roots of $g$, so it is not clear that this strategy will generalize nicely.

Instead, we make use of the fact that a polynomial $g$ is squarefree if and only if $\disc(g) \neq 0$.
With $\hat{r}$ defined through Lagrange interpolation as above, it is easy to see that $\disc(g) \cdot \hat{r} = \disc(g) \cdot r$.
When $\disc(g) = 0$, this is obvious, and when $\disc(g) \neq 0$, $g$ is squarefree, so the preceding sketch implies $\hat{r} = r$.
Thus, to obtain $r$, we can divide $\disc(g) \cdot \hat{r}$ by $\disc(g)$.
Although the resulting circuit is defined only when $\disc(g) \neq 0$, Strassen's theorem on division elimination (\cref{thm:ac0 division elimination}) implies that we can transform this circuit into an equivalent one that does not use division.
This results in a circuit that correctly computes the remainder $r$ even when $\disc(g) = 0$.

We now design an $\AC^0_\F$ algorithm to compute polynomial remainders.

\begin{lemma} \label{lem:remainder}
	Let $\F$ be a field of characteristic zero or characteristic greater than $d$.
	Let $f, g \in \F[x]$ be univariate polynomials of degree at most $d$ given by their coefficients.
	Let $r \in \F[x]$ be the remainder of $f$ divided by $g$.
	Then the coefficients of $r$ can be computed in $\AC^0_\F$.
\end{lemma}

\begin{proof}
	Let $n \coloneqq \deg(f)$ and $m \coloneqq \deg(g)$.
	Recall that $r \in \F[x]$ is the unique polynomial of degree less than $m$ such that $f - r$ is a multiple of $g$.
	Below, we will show that the coefficients of $\disc(g) r(x)$ can be computed in $\AC^0_\F$ without the use of $\sel$ gates.
	Before doing this, we explain why this yields an $\AC^0_\F$ algorithm to compute the coefficients of $r(x)$.

	We can compute $\disc(g)$ in $\AC^0_\F$ without $\sel$ gates via \cref{cor:discriminant ac0}.
	By dividing $\disc(g) r(x)$ by $\disc(g)$, we obtain an $\AC^0_\F$ algorithm \emph{with division} (but no $\sel$ gates) that computes $r(x)$.
	Because this algorithm divides by $\disc(g)$, its output is only defined when $\disc(g) \neq 0$.
	We would like to apply \cref{thm:ac0 division elimination} to conclude a division-free $\AC^0_\F$ algorithm that computes the remainder $r(x)$, which would allow us to compute $r(x)$ even when $\disc(g) = 0$.
	\footnote{As mentioned in \cref{remark: division elimination}, if $\F$ is finite, we can still eliminate this division if there is a polynomial $g \in \F[x]$ of degree at most $d$ such that $\disc(g)$ is nonzero. Irreducible polynomials have nonzero discriminant, and there are irreducible polynomials in $\F[x]$ of degree $d$ for all $d \in \naturals$ when $\F$ is finite, so the conclusion of \cref{thm:ac0 division elimination} still applies.}
	If we could do this, then the coefficients of $r(x)$ can be computed in $\AC^0_\F$ by applying \cref{lem:polynomial interpolation} to this division-free algorithm.
	To apply \cref{thm:ac0 division elimination}, we need to verify that the coefficients of $r(x)$ are polynomial functions of the coefficients of $f$ and $g$.

	To show that the coefficients of $r(x)$ are polynomial functions of the coefficients of $f$ and $g$, let $q$ be the quotient of $f$ by $g$, i.e., let $q$ be the polynomial such that $f = qg + r$.
	Write $q = \sum_{i=0}^{n-m} q_i x^i$ and $r = \sum_{i=0}^{m-1} r_i x^i$.
	By equating coefficients of powers of $x$ on both sides of $f = qg + r$, we obtain the linear system
	\[
		\begin{pmatrix}
			1 \\
			g_{m-1} & \ddots \\
			\vdots & & 1 \\
			\vdots & & g_{m-1} & 1 \\
			g_0 & & \vdots & & 1\\
			& \ddots & \vdots & & & \ddots \\
			& & g_0 & & & & 1
		\end{pmatrix}
		\begin{pmatrix}
			q_{n-m} \\
			\vdots \\
			q_0 \\
			r_{m-1} \\
			\vdots \\
			r_0 
		\end{pmatrix} 
		=
		\begin{pmatrix}
			1 \\
			f_{n-1} \\
			f_{n-2} \\
			\vdots \\
			f_0
		\end{pmatrix}.
	\]
	Let $\tilde{\Syl}(g,1)$ denote the matrix on the left-hand side above.
	(This matrix corresponds to $\Syl(g,1)$ when $1$ is treated as a polynomial of degree $n-m+1$.)
	The matrix $\tilde{\Syl}(g,1)$ is lower triangular and has ones along the diagonal, so it clearly has determinant 1.
	This implies that the entries of $\tilde{\Syl}(g,1)^{-1}$ are polynomial functions of the coefficients of $g$.
	It follows that
	\[
		\begin{pmatrix}
			q_{n-m} \\
			\vdots \\
			q_0 \\
			r_{m-1} \\
			\vdots \\
			r_0 
		\end{pmatrix} 
		=
		\tilde{\Syl}(g,1)^{-1}
		\begin{pmatrix}
			1 \\
			f_{n-1} \\
			f_{n-2} \\
			\vdots \\
			f_0
		\end{pmatrix}
	\]
	are polynomial functions of the coefficients of $f$ and $g$, as desired.

	It remains to show that we can compute $\disc(g) r(x)$ in $\AC^0_\F$.
	Suppose $\disc(g) \neq 0$, so that $g$ has $m$ simple roots $\beta_1,\ldots,\beta_m \in \F$.
	In this case, we can explicitly write $r(x)$ using Lagrange interpolation as
	\[
		r(x) = \sum_{i=1}^m f(\beta_i) \prod_{j \neq i} \frac{x - \beta_j}{\beta_i - \beta_j}.
	\]
	Combined with \cref{lem:discriminant as product}, we can write $\disc(g) r(x)$ as 
	\begin{align*}
		\disc(g) r(x) &= \del{(-1)^{\binom{m}{2}} \prod_{i\in [m]} \prod_{j \neq i} (\beta_i - \beta_j)} \sum_{i=1}^m f(\beta_i) \prod_{j \neq i} \frac{x - \beta_j}{\beta_i - \beta_j} \\
		&= (-1)^{\binom{m}{2}} \sum_{i=1}^m f(\beta_i) \prod_{j \neq i} (x - \beta_j) \prod_{k \neq j} (\beta_j - \beta_k) \\
		&= (-1)^{\binom{m}{2}} \sum_{i=1}^m f(\beta_i) \prod_{j \neq i} (x - \beta_j) \cdot g'(\beta_j).
	\end{align*}
	Let $y$ be a new variable.
	Observe that the above sum corresponds, up to the sign, to the coefficient of $y^{m-1}$ when the product
	\[
		\prod_{i=1}^m (f(\beta_i) + y (x - \beta_i) g'(\beta_i))
	\]
	is expanded as a polynomial in $y$ with coefficients in $\F[x]$.
	This is precisely the kind of expression that \cref{lem:esym over roots} allows us to compute.

	Let $z$ be a fresh variable.
	Define $h(x,y,z) \in \F[x,y,z]$ by
	\[
		h(x,y,z) \coloneqq f(z) + y (x - z) g'(z).
	\]
	It is clear that we can compute $h(x,y,z)$ in $\AC^0_\F$.
	Applying \cref{lem:esym over roots}, we compute
	\[
		\prod_{i=1}^m h(x,y,\beta_i) = \prod_{i=1}^m (f(\beta_i) + y (x - \beta_i) g'(\beta_i)),
	\]
	where $\beta_1,\ldots,\beta_m \in \F$ are the roots of $g$.
	We can then interpolate the coefficient of $y^{m-1}$ in the above polynomial using \cref{lem:polynomial interpolation}.
	In all, this yields an $\AC^0_\F$ algorithm to compute
	\[
		\rho(x) \coloneqq (-1)^{\binom{m}{2}} \sum_{i=1}^m f(\beta_i) \prod_{j \neq i} (x - \beta_j) \cdot g'(\beta_j).
	\]
	As the analysis above shows, when $\disc(g) \neq 0$, we have the polynomial identity $\rho(x) = \disc(g) r(x)$.
	It remains to show that the identity $\rho(x) = \disc(g) r(x)$ still holds when $\disc(g) = 0$.

	In the case $\disc(g) = 0$, the polynomial $g$ has a double root.
	Without loss of generality, suppose that $\beta_1 = \beta_2$.
	By definition of a double root, we have $g'(\beta_1) = g'(\beta_2) = 0$.
	This implies that the product
	\[
		\prod_{i=1}^m (f(\beta_i) + y (x - \beta_i) g'(\beta_i))
	\]
	has degree at most $m-2$ in $y$, so the coefficient of $y^{m-1}$ is zero, hence $\rho(x) = 0$.
	Because $\disc(g) = 0$, we again have the desired equality $\rho(x) = 0 = \disc(g) r(x)$.
\end{proof}

As a corollary, we obtain an $\AC^0_\F$ algorithm for polynomial division with remainder.

\begin{corollary}
	Let $\F$ be a field of characteristic zero or characteristic greater than $d$.
	Let $f, g \in \F[x]$ be univariate polynomials of degree at most $d$ given by their coefficients.
	Let $q, r \in \F[x]$ be the unique polynomials that satisfy $f = qg + r$ and $\deg(r) < \deg(g)$.
	Then the coefficients of $q$ and $r$ can be computed in $\AC^0_\F$.
\end{corollary}

\begin{proof}
	Applying \cref{lem:remainder} to $f$ and $g$, we compute in $\AC^0_\F$ a polynomial $\hat{r}$ of degree $\deg(\hat{r}) < \deg(g)$ such that $\hat{r} \equiv f \pmod{g}$.
	In particular, $f-\hat{r}$ is a multiple of $g$, so there is a polynomial $\hat{q} \in \F[x]$ such that $f - \hat{r} = \hat{q} g$.
	We can compute $\hat{q}$ in $\AC^0_\F$ using \cref{lem:exact division}.
	Thus, we have $f = \hat{q} g + \hat{r}$ and $\deg(\hat{r}) < \deg(g)$.
	Uniqueness of polynomial division with remainder implies that $\hat{q} = q$ and $\hat{r} = r$ are the quotient and remainder, respectively, of $f$ divided by $g$.
\end{proof}

\subsection{Inverting the Sylvester Matrix} \label{subsec:sylvester inverse}

This subsection describes an $\AC^0_\F$ circuit that computes the adjugate $\adj \Syl(f,g)$ of the Sylvester matrix.
Using the fact that $A^{-1} = \frac{1}{\det A} \adj A$ for any matrix $A$ and that $\res(f,g) = \det \Syl(f,g)$ can be computed in $\AC^0_\F$, this yields an $\AC^0_\F$ algorithm to compute the inverse $\Syl(f,g)^{-1}$ of the Sylvester matrix when $\Syl(f,g)$ is invertible.
As a corollary of this and \cref{lem:sylvester inverse}, we obtain an $\AC^0_\F$ algorithm that computes the B\'{e}zout coefficients of $f$ and $g$ when $\gcd(f,g) = 1$.
The algorithm and its proof of correctness are very similar to \cref{lem:remainder}.
The primary difference is that instead of interpolating $f$ over the roots of $g$, we need to interpolate $1/f$ over the roots of $g$.

\begin{theorem} \label{thm:sylvester adjugate}
	Let $\F$ be a field of characteristic zero or characteristic greater than $d$.
	Let $f, g \in \F[x]$ be univariate polynomials of degree at most $d$ given by their coefficients.
	Then the entries of $\adj \Syl(f,g)$ can be computed in $\AC^0_\F$.
\end{theorem}

\begin{proof}
	Let $n \coloneqq \deg(f)$ and $m \coloneqq \deg(g)$.
	Recall \cref{lem:sylvester inverse}: for $\ell \in [n+m]$, the entries in the $\ell$\ts{th} column of $\adj \Syl(f,g)$ correspond to the coefficients of polynomials $a_{\ell}, b_{\ell} \in \F[x]$ such that
	\[
		a_{\ell}(x) f(x) + b_{\ell}(x) g(x) = \res(f,g) x^{n+m-\ell},
	\]
	where $\deg(a_{\ell}) < m$ and $\deg(b_{\ell}) < n$.
	Below, we will show that the coefficients of $\disc(g) a_{\ell}(x)$ and $\disc(f) b_{\ell}(x)$ can be computed in $\AC^0_\F$ without $\sel$ gates.
	We can also compute $\disc(f)$ and $\disc(g)$ in $\AC^0_\F$ without $\sel$ gates using \cref{cor:discriminant ac0}.
	This yields an $\AC^0_\F$ algorithm with division that computes the coefficients of $a_{\ell}$ and $b_{\ell}$.
	Because the coefficients of $a_\ell$ and $b_\ell$ are the entries of the $\ell$\ts{th} column of $\adj \Syl(f,g)$, the coefficients of $a_\ell$ and $b_\ell$ are polynomial functions in the coefficients of $f$ and $g$.
	This allows us to apply \cref{thm:ac0 division elimination}, which yields the claimed $\AC^0_\F$ algorithm to compute the entries of $\adj \Syl(f,g)$.
	\footnote{When $\F$ is finite, by \cref{remark: division elimination}, we can still eliminate this division if there are polynomials $f, g \in \F[x]$ of degree at most $d$ such that $\disc(f)$ and $\disc(g)$ are nonzero. The existence of such polynomials follows from the fact that irreducible polynomials have nonzero discriminant and that $\F[x]$ contains irreducible polynomials of every degree $d \in \naturals$ when $\F$ is finite.}

	It remains to show that for every $\ell$, we can compute the coefficients of the polynomials $\disc(g) a_{\ell}(x)$ and $\disc(f) b_{\ell}(x)$.
	We start by computing $\disc(g) a_{\ell}(x)$.
	By exchanging the roles of $f$ and $g$ in the algorithm below, we can likewise compute $\disc(f) b_{\ell}(x)$.

	Suppose $\disc(g) \neq 0$, so that $g$ has $m$ simple roots $\beta_1,\ldots,\beta_m$.
	In this case, we can explicitly write $a_{\ell}(x)$ using Lagrange interpolation as
	\[
		a_{\ell}(x) = \res(f, g) \sum_{i=1}^m \frac{\beta_i^{n+m-\ell}}{f(\beta_i)} \prod_{j \neq i} \frac{x - \beta_j}{\beta_i - \beta_j}.
	\]
	Using \cref{lem:resultant as product} and \cref{lem:discriminant as product} to expand $\res(f,g)$ and $\disc(g)$, we have
	\begin{align*}
		\disc(g) a_{\ell}(x) &= \del{(-1)^{\binom{m}{2}} \prod_{i \in [m]} \prod_{j \neq i} (\beta_i - \beta_j)} \del{(-1)^{nm} \prod_{i=1}^m f(\beta_i)} \sum_{i=1}^m \frac{\beta_i^{n+m-\ell}}{f(\beta_i)} \prod_{j \neq i} \frac{x - \beta_j}{\beta_i - \beta_j} \\
		&= (-1)^{nm + \binom{m}{2}} \sum_{i=1}^m \beta_i^{n+m-\ell} \prod_{j \neq i} f(\beta_j) (x - \beta_j) \prod_{k \neq j} (\beta_j - \beta_k) \\
		&= (-1)^{nm + \binom{m}{2}} \sum_{i=1}^m \beta_i^{n+m-\ell} \prod_{j \neq i} f(\beta_j) g'(\beta_j) (x - \beta_j).
	\end{align*}
	Let $y$ be a new variable.
	Observe that the summation above is the coefficient of $y^{m-1}$ when the product
	\[
		\prod_{i=1}^m (\beta_i^{n+m-\ell} + y f(\beta_i) g'(\beta_i) (x - \beta_i))
	\]
	is expanded as a polynomial in $y$ with coefficients in $x$.
	We will compute this product by an application of \cref{lem:esym over roots}.

	Let $z$ be a fresh variable.
	Define $h(x,y,z) \in \F[x,y,z]$ by
	\[
		h(x,y,z) \coloneqq z^{n + m - \ell} + y (x - z) f(z) g'(z).
	\]
	It is clear from the definition that $h$ can be computed in $\AC^0_\F$.
	Using \cref{lem:esym over roots}, we compute
	\[
		\prod_{i=1}^m h(x,y,\beta_i) = \prod_{i=1}^m (\beta_i^{n + m - \ell} + y (x - \beta_i) f(\beta_i) g'(\beta_i))
	\]
	in $\AC^0_\F$, where $\beta_1,\ldots,\beta_m \in \F$ are the (not necessarily distinct) roots of $g$.
	Interpolating the coefficient of $y^{m-1}$ using \cref{lem:polynomial interpolation} and multiplying by $(-1)^{nm + \binom{m}{2}}$, we have an $\AC^0_\F$ algorithm to compute
	\[
		\hat{a}(x) \coloneqq (-1)^{nm + \binom{m}{2}} \sum_{i=1}^m \beta_i^{n+m-\ell} \prod_{j \neq i} f(\beta_j) g'(\beta_j) (x - \beta_j).
	\]
	As the analysis above shows, when $\disc(g) \neq 0$, we have the equality $\hat{a}(x) = \disc(g) a_\ell(x)$.
	It remains to show that this equality holds when $\disc(g) = 0$.

	If $\disc(g) = 0$, then $g$ has a double root.
	Without loss of generality, suppose that $\beta_1 = \beta_2$.
	By definition of a double root, we have $g'(\beta_1) = g'(\beta_2)$.
	This implies that the product
	\[
		\prod_{i=1}^m (\beta_i^{n + m - \ell} + y (x - \beta_i) f(\beta_i) g'(\beta_i))
	\]
	has degree at most $m-2$ in $y$, so $\hat{a}(x) = 0 = \disc(g) a_\ell(x)$ as desired.
\end{proof}

Recall that for coprime polynomials $f, g \in \F[x]$, the coefficients of the B\'{e}zout coefficients $a, b \in \F[x]$ appear in the last column of $\Syl(f,g)^{-1}$.
Using the fact that $\Syl(f,g)^{-1} = \frac{1}{\res(f,g)}\adj\Syl(f,g)$ and that both $\res(f,g)$ and $\adj \Syl(f,g)$ can be computed in $\AC^0_\F$, we see that $\Syl(f,g)^{-1}$ can be computed in $\AC^0_\F$, provided that $\Syl(f,g)^{-1}$ exists.
This yields an $\AC^0_\F$ algorithm to compute the B\'{e}zout coefficients of two coprime polynomials, which we record in the following corollary.
We will remove the requirement that $f$ and $g$ are coprime later in \cref{sec:gcd and lcm}.

\begin{corollary} \label{cor:bezout coprime}
	Let $\F$ be a field of characteristic zero or characteristic greater than $d$.
	Let $f, g \in \F[x]$ be univariate polynomials of degree at most $d$ given by their coefficients.
	Suppose that $\gcd(f,g) = 1$.
	Then the B\'{e}zout coefficients of $f$ and $g$ can be computed in $\AC^0_\F$.
\end{corollary}

As a second corollary of \cref{thm:sylvester adjugate}, we obtain an $\AC^0_\F$ algorithm to invert triangular Toeplitz matrices.
This follows from the simple fact that a triangular Toeplitz matrix can be embedded as a block of the Sylvester matrix for a particular pair of polynomials.
A different $\AC^0_\F$ algorithm for this problem was described by \textcite{Bini84}.

\begin{corollary}
	Let $\F$ be a field of characteristic zero or characteristic greater than $n$.
	Let $A \in \F^{n \times n}$ be a triangular Toeplitz matrix.
	Then the inverse $A^{-1}$ can be computed in $\AC^0_\F$.
\end{corollary}

\begin{proof}
	Write
	\[
		A = \begin{pmatrix}
			a_0 & a_1 & a_2 & \cdots & a_{n-1} \\
			0 & a_0 & a_1 & \cdots & a_{n-2} \\
			0 & 0 & a_0 & \cdots & a_{n-3} \\
			\vdots & \vdots & \vdots & \ddots & \vdots \\
			0 & 0 & 0 & \cdots & a_0
		\end{pmatrix}.
	\]
	Let $f(x) \coloneqq x^n + \sum_{i=0}^{n-1} a_i x^i$ and let $g(x) \coloneqq x^{n}$.
	Observe that the Sylvester matrix $\Syl(f,g)$ is a $2n \times 2n$ matrix that decomposes into $n \times n$ blocks as
	\[
		\Syl(f,g) =
		\begin{pmatrix}
			B & I_n \\
			A & 0
		\end{pmatrix},
	\]
	where $B$ is an $n \times n$ matrix corresponding to the coefficients of $f$ and $I_n$ is the $n \times n$ identity matrix.
	The inverse of $\Syl(f,g)$ has the block decomposition
	\[
		\Syl(f,g)^{-1} = 
		\begin{pmatrix}
			0 & A^{-1} \\
			I_n & -B A^{-1}
		\end{pmatrix}.
	\]
	In particular, the upper-right $n \times n$ block of $\Syl(f,g)^{-1}$ corresponds to $A^{-1}$.
	Thus, we can compute $A^{-1}$ in $\AC^0_\F$ by inverting $\Syl(f,g)$ using \cref{thm:sylvester adjugate}.
\end{proof}

\subsection{Inverting the B\'{e}zout Matrix}

In this subsection, we design an algorithm to compute the inverse of the B\'{e}zout matrix $\Bez_n(f,g)$ of two polynomials.
This is an easy corollary of the fact that we can compute the B\'{e}zout coefficients of coprime polynomials in $\AC^0_\F$.

\begin{theorem} \label{thm:bezout inverse}
	Let $\F$ be a field of characteristic zero or characteristic greater than $d$.
	Let $f, g \in \F[x]$ be univariate polynomials of degree at most $d$ given by their coefficients.
	Let $\delta = \max(\deg(f), \deg(g))$.
	Suppose that $\gcd(f,g) = 1$, so that the matrix $\Bez_\delta(f,g)$ is invertible.
	Then the inverse $\Bez_\delta(f,g)^{-1}$ can be computed in $\AC^0_\F$.
\end{theorem}

\begin{proof}
	Without loss of generality, suppose that $\delta = \deg(f)$.
	By \cref{prop:bezout inverse}, to invert $\Bez_\delta(f,g)$, it suffices to compute a polynomial $p \in \F[x]$ of degree at most $n-1$ such that $p(x) g(x) \equiv 1 \pmod{f(x)}$.
	The desired polynomial $p$ is precisely the B\'{e}zout coefficient of $g$, which we can compute in $\AC^0_\F$ using \cref{cor:bezout coprime}.
\end{proof}

\section{Operations on Roots} \label{sec:operations on roots}

In this section, we develop tools to manipulate the factors of univariate polynomials without having explicit access to the factorization of a polynomial.
These results are the technical highlight of this work and will be essential for our $\AC^0_\F$ algorithm for the GCD.

\subsection{Filtering}

All of the algorithms we have seen so far proceed by computing symmetric functions of the roots of two univariate polynomials $f$ and $g$.
To compute the GCD, we also need a means of comparing the (multisets of) roots of polynomials.
The algorithm of \cref{thm:resultant ac0} for the resultant tells us if $f$ and $g$ share a common root, but the resultant itself does not tell us any more about how the roots of $f$ and $g$ compare.
We can learn more by inspecting other elementary symmetric functions of the evaluations of $g$ at the roots of $f$ (and vice-versa).
Through a careful application of \cref{lem:esym as multiplication}, we can filter the common roots of $f$ and $g$ out of $f$ (and likewise, out of $g$), as we now illustrate.

Let 
\begin{align*}
	f(x) &\coloneqq \prod_{i=1}^n (x - \alpha_i) \prod_{i=1}^d (x - \gamma_i) \\
	g(x) &\coloneqq \prod_{i=1}^m (x - \beta_i) \prod_{i=1}^d (x - \gamma_i)
\end{align*}
be two squarefree polynomials, where $\alpha_1,\ldots,\alpha_n,\beta_1,\ldots,\beta_m,\gamma_1,\ldots,\gamma_d$ are $n+m+d$ distinct field elements.
Although the values of $n$, $m$, and $d$ are known to us for the purpose of analysis, these values are unknown to the algorithm.
Our goal is to compute $\prod_{i=1}^d (x - \gamma_i)$, which allows us to distinguish the shared roots of $f$ and $g$ from the set of all roots of $f$.
Note that this will allow us to compute the GCD when $f$ and $g$ are squarefree, which represents serious progress towards a general algorithm for the GCD!

Let $y$ be a fresh variable and consider the polynomial 
\[
	h(x,y) \coloneqq (y - x)\, g(x),
\]
viewed as an element of $\F(y)[x]$.
The evaluations of $h(x,y)$ at the roots of $f(x)$ are given by
\begin{align*}
	\lambda_1 &\coloneqq h(\alpha_1, y) = (y - \alpha_1)\, g(\alpha_1) \neq 0 \\
	&\vdots \\
	\lambda_n &\coloneqq h(\alpha_n, y) = (y - \alpha_n)\, g(\alpha_n) \neq 0 \\
	\lambda_{n + 1} &\coloneqq h(\gamma_1, y) = (y - \gamma_1)\, g(\gamma_1) = 0 \\
	&\vdots \\
	\lambda_{n + d} &\coloneqq h(\gamma_d, y) = (y - \gamma_d)\, g(\gamma_d) = 0.
\end{align*}
By taking $h$ to be a multiple of $g$, we ensure that $h(x,y)$ vanishes whenever we substitute a root of $g$ for $x$.
Applying \cref{lem:esym as multiplication}, we can compute
\[
	\prod_{i=1}^n \lambda_i = \prod_{i=1}^n (y - \alpha_i) g(\alpha_i)
\]
piecewise in $\AC^0_\F$.
Normalizing this polynomial to have leading coefficient 1 yields $\prod_{i=1}^n (y - \alpha_i)$, which is precisely the product of factors of $f$ that are not shared with $g$.
From here, we can obtain $\prod_{i=1}^d (y - \gamma_i)$ as the quotient
\[
	\prod_{i=1}^d (y - \gamma_i) = \frac{f(y)}{\prod_{i=1}^n (y - \alpha_i)},
\]
which can be computed in $\AC^0_\F$ using \cref{lem:exact division}.

When the polynomials $f$ and $g$ are not squarefree, the algorithm sketched above separates the factors of $f$ into two sets: those that correspond to roots of $g$, and those that are not roots of $g$.
The factors of $f$ maintain their multiplicity in the output of this algorithm.
We can use this to compute a polynomial $h$ such that $h$ and $\gcd(f,g)$ have the same set of irreducible factors, but possibly with different multiplicities.
Later in \cref{sec:gcd and lcm}, we will see how to compute the GCD with the correct multiplicities.

We now formalize the result of the preceding sketch.

\begin{lemma} \label{lem:filter}
	Let $\F$ be a field of characteristic zero or characteristic greater than $d$.
	Let $f, g \in \F[x]$ be univariate polynomials of degree at most $d$ given by their coefficients.
	Suppose that $f$ factors as $\prod_{i=1}^n (x - \alpha_i)^{a_i}$.
	Then the coefficients of the polynomials
	\begin{align*}
		f_{g=0}(x) & \coloneqq \prod_{i : g(\alpha_i) = 0} (x - \alpha_i)^{a_i} \\
		f_{g \neq 0}(x) & \coloneqq \prod_{i : g(\alpha_i) \neq 0} (x - \alpha_i)^{a_i}
	\end{align*}
	can be computed piecewise in $\AC^0_\F$.
\end{lemma}

\begin{proof}
	Let $y$ be a fresh variable.
	Define $h(x,y) \in \F[x,y]$ as
	\[
		h(x,y) \coloneqq (y - x)\, g(x).
	\]
	For a root $\alpha_i$ of $f$, the polynomial $y - \alpha_i$ is clearly nonzero, so $h(\alpha_i, y) = 0$ if and only if $g(\alpha_i) = 0$.
	By \cref{lem:esym as multiplication}, we can compute the polynomial
	\[
		r(y) \coloneqq \prod_{i : g(\alpha_i) \neq 0} h(\alpha_i, y)^{a_i} = \prod_{i : g(\alpha_i) \neq 0} (y - \alpha_i)^{a_i} g(\alpha_i)^{a_i}
	\]
	piecewise in $\AC^0_\F$.
	The leading coefficient of $r(y)$ is given by $\prod_{i : g(\alpha_i) \neq 0} g(\alpha_i)^{a_i}$, which we can compute piecewise in $\AC^0_\F$ using \cref{lem:leading coefficient}.
	By normalizing $r(y)$ to have leading coefficient 1, we obtain 
	\[
		\prod_{i : g(\alpha_i) \neq 0} (y - \alpha_i)^{a_i} = f_{g \neq 0}(y).
	\]
	We then compute $f_{g=0}(x)$ in $\AC^0_\F$ using the identity $f_{g=0}(x) = \frac{f(x)}{f_{g \neq 0}(x)}$ and \cref{lem:exact division}.
\end{proof}

The preceding lemma easily generalizes to the case where we have multiple polynomials $g_1,\ldots,g_m \in \F[x]$ and want to extract from $f$ all factors that correspond to common roots of $g_1,\ldots,g_m$.
As we saw in the proof of \cref{lem:filter}, multiplying $(y - x)$ by $g(x)$ suppressed the common roots of $f$ and $g$.
To handle multiple polynomials $g_1,\ldots,g_m$, we need to find a polynomial in $x$ whose roots are exactly the common roots of the $g_i$.
If we add a fresh variable $z$, then this is easy to do with the polynomial
\[
	g(x,z) \coloneqq g_1(x) + z g_2(x) + \cdots + z^{m-1} g_m(x).
\]
We now extend \cref{lem:filter} to filter out the common roots of multiple polynomials from a given polynomial $f$.
The proof is a straightforward generalization of the proof of \cref{lem:filter} using $g(x,z)$ in place of $g(x)$.

\begin{lemma} \label{lem:multifilter}
	Let $\F$ be a field of characteristic zero or characteristic greater than $d$.
	Let $f, g_1, \ldots, g_m \in \F[x]$ be univariate polynomials of degree at most $d$ given by their coefficients.
	Let 
	\[
		V \coloneqq \set{\alpha \in \F : g_1(\alpha) = \cdots = g_m(\alpha) = 0}
	\]
	be the set of common roots of $g_1,\ldots,g_m$.
	Suppose that $f$ factors as $\prod_{i=1}^n (x - \alpha_i)^{a_i}$.
	Then the coefficients of the polynomials
	\begin{align*}
		f_{\in V}(x) &\coloneqq \prod_{i : \alpha_i \in V} (x - \alpha_i)^{a_i} \\
		f_{\notin V}(x) &\coloneqq \prod_{i : \alpha_i \notin V} (x - \alpha_i)^{a_i}
	\end{align*}
	can be computed piecewise in $\AC^0_\F$.
\end{lemma}

\begin{proof}
	Let $y$ and $z$ be fresh variables.
	Define $g(x,z) \in \F[x,z]$ as
	\[
		g(x,z) \coloneqq \sum_{i=1}^m z^{i-1} g_i(x)
	\]
	and define $h(x,y,z) \in \F[x,y,z]$ by
	\[
		h(x,y,z) \coloneqq (y - x) g(x,z).
	\]
	For a root $\alpha_i$ of $f$, the polynomial $y - \alpha_i$ is nonzero, so $h(\alpha_i, y, z) = 0$ if and only if $g(\alpha_i,z) = 0$, which occurs exactly when $g_1(\alpha_i) = \cdots = g_m(\alpha_i) = 0$.
	By \cref{lem:esym as multiplication}, we can compute the polynomial
	\[
		r(y,z) \coloneqq \prod_{i : \alpha_i \notin V} h(\alpha_i, y, z)^{a_i} = \prod_{i : \alpha_i \notin V} (y - \alpha_i)^{a_i} g(\alpha_i, z)^{a_i}
	\]
	piecewise in $\AC^0_\F$.

	View $r(y,z)$ as a polynomial in $y$ with coefficients in $\F[z]$.
	The leading coefficient of the term $(y - \alpha_i)^{a_i} g(\alpha_i, z)^{a_i}$ is $g(\alpha_i, z)^{a_i}$.
	Leading coefficients are homomorphic with respect to multiplication, so the leading coefficient of $r(y,z)$ (as a polynomial in $\F[z][y]$) is given by
	\[
		s(z) \coloneqq \prod_{i : \alpha_i \notin V} g(\alpha_i, z)^{a_i}.
	\]
	Because $\deg(r(y,z)) \le dm$, we can compute the leading coefficient of $r(y,z)$ piecewise in $\AC^0_\F$ via \cref{lem:leading coefficient}.
	Dividing $r(y,z)$ by $s(z)$ yields
	\[
		\frac{r(y,z)}{s(z)} = \prod_{i : \alpha_i \notin V} (y - \alpha_i)^{a_i} = f_{\notin V}(y),
	\]
	which we can computed in $\AC^0_\F$ using \cref{lem:exact division}.
	We can then use the identity $f_{\in V}(x) = \frac{f(x)}{f_{\notin V}(x)}$ and \cref{lem:exact division} to compute $f_{\in V}$ in $\AC^0_\F$.
\end{proof}

\subsection{Thresholding} \label{subsec:threshold}

In the previous subsection, we saw how to distinguish the common roots of two polynomials from the roots of a single polynomial.
To compute the GCD, we also need to distinguish between roots of a polynomial of different multiplicities.
As we shall see, this can be done as a straightforward application of \cref{lem:multifilter}.

Let $f, g \in \F[x]$ be univariate polynomials and suppose that $f$ and $g$ factor as
\begin{align*}
	f(x) &= \prod_{i=1}^n (x - \alpha_i)^{a_i} \\
	g(x) &= \prod_{i=1}^n (x - \alpha_i)^{b_i}.
\end{align*}
Here, we have switched notation from previous sections, and now take $\alpha_1,\ldots,\alpha_n \in \F$ to be the union of the roots of $f$ and $g$, as this will be convenient for us here.
Because the $\alpha_i$ are the union of the roots of $f$ and $g$, some of the $a_i$ or $b_i$ may be zero. 

Recall that a point $\alpha \in \F$ is a root of $f$ of multiplicity $r$ if and only if $f$ and its first $r-1$ derivatives vanish at $\alpha$, i.e.,
\[
	f(\alpha) = f'(\alpha) = \cdots = f^{(r-1)}(\alpha) = 0.
\]
This suggests that we can distinguish between factors of $f$ of different multiplicities by invoking \cref{lem:multifilter} with $f, f', \ldots, f^{(r-1)}$ as the filter polynomials.
More generally, we can filter the roots of $f$ by their multiplicity in $g$.

For a parameter $r \in \naturals$, let $f_{g < r}$ be the polynomial defined as 
\[
	f_{g < r}(x) \coloneqq \prod_{i : b_i < r} (x - \alpha_i)^{a_i}.
\]
That is, $f_{g < r}$ selects from $f$ the factors $(x - \alpha_i)$ that occur with multiplicity less than $r$ in $g$, but preserves the multiplicity $a_i$ with which $(x - \alpha_i)$ occurs in $f$.
The next lemma describes an $\AC^0_\F$ algorithm that piecewise computes the coefficients of $f_{g < r}$ given the coefficients of $f$ and $g$ and the value of $r$.

\begin{lemma} \label{lem:threshold}
	Let $\F$ be a field of characteristic zero or characteristic greater than $d$.
	Let $f, g \in \F[x]$ be univariate polynomials of degree at most $d$ given by their coefficients.
	Suppose that $f$ and $g$ factor as $f(x) = \prod_{i=1}^n (x - \alpha_i)^{a_i}$ and $g(x) = \prod_{i=1}^n (x - \alpha_i)^{b_i}$, where $\alpha_i \in \F$ and $a_i, b_i \in \naturals$.
	Let $r \in \naturals$.
	Then the coefficients of the polynomials 
	\begin{align*}
		f_{g \ge r}(x) \coloneqq \prod_{i : b_i \ge r} (x - \alpha_i)^{a_i} \\
		f_{g < r}(x) \coloneqq \prod_{i : b_i < r} (x - \alpha_i)^{a_i}
	\end{align*}
	can be computed piecewise in $\AC^0_\F$.
\end{lemma}

\begin{proof}
	Note that if $r > \deg(g)$, then $f_{g \ge r} = 1$ and $f_{g < r} = f$, so we may assume without loss of generality that $r \le \deg(g) < \ch(\F)$.
	Define $V_{\ge r} \subseteq \F$ to be the roots of $g$ of multiplicity at least $r$.
	Because $r < \ch(\F)$, we can express $V_{\ge r}$ as
	\begin{align*}
		V_{\ge r} &\coloneqq \set{\alpha \in \F : (x - \alpha)^r \mid g(x)} \\
		&= \set{\alpha \in \F : g(\alpha) = g^{(1)}(\alpha) = \cdots = g^{(r-1)}(\alpha) = 0}.
	\end{align*}
	Note that $f_{g \ge r}(x)$ and $f_{g < r}(x)$ satisfy
	\begin{align*}
		f_{g \ge r}(x) &= \prod_{i : \alpha_i \in V_{\ge r}} (x - \alpha_i)^{a_i} \\
		f_{g < r}(x) &= \prod_{i : \alpha_i \notin V_{\ge r}} (x - \alpha_i)^{a_i}.
	\end{align*}
	We can compute the coefficients of $g^{(1)}(x), \ldots, g^{(r-1)}(x)$ in $\AC^0_\F$.
	A direct application of \cref{lem:multifilter} computes the coefficients of $f_{g \ge r}(x)$ and $f_{g < r}(x)$ piecewise in $\AC^0_\F$.
\end{proof}

Suppose we are instead given multiple polynomials $f, g_1, \ldots, g_m \in \F[x]$ that factor as 
\begin{align*}
	f(x) &= \prod_{i=1}^n (x - \alpha_i)^{a_i} \\
	g_1(x) &= \prod_{i=1}^n (x - \alpha_i)^{b_{1,i}} \\
	&\vdots \\
	g_m(x) &= \prod_{i=1}^n (x - \alpha_i)^{b_{m,i}}.
\end{align*}
\cref{lem:threshold} allows us to extract factors from $f$ based on their multiplicity in one of the $g_i$.
Just as \cref{lem:multifilter} generalized \cref{lem:filter} to allow for multiple filter polynomials, we can generalize \cref{lem:threshold} to filter the roots of $f$ according to their multiplicities in multiple polynomials, as the following lemma shows.

\begin{lemma} \label{lem:multithreshold}
	Let $\F$ be a field of characteristic zero or characteristic greater than $d$.
	Let $f, g_1, \ldots, g_m \in \F[x]$ be univariate polynomials of degree at most $d$ given by their coefficients.
	Suppose that $f$ and $g_j$ factor as $f(x) = \prod_{i=1}^n (x - \alpha_i)^{a_i}$ and $g_j(x) = \prod_{i=1}^n (x - \alpha_i)^{b_{j,i}}$, where $\alpha_i \in \F$ and $a_i, b_{j,i} \in \naturals$.
	Let $r_1,\ldots,r_m \in \naturals$.
	Then the coefficients of the polynomials 
	\begin{align*}
		f_{\vec{g} \ge \vec{r}}(x) \coloneqq \prod_{i : b_{1,i} \ge r_1 \land \cdots \land b_{m,i} \ge r_m} (x - \alpha_i)^{a_i} \\
		f_{\vec{g} \not\ge \vec{r}}(x) \coloneqq \prod_{i : b_{1,i} < r_1 \lor \cdots \lor b_{m,i} < r_m} (x - \alpha_i)^{a_i}
	\end{align*}
	can be computed piecewise in $\AC^0_\F$.
\end{lemma}

\begin{proof}
	Note that if $r_i > \deg(g_i)$ for some $i \in [m]$, then $f_{\vec{g} \ge \vec{r}} = 1$ and $f_{\vec{g} \not \ge \vec{r}} = f$, so we may assume without loss of generality that $r_i \le \deg(g_i) < \ch(\F)$ for all $i \in [m]$.

	Define $V_{\ge \vec{r}} \subseteq \F$ to be the set of $\alpha \in \F$ that occur with multiplicity at least $r_i$ in $g_i$ for all $i \in [m]$.
	Because $r_i < \ch(\F)$ for all $i \in [m]$, we can write $V_{\ge \vec{r}}$ as
	\begin{align*}
		V_{\ge \vec{r}} &\coloneqq \bigcap_{i=1}^m \set{\alpha \in \F : (x - \alpha)^{r_i} \mid g_i(x)} \\
		&= \bigcap_{i=1}^m \set{\alpha \in \F : g_i(\alpha) = g_i^{(1)}(\alpha) = \cdots = g_i^{(r_i-1)}(\alpha) = 0} \\
		&= \set{\alpha \in \F : g_i^{(j)}(\alpha) = 0 \ \forall i \in [m], j \in \set{0,1,\ldots,r_i-1}}.
	\end{align*}
	Note that $f_{\vec{g} \ge \vec{r}}(x)$ and $f_{\vec{g} \not\ge \vec{r}}(x)$ satisfy
	\begin{align*}
		f_{\vec{g} \ge \vec{r}}(x) &= \prod_{i : \alpha_i \in V_{\ge \vec{r}}} (x - \alpha_i)^{a_i} \\
		f_{\vec{g} \not\ge \vec{r}}(x) &= \prod_{i : \alpha_i \notin V_{\ge \vec{r}}} (x - \alpha_i)^{a_i}.
	\end{align*}
	We can compute the coefficients of $g_i^{(j)}(x)$ for all $i \in [m]$ and $j \in \set{0,1,\ldots,r_i-1}$ in $\AC^0_\F$.
	A direct application of \cref{lem:multifilter} computes the coefficients of $f_{\vec{g} \ge \vec{r}}(x)$ and $f_{\vec{g} \not\ge \vec{r}}(x)$ piecewise in $\AC^0_\F$.
\end{proof}

\begin{remark}
	Although \cref{lem:multithreshold} holds for an arbitrary threshold vector $\vec{r} \in \naturals^m$, our applications only make use of the case where all entries of $\vec{r}$ are equal.
\end{remark}

\subsection{Squarefree Decomposition} \label{subsec:squarefree}

We conclude this section with an application of \cref{lem:threshold} to computing the squarefree decomposition of a univariate polynomial $f \in \F[x]$.
Typically, the squarefree decomposition is computed by reduction to the GCD.
If $f$ factors as $f = \prod_{i=1}^n (x - \alpha_i)^{a_i}$, then it is easy to show that
\[
	\gcd(f, f', \ldots, f^{(r)}) = \prod_{i=1}^n (x - \alpha_i)^{\max(a_i - r, 0)}.
\]
Denoting the squarefree decomposition of $f$ by $(f_1,\ldots,f_m)$, when $r < \deg(f)$, one can write $f_r$ as
\[
	f_r = \frac{\gcd(f, f', \ldots, f^{(r-1)}) \cdot \gcd(f, f', \ldots, f^{(r+1)})}{\gcd(f, f', \ldots, f^{(r)})^2},
\]
and otherwise $f_{\deg(f)} = \gcd(f, f', \ldots, f^{(\deg(f)-1)})$.
This reduction, combined with an $\NC^2_\F$ algorithm for the GCD of many polynomials, leads to the $\NC^2_\F$ algorithm of \textcite{vonzurGathen84} for the squarefree decomposition.
In his survey, \textcite{vonzurGathen86survey} asked if there is an $\NC^1_\F$ algorithm to compute the squarefree decomposition.
As an application of \cref{lem:threshold}, we will show that the squarefree decomposition can in fact be computed piecewise in $\AC^0_\F$.

We note that while the squarefree decomposition is usually obtained by reducing to the GCD, this is not the route we take.
Our techniques seem to naturally proceed in the reverse direction, instead computing the GCD by reducing to the squarefree decomposition.

\begin{lemma} \label{lem:squarefree decomp}
	Let $\F$ be a field of characteristic zero or characteristic greater than $d$.
	Let $f \in \F[x]$ be a univariate polynomial of degree $d$ given by its coefficients.
	Then the squarefree decomposition of $f$ can be computed piecewise in $\AC^0_\F$.
\end{lemma}

\begin{proof}
	Suppose that $f$ factors as
	\[
		f(x) = \prod_{i=1}^n (x - \alpha_i)^{a_i}.
	\]
	Let $d \coloneqq \deg(f)$.
	For all $r \in [d]$, we can compute 
	\[
		f_{\le r}(x) \coloneqq \prod_{i : a_i \le r}(x - \alpha_i)^{a_i}
	\]
	in $\AC^0_\F$ piecewise by applying \cref{lem:threshold} to $f$, using $f$ itself as the threshold polynomial with threshold $r + 1$.
	We then compute, for all $r \in [d]$, the polynomials
	\[
		f_{= r}(x) \coloneqq \prod_{i : a_i = r} (x - \alpha_i)^{a_i} = \frac{f_{\le r}(x)}{f_{\le r - 1}(x)}
	\]
	in $\AC^0_\F$ using \cref{lem:exact division}.
	Observe that if $(f_1,\ldots,f_d)$ is the squarefree decomposition of $f$, then the $f_i$ satisfy $f_r(x)^r = f_{= r}(x)$.
	We can obtain $f_r(x)$ in $\AC^0_\F$ by applying \cref{lem:perfect power} to $f_{=r}(x)$.
\end{proof}

As a corollary of \cref{lem:squarefree decomp}, we can also compute the squarefree part of a univariate polynomial $f \in \F[x]$ piecewise in $\AC^0_\F$.

\begin{corollary} \label{cor:squarefree part}
	Let $\F$ be a field of characteristic zero or characteristic greater than $d$.
	Let $f \in \F[x]$ be a univariate polynomial of degree $d$ given by its coefficients.
	Then the squarefree part of $f$ can be computed piecewise in $\AC^0_\F$.
\end{corollary}

\begin{proof}
	By \cref{lem:squarefree decomp}, the squarefree decomposition $(f_1,\ldots,f_m)$ of $f$ can be computed piecewise in $\AC^0_\F$.
	The squarefree part of $f$ is given by $\prod_{i=1}^m f_i$.
	We can compute the coefficients of $\prod_{i=1}^m f_i$ in $\AC^0_\F$ by \cref{lem:polynomial interpolation}.
\end{proof}

\section{Greatest Common Divisor and Least Common Multiple} \label{sec:gcd and lcm}

In this section, we describe piecewise $\AC^0_\F$ algorithms to compute the greatest common divisor and least common multiple of $m$ univariate polynomials $f_1,\ldots,f_m \in \F[x]$.
We also describe a piecewise $\AC^0_\F$ algorithm to compute the B\'{e}zout coefficients of a pair of polynomials $f_1, f_2 \in \F[x]$.
Our algorithms for the greatest common divisor and least common multiple turn out to be straightforward applications of \cref{lem:multithreshold} and \cref{cor:squarefree part}.
To compute the B\'{e}zout coefficients, we reduce to the case where $f_1$ and $f_2$ are coprime by dividing out $\gcd(f_1, f_2)$ using \cref{lem:exact division}, at which point we can apply \cref{cor:bezout coprime}.

\subsection{Two Polynomials}

In this subsection, we compute the greatest common divisor, least common multiple, and B\'{e}zout coefficients of a pair of polynomials, all in (piecewise) $\AC^0_\F$.
We start with the greatest common divisor.

\begin{theorem} \label{thm:gcd}
	Let $\F$ be a field of characteristic zero or characteristic greater than $2d$.
	Let $f, g \in \F[x]$ be univariate polynomials of degree at most $d$ given by their coefficients.
	Then their greatest common divisor $\gcd(f,g)$ can be computed piecewise in $\AC^0_\F$.
\end{theorem}

\begin{proof}
	Suppose that $f$ and $g$ factor as 
	\begin{align*}
		f(x) &= \prod_{i=1}^n (x - \alpha_i)^{a_i} \\
		g(x) &= \prod_{i=1}^n (x - \alpha_i)^{b_i},
	\end{align*}
	where $\alpha_1,\ldots,\alpha_n \in \F$ and $a_i, b_i \in \naturals$.
	Note that some of the $a_i$ and $b_i$ may be zero.
	Let $s(x)$ be the squarefree part of $f(x) g(x)$, i.e.,
	\[
		s(x) \coloneqq \prod_{i=1}^n (x - \alpha_i).
	\]
	We may compute $s(x)$ piecewise in $\AC^0_\F$ using \cref{cor:squarefree part}.

	For a parameter $r \in \naturals$, let
	\[
		s_{\ge r}(x) \coloneqq \prod_{i : a_i \ge r \land b_i \ge r} (x - \alpha_i).
	\]
	We can compute $s_{\ge r}(x)$ piecewise in $\AC^0_\F$ for all $1 \le r \le \min(\deg(f), \deg(g))$ by applying \cref{lem:multithreshold} to $s(x)$, $f(x)$, and $g(x)$.
	Observe that
	\begin{align*}
		\gcd(f(x), g(x)) &= \prod_{i=1}^n (x - \alpha_i)^{\min(a_i, b_i)} \\
		&= \prod_{r=1}^{\min(\deg(f), \deg(g))} \prod_{i : \min(a_i, b_i) \ge r} (x - \alpha_i) \\
		&= \prod_{r=1}^{\min(\deg(f), \deg(g))} s_{\ge r}(x),
	\end{align*}
	so we can compute $\gcd(f,g)$ piecewise in $\AC^0_\F$.
	We can then interpolate the coefficients of $\gcd(f,g)$ in $\AC^0_\F$ using \cref{lem:polynomial interpolation}.
\end{proof}

As an easy corollary of \cref{thm:gcd}, we obtain an $\AC^0_\F$ algorithm to compute the least common multiple of two polynomials.

\begin{corollary} \label{cor:lcm}
	Let $\F$ be a field of characteristic zero or characteristic greater than $2d$.
	Let $f, g \in \F[x]$ be univariate polynomials of degree at most $d$ given by their coefficients.
	Then $\lcm(f,g)$ can be computed in piecewise $\AC^0_\F$.
\end{corollary}

\begin{proof}
	We can compute $\gcd(f(x), g(x))$ piecewise in $\AC^0_\F$ using \cref{thm:gcd}.
	The coefficients of the product $f(x) g(x)$ can be computed in $\AC^0_\F$ using \cref{lem:polynomial interpolation}.
	Using the identity
	\[
		\lcm(f(x), g(x)) = \frac{f(x) g(x)}{\gcd(f(x), g(x))},
	\]
	we can compute $\lcm(f(x), g(x))$ piecewise in $\AC^0_\F$ using \cref{lem:exact division}.
\end{proof}

By combining \cref{thm:gcd} with \cref{cor:bezout coprime}, we obtain an $\AC^0_\F$ algorithm to compute the B\'{e}zout coefficients of any pair of polynomials.

\begin{corollary} \label{cor:bezout general}
	Let $\F$ be a field of characteristic zero or characteristic greater than $2d$.
	Let $f, g \in \F[x]$ be univariate polynomials of degree at most $d$ given by their coefficients.
	Then the B\'{e}zout coefficients of $f$ and $g$ can be computed piecewise in $\AC^0_\F$.
\end{corollary}

\begin{proof}
	Using \cref{thm:gcd} and \cref{lem:exact division}, we can compute $\hat{f}(x) \coloneqq \frac{f(x)}{\gcd(f(x), g(x))}$ and $\hat{g}(x) \coloneqq \frac{g(x)}{\gcd(f(x), g(x))}$ piecewise in $\AC^0_\F$.
	By construction, we have $\gcd(\hat{f}(x), \hat{g}(x)) = 1$.
	Observe that the B\'{e}zout coefficients of $\hat{f}(x)$ and $\hat{g}(x)$ are equal to those of $f(x)$ and $g(x)$.
	We may compute the B\'{e}zout coefficients of $\hat{f}(x)$ and $\hat{g}(x)$ in $\AC^0_\F$ using \cref{cor:bezout coprime}.
\end{proof}

\subsection{Multiple Polynomials}

We now extend the results of the previous subsection to compute greatest common divisors and least common multiples of many polynomials.
The proof of \cref{thm:gcd} easily generalizes to compute the GCD of many polynomials.

\begin{theorem} \label{thm:multigcd}
	Let $\F$ be a field of characteristic zero or characteristic greater than $md$.
	Let $f_1, \ldots, f_m \in \F[x]$ be univariate polynomials of degree at most $d$ given by their coefficients.
	Then their greatest common divisor $\gcd(f_1,\ldots,f_m)$ can be computed piecewise in $\AC^0_\F$.
\end{theorem}

\begin{proof}
	Suppose that $f_1, \ldots, f_m$ factor as
	\begin{align*}
		f_1(x) &= \prod_{i=1}^n (x - \alpha_i)^{a_{1,i}} \\
		&\vdots \\
		f_m(x) &= \prod_{i=1}^n (x - \alpha_i)^{a_{m,i}},
	\end{align*}
	where $\alpha_1,\ldots,\alpha_n \in \F$ and $a_{j,i} \in \naturals$.
	Note that some of the $a_{j,i}$ may be zero.
	Let $s(x)$ be the squarefree part of $f_1(x) \cdots f_m(x)$, i.e., 
	\[
		s(x) \coloneqq \prod_{i=1}^n (x - \alpha_i).
	\]
	We may compute $s(x)$ piecewise in $\AC^0_\F$ using \cref{cor:squarefree part}.

	For a parameter $r \in \naturals$, let
	\[
		s_{\ge r}(x) \coloneqq \prod_{i : a_{1,i} \ge r \land \cdots \land a_{m,i} \ge r} (x - \alpha_i).
	\]
	Let $\delta \coloneqq \min_{i \in [m]} \deg(f_i)$.
	We can compute the coefficients of $s_{\ge r}(x)$ in $\AC^0_\F$ for all $1 \le r \le \delta$ by applying \cref{lem:multithreshold} to $s(x)$ and $f_1(x), \ldots, f_m(x)$.

	Observe that
	\begin{align*}
		\gcd(f_1(x), \ldots, f_m(x)) &= \prod_{i=1}^n (x - \alpha_i)^{\min(a_{1,i}, \ldots, a_{m,i})} \\
		&= \prod_{r=1}^{\delta} \prod_{i : \min(a_{1,i}, \ldots, a_{m,i}) \ge r} (x - \alpha_i) \\
		&= \prod_{r=1}^{\delta} s_{\ge r}(x),
	\end{align*}
	so we can compute $\gcd(f_1(x), \ldots, f_m(x))$ piecewise in $\AC^0_\F$.
	We can then interpolate the coefficients of $\gcd(f_1(x), \ldots, f_m(x))$ in $\AC^0_\F$ using \cref{lem:polynomial interpolation}.
\end{proof}

Using \cref{thm:multigcd}, we can also compute the least common multiple of multiple polynomials $f_1(x), \ldots, f_m(x)$ in $\AC^0_\F$.
Although the identity $f(x) g(x) = \gcd(f(x), g(x)) \lcm(f(x), g(x))$ no longer holds for $m \ge 3$ polynomials, a suitable analogue of this identity will allow us to reduce the task of computing $\lcm(f_1,\ldots,f_m)$ to computing $\gcd(f_1,\ldots,f_m)$.
The following elementary lemma relates the GCD and LCM of $m \ge 3$ polynomials.

\begin{lemma} \label{lem:multilcm to multigcd}
	Let $\F$ be any field.
	Let $f_1, \ldots, f_m \in \F[x]$ be univariate polynomials.
	Let $p \in \F[x]$ be a polynomial such that $f_i \mid p$ for all $i \in [m]$.
	Then
	\[
		\lcm(f_1,\ldots,f_m) = \frac{p}{\gcd \del{\frac{p}{f_1}, \ldots, \frac{p}{f_m}}}.
	\]
\end{lemma}

Using \cref{lem:multilcm to multigcd}, we now compute the LCM of many polynomials in $\AC^0_\F$.

\begin{corollary} \label{cor:multilcm}
	Let $\F$ be a field of characteristic zero or characteristic greater than $m^2 d$.
	Let $f_1, \ldots, f_m \in \F[x]$ be univariate polynomials of degree at most $d$ given by their coefficients.
	Then their least common multiple $\lcm(f_1,\ldots,f_m)$ can be computed piecewise in $\AC^0_\F$.
\end{corollary}

\begin{proof}
	By \cref{lem:multilcm to multigcd}, we have
	\[
		\lcm(f_1,\ldots,f_m) = \frac{f_1 \cdots f_m}{\gcd\del{\frac{f_1 \cdots f_m}{f_1}, \ldots, \frac{f_1 \cdots f_m}{f_m}}}.
	\]
	The coefficients of $p \coloneqq f_1 \cdots f_m$ can be computed in $\AC^0_\F$ using \cref{lem:polynomial interpolation}.
	Using \cref{lem:exact division}, we can compute the coefficients of $p / f_1, \ldots, p / f_m$ in $\AC^0_\F$.
	It is clear that $\deg(p / f_i) \le m d$.
	Because $\ch(\F) > m^2 d$, we can compute $\gcd\del{\frac{p}{f_1}, \ldots, \frac{p}{f_m}}$ piecewise in $\AC^0_\F$ using \cref{thm:multigcd}.
	We can then compute the LCM by taking the quotient as above, which can be done in $\AC^0_\F$ using \cref{lem:exact division}.
\end{proof}

\section{Arbitrary Functions of Root Multiplicities} \label{sec:arbitrary functions}

The GCD and LCM manipulate the multiplicities of the roots of two polynomials $f$ and $g$ in specific ways: the GCD assigns a root $\alpha$ its minimum multiplicity in $f$ and $g$, while the LCM assigns $\alpha$ its maximum multiplicity.
What other functions can we apply to the root multiplicities in $\AC^0_\F$?
We discuss this first in the case of two polynomials, where \emph{any} function of the root multiplicities can be implemented in $\AC^0_\F$, and then handle the case of many polynomials, where again any function (reasonably encoded) can be implemented in $\AC^0_\F$.

\subsection{Two Polynomials}

We now consider a generalization of the problem of computing greatest common divisors and least common multiples.
Suppose we have two polynomials $f, g \in \F[x]$, which factor as 
\begin{align*}
	f(x) &= \prod_{i=1}^n (x - \alpha_i)^{a_i} \\
	g(x) &= \prod_{i=1}^n (x - \alpha_i)^{b_i},
\end{align*}
where $\alpha_1, \ldots, \alpha_n \in \F$ are the union of the sets of roots of $f$ and $g$ and $a_i, b_i \in \naturals$, where some of the $a_i$ or $b_i$ may be zero.
The greatest common divisor and least common multiple are given by
\begin{align*}
	\gcd(f(x), g(x)) &= \prod_{i=1}^n (x - \alpha_i)^{\min(a_i, b_i)} \\
	\lcm(f(x), g(x)) &= \prod_{i=1}^n (x - \alpha_i)^{\max(a_i, b_i)}.
\end{align*}
That is, to form the greatest common divisor, if a factor $(x - \alpha_i)$ appears with multiplicity $a_i$ in $f$ and multiplicity $b_i$ in $g$, the same factor should appear with multiplicity $\min(a_i, b_i)$ in $\gcd(f, g)$.
The least common multiple replaces $\min(a_i, b_i)$ with $\max(a_i, b_i)$.
As we saw in \cref{sec:gcd and lcm}, we can implement these minimum and maximum operations without directly factorizing $f$ and $g$.

What other functions of these multiplicities can we implement in (piecewise) $\AC^0_\F$?
Addition is easy, as we have
\[
	f(x) \cdot g(x) = \prod_{i=1}^n (x - \alpha_i)^{a_i + b_i}.
\]
What about products of multiplicities?
That is, can we efficiently compute the coefficients of the polynomial
\[
	(f \diamond g)(x) \coloneqq \prod_{i=1}^n (x - \alpha_i)^{a_i b_i}?
\]
It is not obvious that the coefficients of $f \diamond g$ can even be expressed algebraically in terms of the coefficients of $f$ and $g$, much less that there is an efficient algorithm for this problem.
However, a modification of our GCD algorithm from \cref{thm:gcd} enables us to compute $f \diamond g$, and indeed solve a more general problem.

Suppose we are given an \emph{arbitrary} function $P : \naturals \times \naturals \to \naturals$ and we would like to compute the polynomial
\[
	(f \diamond_P g)(x) \coloneqq \prod_{i = 1}^n (x - \alpha_i)^{P(a_i, b_i)}.
\]
By combining our algorithms for the squarefree decomposition and the GCD, we can design a piecewise $\AC^0_\F$ algorithm for $f \diamond_P g$ in a rather straightforward manner.

We start with the case where the function $P : \naturals \times \naturals \to \naturals$ is a delta function.
That is, we consider functions of the form $\delta_{i,j} : \naturals \times \naturals \to \naturals$ given by
\[
	\delta_{i,j}(a,b) \coloneqq \begin{cases}
		1 & \text{if $a = i$ and $b = j$} \\
		0 & \text{otherwise.}
	\end{cases}
\]
Every function $P : \naturals \times \naturals \to \naturals$ can be expressed as a sum of delta functions, so handling delta functions will allow us to easily handle arbitrary functions.

\begin{lemma} \label{lem:delta function}
	Let $\F$ be a field of characteristic zero or characteristic greater than $2d$.
	Let $f, g \in \F[x]$ be univariate polynomials of degree at most $d$ given by their coefficients.
	Suppose that $f$ and $g$ factor as 
	\begin{align*}
		f(x) &= \prod_{i=1}^n (x - \alpha_i)^{a_i} \\
		g(x) &= \prod_{i=1}^n (x - \alpha_i)^{b_i},
	\end{align*}
	where $\alpha_1, \ldots, \alpha_n \in \F$ and $a_i, b_i \in \naturals$, where some of the $a_i$ or $b_i$ may be zero.
	Then the coefficients of the polynomial 
	\[
		(f \diamond_{\delta_{i,j}} g)(x) \coloneqq \prod_{k : a_k = i, b_k = j} (x - \alpha_k)
	\]
	can be computed piecewise in $\AC^0_\F$.
\end{lemma}

\begin{proof}
	Without loss of generality, let $d = \max(\deg(f), \deg(g))$.
	Let $(f_1,\ldots,f_d)$ and $(g_1,\ldots,g_d)$ be the squarefree decompositions of $f$ and $g$, respectively.
	Recall that, by definition of the squarefree decomposition, we have $f = \prod_i f_i^i$, each polynomial $f_i$ is squarefree, and $\gcd(f_i, f_j) = 1$ when $i \neq j$, and likewise for $g$.
	For notational ease, let $s$ be the squarefree part of $fg$ and define $f_0 \coloneqq s/f$ and $g_0 \coloneqq s/g$.

	Observe that for $i, j \in \set{0, 1, \ldots, d}$, we have
	\[
		\gcd(f_i(x), g_j(x)) = \prod_{k : a_k = i, b_k = j} (x - \alpha_k) = (f \diamond_{\delta_{i,j}} g)(x).
	\]
	Thus, we can compute $f \diamond_{\delta_{i,j}} g$ piecewise in $\AC^0_\F$ by first computing the squarefree decompositions of $f$, $g$, and $fg$ using \cref{lem:squarefree decomp}, and then computing $\gcd(f_i, g_j)$ via \cref{thm:gcd}.
\end{proof}

We now compute $f \diamond_P g$ for an arbitrary function $P : \naturals \times \naturals \to \naturals$ by writing $P$ as a sum of delta functions.

\begin{theorem} \label{thm:arbitrary function}
	Let $\F$ be a field of characteristic zero or characteristic greater than $2d$.
	Let $f, g \in \F[x]$ be univariate polynomials of degree at most $d$ given by their coefficients.
	Let $P : \naturals \times \naturals \to \naturals$ be an arbitrary function.
	Suppose that $f$ and $g$ factor as 
	\begin{align*}
		f(x) &= \prod_{i=1}^n (x - \alpha_i)^{a_i} \\
		g(x) &= \prod_{i=1}^n (x - \alpha_i)^{b_i},
	\end{align*}
	where $\alpha_1, \ldots, \alpha_n \in \F$ and $a_i, b_i \in \naturals$, where some of the $a_i$ or $b_i$ may be zero.
	Then the coefficients of the polynomial 
	\[
		(f \diamond_P g)(x) \coloneqq \prod_{i=1}^n (x - \alpha_i)^{P(a_i, b_i)}
	\]
	can be computed piecewise in $\AC^0_\F$, where the size of the circuit is measured with respect to both $d$ and the maximum possible degree of $f \diamond_P g$.
\end{theorem}

\begin{proof}
	Using \cref{lem:delta function}, we can compute
	\[
		(f \diamond_{\delta_{i,j}} g)(x) = \prod_{k : a_k = i, b_k = j} (x - \alpha_i)
	\]
	piecewise in $\AC^0_\F$.
	We can then write $f \diamond_P g$ as
	\begin{align*}
		(f \diamond_P g)(x) &= \prod_{i=0}^d \prod_{j=0}^d \prod_{k : a_k = i, b_k = j} (x - \alpha_k)^{P(i,j)} \\
		&= \prod_{i=0}^d \prod_{j=0}^d (f \diamond_{\delta_{i,j}} g)(x)^{P(i,j)}.
	\end{align*}
	This results in a piecewise $\AC^0_\F$ circuit that computes $f \diamond_P g$.
	Applying \cref{lem:polynomial interpolation} yields a piecewise $\AC^0_\F$ circuit that computes the coefficients of $f \diamond_P g$.
\end{proof}

\subsection{Multiple Polynomials}

Just as the greatest common divisor and least common multiple can be defined for a set of more than two polynomials, we can extend the setting and result of \cref{thm:arbitrary function} to more polynomials.
Suppose we are given polynomials $f_1, \ldots, f_m \in \F[x]$ where $f_i$ factors as
\[
	f_i(x) = \prod_{j=1}^n (x - \alpha_j)^{a_{i,j}},
\]
where $\alpha_1,\ldots,\alpha_n \in \F$ are the union of the set of roots of $f_1,\ldots,f_m$.
For a function $P : \naturals^m \to \naturals$, we can define
\[
	\diamond_P(f_1,\ldots,f_m) \coloneqq \prod_{j=1}^n (x - \alpha_j)^{P(a_{1,j}, \ldots, a_{m,j})}.
\]
As we will see, an appropriate generalization of \cref{thm:arbitrary function} allows us to compute $\diamond_P(f_1,\ldots,f_m)$ in $\AC^0_\F$.

Here, we have two results, depending on how the function $P : \naturals^m \to \naturals$ is encoded.
Our first result works with the dense representation, where the function $P$ is specified as a list of input-output pairs.
If the polynomials $f_1,\ldots,f_m$ have degree bounded by $d$, then the polynomial $\diamond_P(f_1,\ldots,f_m)$ depends only on the restricted function $P : \set{0,1,\ldots,d}^m \to \naturals$, which can be represented by a list of $(d+1)^m$ values.
Denote by $|P|_d$ the number of nonzero outputs of the restriction $P : \set{0,1,\ldots,d}^m \to \naturals$.
We will construct a circuit of constant depth and size polynomial in $m$ and $|P|_d$ that computes $\diamond_P(f_1,\ldots,f_m)$.

Our second result handles a more succinct representation of the function $P : \naturals^m \to \naturals$.
Namely, we consider functions represented by a special kind of circuit over the integers we call a ``tropical threshold circuit.'' (See \cref{def:tropical threshold circuit}.)
Again, we show that $\diamond_P(f_1,\ldots,f_m)$ can be computed efficiently with respect to the size of the given tropical threshold circuit that computes $P$.
When the circuit representing $P$ is itself of constant depth, we again compute $\diamond_P(f_1,\ldots,f_m)$ in $\AC^0_\F$.

We start by computing $\diamond_P(f_1,\ldots,f_m)$ when $P : \naturals^m \to \naturals$ is given in the sparse representation.
As in the case of two polynomials, we start with delta functions.

\begin{lemma} \label{lem:multi delta function}
	Let $\F$ be a field of characteristic zero or characteristic greater than $2md$.
	Let $f_1, \ldots, f_m \in \F[x]$ be univariate polynomials of degree at most $d$ given by their coefficients.
	Suppose that $f_i$ factors as
	\[
		f_i(x) = \prod_{j=1}^n (x - \alpha_j)^{a_{i,j}}
	\]
	where $\alpha_1, \ldots, \alpha_n \in \F$ and $a_{i,j} \in \naturals$, where some of the $a_{i,j}$ may be zero.
	Then for any $i_1, \ldots, i_m \in \set{0,1,\ldots,d}$, the coefficients of the polynomial 
	\[
		\diamond_{\delta_{i_1,\ldots,i_m}}(f_1,\ldots,f_m) \coloneqq \prod_{k : a_{1,k} = i_1 \land \cdots \land a_{m,k} = i_m} (x - \alpha_k)
	\]
	can be computed piecewise in $\AC^0_\F$.
\end{lemma}

\begin{proof}
	Without loss of generality, let $d = \max_{i \in [m]} \deg(f_i)$.
	For each $i \in [m]$, let $(f_{i,1}, \ldots, f_{i,d})$ be the squarefree decomposition of $f_i$.
	For notational ease, let $s$ be the squarefree part of $f_1 \cdots f_m$ and define $f_{i,0} \coloneqq s/f_i$.

	Observe that 
	\[
		\gcd(f_{1, i_1}(x), \ldots, f_{m, i_m}(x)) = \prod_{k : a_{1, k} = i_1 \land \cdots \land a_{m, k} = i_m} (x - \alpha_k) = \diamond_{\delta_{i_1,\ldots,i_m}}(f_1,\ldots,f_m).
	\]
	Thus, we can compute $\diamond_{\delta_{i_1,\ldots,i_m}}(f_1,\ldots,f_m)$ piecewise in $\AC^0_\F$ by first computing the squarefree decompositions of $f_i$ for all $i \in [m]$ using \cref{lem:squarefree decomp} and then computing $\gcd(f_{1,i_1},\ldots,f_{m,i_m})$ via \cref{thm:multigcd}.
\end{proof}

We now extend \cref{lem:multi delta function} to compute $\diamond_P(f_1,\ldots,f_m)$ for arbitrary functions $P : \naturals^m \to \naturals$.

\begin{theorem} \label{thm:multi arbitrary function}
	Let $\F$ be a field of characteristic zero or characteristic greater than $md$.
	Let $f_1, \ldots, f_m \in \F[x]$ be univariate polynomials of degree at most $d$ given by their coefficients.
	Let $P : \naturals^m \to \naturals$ be a function given in the dense representation.
	Suppose that $f_i$ factors as
	\[
		f_i(x) = \prod_{j=1}^n (x - \alpha_j)^{a_{i,j}}
	\]
	where $\alpha_1, \ldots, \alpha_n \in \F$ and $a_{i,j} \in \naturals$, where some of the $a_{i,j}$ may be zero.
	Then the coefficients of the polynomial 
	\[
		\diamond_P(f_1,\ldots,f_m) \coloneqq \prod_{j=1}^n (x - \alpha_j)^{P(a_{1, j}, \ldots, a_{m, j})}
	\]
	can be computed piecewise in $\AC^0_\F$.
\end{theorem}

\begin{proof}
	Using \cref{lem:multi delta function}, we can compute
	\[
		\diamond_{\delta_{i_1,\ldots,i_m}}(f_1,\ldots,f_m) = \prod_{k : a_{1,k} = i_1 \land \cdots \land a_{m,k} = i_m} (x - \alpha_k)
	\]
	piecewise in $\AC^0_\F$.
	We can then write $\diamond_P(f_1,\ldots,f_m)$ as
	\begin{align*}
		\diamond_P(f_1,\ldots,f_m)(x) &= \prod_{i_1=0}^d \cdots \prod_{i_m=0}^d \prod_{k : a_{1,k} = i_1 \land \cdots \land a_{m,k} = i_m} (x - \alpha_k)^{P(i_1,\ldots,i_m)} \\
		&= \prod_{i_1=0}^d \cdots \prod_{i_m=0}^d \diamond_{\delta_{i_1,\ldots,i_m}}(f_1,\ldots,f_m)^{P(i_1,\ldots,i_m)}.
	\end{align*}
	This yields a piecewise circuit of constant depth and size $m^{O(1)} |P|_d$ that computes $\diamond_P(f_1,\ldots,f_m)$, where $|P|_d$ denotes the number of tuples $(i_1, \ldots, i_m) \in \set{0,1,\ldots,d}^m$ for which $P(i_1, \ldots, i_m)$ is nonzero.
	Applying \cref{lem:polynomial interpolation} yields a piecewise $\AC^0_\F$ circuit that computes the coefficients of $\diamond_P(f_1,\ldots,f_m)$.
\end{proof}

For constant $m$, \cref{thm:multi arbitrary function} yields a circuit for $\diamond_P(f_1,\ldots,f_m)$ of size that is polynomially-bounded in the description length of $f_1,\ldots,f_m$.
When $m$ is large, for which functions $P : \naturals^m \to \naturals$ can we compute $\diamond_P(f_1,\ldots,f_m)$ more efficiently than \cref{thm:multi arbitrary function}?
We can do this when $P$ is one of the $m$-ary addition, maximum, minimum, or threshold functions.
By an $m$-ary threshold function, we mean a Boolean function of the form
\[
	\Thr_{\vec{r}}(a_1,\ldots,a_m) = \begin{cases}
		1 & \text{if $a_i \ge r_i$ for all $i \in [m]$}, \\
		0 & \text{otherwise,}
	\end{cases}
\]
where $\vec{r} \in \naturals^m$ is a fixed vector of thresholds, or its negation $\neg\Thr_{\vec{r}}$.

To compute $\diamond_P(f_1,\ldots,f_m)$ when $P$ is the addition function, we use the basic fact that 
\[
	\diamond_+(f_1,\ldots,f_m) = \prod_{i=1}^m f_i.
\]
For the maximum and minimum functions, this is a straightforward application of the identities
\begin{align*}
	\diamond_{\max}(f_1,\ldots,f_m) &= \lcm(f_1,\ldots,f_m) \\
	\diamond_{\min}(f_1,\ldots,f_m) &= \gcd(f_1,\ldots,f_m)
\end{align*}
combined with \cref{thm:multigcd} and \cref{cor:multilcm}.
When $P$ is a threshold function $\Thr_{\vec{r}}$ or its negation, the polynomial $\diamond_P(f_1,\ldots,f_m)$ corresponds to one of the outputs of \cref{lem:multithreshold}.

Naturally, if the function $P : \naturals^m \to \naturals$ has a succinct description in terms of these basic functions, then we can also hope to compute $\diamond_P(f_1,\ldots,f_m)$ efficiently.
We formalize this notion of complexity using tropical threshold circuits, defined below.
The adjective ``tropical'' refers to tropical geometry \cite{MS15}, a form of algebraic geometry done over the tropical semiring $(\reals \cup \set{+\infty}, \min, +)$ where addition and multiplication are replaced by minimum and addition, respectively.

\begin{definition} \label{def:tropical threshold circuit}
	A \emph{tropical threshold circuit} is a directed acyclic graph.
	Vertices of in-degree zero are called \emph{input gates} and each are labeled by some variable $x_i$.
	Vertices of positive in-degree are called \emph{internal gates} and are labeled by an element of $\set{+, c \times, \min, \max, \Thr_{\vec{r}}, \neg \Thr_{\vec{r}}}$.
	Vertices of out-degree zero are called \emph{output gates}.
	Each gate computes a function $\naturals^m \to \naturals$ in the natural way.
	The \emph{size} of the circuit is the number of gates in the circuit and the sum of all constants appearing in $c \times$ gates.
	The \emph{depth} of the circuit is the length of the longest path from an input gate to an output gate.
\end{definition}

We now state a version of \cref{thm:multi arbitrary function} that shows we can compute $\diamond_P(f_1,\ldots,f_m)$ in size and depth proportional to the size and depth of a tropical threshold circuit that computes the function $P : \naturals^m \to \naturals$.

\begin{theorem} \label{thm:ckt multi arbitrary function}
	Let $\F$ be a field of characteristic zero.
	Let $f_1, \ldots, f_m \in \F[x]$ be univariate polynomials of degree at most $d$ given by their coefficients.
	Let $P : \naturals^m \to \naturals$ be a function computed by a tropical threshold circuit of size $s$ and depth $\Delta$.
	Suppose that $f_i$ factors as
	\[
		f_i(x) = \prod_{j=1}^n (x - \alpha_j)^{a_{i,j}}
	\]
	where $\alpha_1, \ldots, \alpha_n \in \F$ and $a_{i,j} \in \naturals$, where some of the $a_{i,j}$ may be zero.
	Then the coefficients of the polynomial 
	\[
		\diamond_P(f_1,\ldots,f_m) \coloneqq \prod_{j=1}^n (x - \alpha_j)^{P(a_{1, j}, \ldots, a_{m, j})}
	\]
	can be computed piecewise by a circuit of depth $O(\Delta)$ and size $(smd)^{O(1)}$.
\end{theorem}

\begin{proof}
	We proceed by induction on $\Delta$, simulating the tropical threshold circuit that computes $P$.
	When $\Delta = 1$, the function $P$ must be computed by a single gate of a tropical threshold circuit.
	The following case analysis shows that for each gate type in a tropical threshold circuit, there is a piecewise arithmetic circuit of size $(md)^{O(1)}$ and depth $O(1)$ that computes $\diamond_P(f_1,\ldots,f_m)$.
	\begin{itemize}
		\item
			If $P$ is computed by an addition gate, we have
			\[
				\diamond_P(f_1,\ldots,f_m) = \prod_{i=1}^m f_i,
			\]
			which can clearly be computed by a circuit of size $(md)^{O(1)}$ and depth $O(1)$.
		\item
			If $P$ is computed by a $c \times$ gate, then
			\[
				\diamond_P(f) = f^c,
			\]
			which is easily computed by a circuit of size $(md)^{O(1)} + s$ and depth $O(1)$.
			Here, we use the fact $c \le s$ by definition, so implementing the powering operation $f \mapsto f^c$ can be done with at most $s$ wires.
		\item
			If $P$ is computed by a minimum gate, we have
			\[
				\diamond_P(f_1,\ldots,f_m) = \gcd(f_1,\ldots,f_m),
			\]
			which can be computed piecewise by a circuit of size $(md)^{O(1)}$ and depth $O(1)$ using \cref{thm:multigcd}.
			Likewise, if $P$ is computed by a maximum gate, we have
			\[
				\diamond_P(f_1,\ldots,f_m) = \lcm(f_1,\ldots,f_m),
			\]
			which can be computed in the claimed size and depth using \cref{cor:multilcm}.
		\item
			Suppose $P$ is computed by a threshold gate $\Thr_{\vec{r}}$ or its negation $\neg \Thr_{\vec{r}}$.
			Let $h(x)$ be the squarefree part of $f_1(x) \cdots f_m(x)$.
			We can compute $h$ piecewise with a circuit of size $(md)^{O(1)}$ and depth $O(1)$ using \cref{cor:squarefree part}.
			By applying \cref{lem:multithreshold} to $h$ and $f_1, \ldots, f_m$ with threshold vector $\vec{r}$, we obtain $\diamond_{\Thr_{\vec{r}}}(f_1,\ldots,f_m)$ and $\diamond_{\neg \Thr_{\vec{r}}}(f_1,\ldots,f_m)$ piecewise using a circuit of size $(md)^{O(1)}$ and depth $O(1)$.
	\end{itemize}
	Overall, we obtain a circuit of size $(md)^{O(1)}$ and depth $c$ for $\diamond_P(f_1,\ldots,f_m)$, where $c \in \naturals$ is some universal constant.

	Suppose now that $\Delta \ge 2$.
	Let $P_1,\ldots,P_s$ be the functions computed by the children of the output gate in the tropical threshold circuit computing $P$.
	By induction, we can compute $g_i \coloneqq \diamond_{P_i}(f_1,\ldots,f_m)$ in size $(smd)^{O(1)}$ and depth $c \cdot (\Delta - 1)$.
	Let $P_{\text{out}}$ be the function labeling the output gate of the circuit that computes $P$.
	By the same analysis as in the base case, we can compute
	\[
		\diamond_{P_{\text{out}}}(g_1,\ldots,g_s) = \diamond_P(f_1,\ldots,f_m)
	\]
	using $(md)^{O(1)}$ additional gates and increasing the depth by $c$.
	This yields a circuit that computes $\diamond_P(f_1,\ldots,f_m)$ of size $(smd)^{O(1)}$ and depth $c \cdot \Delta = O(\Delta)$ as claimed.
\end{proof}

\section{Extensions to Multivariate Polynomials} \label{sec:multivariate}

Suppose we are given an arithmetic circuit of size $s$ that computes a multivariate polynomial $f \in \F[\vec{x}]$.
What can we say, if anything, about the complexity of the factors of $f$?
A landmark result of \textcite{Kaltofen89} shows that every irreducible factor of $f$ can be computed by an arithmetic circuit of size $\poly(s, \deg(f))$.
This result is also constructive: there is a randomized polynomial-time algorithm that receives a circuit for $f$ as input and outputs circuits for the irreducible factors of $f$.
This was later extended to the black-box model by \textcite{KT90}, where the factorization algorithm is given access to an evaluation oracle for $f$ and must implement an evaluation oracle for each of the irreducible factors of $f$.
Kaltofen's algorithm raises a natural question: what other classes of circuits are closed under factorization?
By suitably modifying Kaltofen's algorithm, \textcite{ST21c} showed that $\VBP$, the class of polynomials computable by polynomial-size arithmetic branching programs, is closed under factorization.
It is an open question whether $\VNC^1$ or $\VAC^0$ are closed under factorization, but partial results are known for both classes \cite{DSY09, Oliveira16, CKS19a, DSS22}.
The fact that factors of small circuits can themselves be computed by small circuits is not only interesting in its own right, but it also plays a key role in the algebraic hardness versus randomness paradigm \cite{KI04}.
To establish hardness-to-pseudorandomness results for weaker circuit classes (where we have more hope of proving unconditional lower bounds), it would be enough to show that these weaker classes are closed under factorization.

In this section, we provide evidence that $\VAC^0$ and $\VNC^1$ are closed under factorization by showing that $\VAC^0$ and $\VNC^1$ are closed under the related operations of squarefree decomposition and GCD.
For this, we need to extend some of our results from the univariate to the multivariate setting.
By applying \cref{lem:squarefree decomp} in a straightforward manner, we show that for every polynomial $f \in \VAC^0$, all elements $(f_1,\ldots,f_m)$ of the squarefree decomposition of $f$ can themselves be computed in $\VAC^0$. 
Likewise, when $f \in \VNC^1$, all elements of the squarefree decomposition of $f$ can be computed in $\VNC^1$.
A similar application of \cref{thm:gcd} shows that $\VAC^0$ and $\VNC^1$ are closed under taking greatest common divisors.
Just like Kaltofen's theorem on factorization, these results are algorithmic, where the corresponding algorithms run in randomized polynomial time.
In the case of $\VAC^0$, these algorithms can be derandomized in subexponential time by using the deterministic subexponential-time polynomial identity testing algorithm of \cite{LST21a} for $\VAC^0$ circuits.

Before moving to the details, a word on why it is reasonable for algorithms for univariate polynomials to be applicable in the multivariate setting.
Thanks to Gauss's Lemma (\cref{lem:gauss}), questions related to multivariate factorization in $\F[\vec{x},y] \cong \F[\vec{x}][y]$ are often reducible to questions about factorization in $\F(\vec{x})[y]$.
That is, we can regard a multivariate polynomial $f(\vec{x}, y) \in \F[\vec{x}, y]$ as a univariate polynomial in $y$ whose coefficients come from the field $\K \coloneqq \F(\vec{x})$, and under certain assumptions on $f$, enlarging the ring of coefficients from $\F[\vec{x}]$ to $\F(\vec{x})$ does not affect the factorization of $f$.
In this setting, it is natural to apply our earlier univariate algorithms to solve factorization problems.
Although the field $\K$ contains elements of very high complexity, \cref{lem:polynomial interpolation} implies that when the input $f$ is represented as a small arithmetic circuit, the $\K$-coefficients of $f$ are also representable as small arithmetic circuits.
This allows us to apply our machinery from the univariate setting without a large blow-up in complexity.

Of course, we need to ensure that the answer to our problem over $\K[y]$ is the same as the answer over $\F[\vec{x}][y]$.
This is not always the case.
For example, over $\F[x_1, x_2][y]$, we have
\[
	\gcd(x_1 x_2, x_1) = x_1,
\]
whereas over $\F(x_1, x_2)[y]$, we have
\[
	\gcd(x_1 x_2, x_1) = 1.
\]
As we will see, Gauss's Lemma guarantees that if the polynomial $f \in \F[\vec{x}][y]$ is monic in $y$, then the factorizations of $f$ into irreducibles over $\F[\vec{x}][y]$ and over $\F(\vec{x})[y]$ are the same.
In particular, if we were to take the GCD in $\K[y]$ of two monic polynomials, we would obtain the same result as the GCD in $\F[\vec{x}][y]$.
Although our input polynomials are not guaranteed to be monic, we will see that there is a simple, standard way to reduce to the monic case.

\subsection{Preliminaries on Polynomial Factorization and Identity Testing}

We now recall some preliminary material on polynomial factorization and polynomial identity testing.
This material is standard, and the reader familiar with these problems can safely skip ahead to the next subsection, returning here as necessary.

\subsubsection{Gauss's Lemma}

We first state Gauss's Lemma, which allows us to pass from $\F[\vec{x}][y]$ to $\F(\vec{x})[y]$ when studying questions related to polynomial factorization.
The version of Gauss's Lemma we state here suffices for our purposes; for a more general statement of the lemma, see \textcite[Corollary 6.10]{vzGG13}.

\begin{lemma}[{Gauss's Lemma \cite[Corollary 6.10]{vzGG13}}] \label{lem:gauss}
	Let $f \in \F[\vec{x}, y]$ be monic in $y$.
	Then $f$ is irreducible in $\F(\vec{x})[y]$ if and only if $f$ is irreducible in $\F[\vec{x}][y] \cong \F[\vec{x},y]$.
\end{lemma}

If a polynomial $f \in \F[\vec{x}, y]$ is monic in $y$, then every irreducible factor of $f$ must also be monic in $y$, and so Gauss's Lemma implies that the irreducible factors of $f$ over $\F[\vec{x}][y]$ are also irreducible over $\F(\vec{x})[y]$.
In particular, for polynomials that are monic in $y$, their factorization into irreducibles is the same in $\F[\vec{x}][y]$ and $\F(\vec{x})[y]$.

\begin{corollary} \label{cor:monic preserves factorization}
	Let $f \in \F[\vec{x}, y]$ be monic in $y$.
	Then the factorization of $f$ into irreducibles is the same over $\F(\vec{x})[y]$ and $\F[\vec{x}][y]$.
\end{corollary}

\cref{cor:monic preserves factorization} is useful for reducing multivariate factorization problems to univariate ones.
If a polynomial $f \in \F[\vec{x},y]$ is monic in $y$, then we lose no information about the factorization of $f$ by viewing $f$ as a univariate polynomial in $y$ with coefficients in the field $\F(\vec{x})$.
A standard technique to reduce to the monic case is to apply a random change of variables to a polynomial $f \in \F[\vec{x}]$.

\begin{lemma} \label{lem:transform to monic}
	Let $f \in \F[\vec{x}]$ and let $\vec{\alpha} \in \F^n$.
	Let $d = \deg(f)$ and let $f_d \in \F[\vec{x}]$ be the top-degree homogeneous component of $f$.
	Let $y$ be a fresh variable and define
	\[
		\hat{f}(\vec{x},y) \coloneqq f(\vec{x} + y \cdot \vec{\alpha}) = f(x_1 + y \cdot \alpha_1, \ldots, x_n + y \cdot \alpha_n).
	\]
	Then the following hold.
	\begin{enumerate}
		\item
			If $f_d(\vec{\alpha}) \neq 0$, then $\frac{1}{f_d(\vec{\alpha})} \hat{f}(\vec{x}, y)$ is monic in $y$.
		\item
			$g(\vec{x})$ is an irreducible factor of $f(\vec{x})$ in $\F[\vec{x}]$ if and only if $g(\vec{x} + y \cdot \vec{\alpha})$ is an irreducible factor of $\hat{f}(\vec{x},y)$ in $\F[\vec{x}, y]$.
	\end{enumerate}
\end{lemma}

\begin{proof}
	Write $f$ as a sum of monomials
	\[
		f(\vec{x}) = \sum_{\vec{e} \in \naturals^n} c_{\vec{e}} \vec{x}^{\vec{e}}.
	\]
	Expand $\hat{f}(\vec{x},y)$ using the binomial theorem as
	\[
		\hat{f}(\vec{x}, y) = \sum_{\vec{e} \in \naturals^n} c_{\vec{e}} \cdot (\vec{x} + y \cdot \vec{\alpha})^{\vec{e}} = \sum_{\vec{e} \in \naturals^n} c_{\vec{e}} \sum_{\vec{a} + \vec{b} = \vec{e}} \binom{\vec{e}}{\vec{a}} \vec{x}^{\vec{a}} \vec{\alpha}^{\vec{b}} y^{\norm{\vec{b}}_1},
	\]
	where $\binom{\vec{e}}{\vec{a}} \coloneqq \prod_{i=1}^n \binom{e_i}{a_i}$.
	A term in the above summation is divisible by $y^d$ if and only if $\norm{\vec{b}}_1 \ge d$.
	Each term in the above sum satisfies $\norm{\vec{b}}_1 \le \norm{\vec{e}}_1$ and $\norm{\vec{e}}_1 \le \deg(f) = d$, so the terms divisible by $y^d$ are those for which $\norm{\vec{b}}_1 = \norm{\vec{e}}_1$, which implies $\vec{b} = \vec{e}$ and consequently $\vec{a} = \vec{0}$.
	Thus, the coefficient of $y^d$ in the expansion of $\hat{f}(\vec{x},y)$ is given by
	\[
		\sum_{\vec{e} \in \naturals^n} c_{\vec{e}} \vec{\alpha}^{\vec{e}} = f_d(\vec{\alpha}).
	\]
	Hence $\frac{1}{f_d(\vec{\alpha})} \hat{f}(\vec{x},y)$ is monic in $y$ if $f_d(\vec{\alpha}) \neq 0$.

	The second claim follows from the fact that the map $(\vec{x}, y) \mapsto (\vec{x} + y \cdot \vec{\alpha}, y)$ is invertible with inverse $(\vec{x}, y) \mapsto (\vec{x} - y \cdot \vec{\alpha}, y)$.
	That is, suppose $f(\vec{x}) = g(\vec{x}) h(\vec{x})$ where $g$ is irreducible and $h$ is nonzero.
	Then it is clear that 
	\[
		f(\vec{x} + y \cdot \vec{\alpha}) = g(\vec{x} + y \cdot \vec{\alpha}) \cdot h(\vec{x} + y \cdot \vec{\alpha}),
	\]
	so $g(\vec{x} + y \cdot \vec{\alpha})$ is a factor of $f(\vec{x} + y \cdot \vec{\alpha})$.
	To see that $g(\vec{x} + y \cdot \vec{\alpha})$ is irreducible, suppose that $g(\vec{x} + y \cdot \vec{\alpha})$ factors as
	\[
		g(\vec{x} + y \cdot \vec{\alpha}) = r(\vec{x}, y) \cdot s(\vec{x}, y).
	\]
	Then we have
	\[
		g(\vec{x}) = r(\vec{x} - y \cdot \vec{\alpha}, y) \cdot s(\vec{x} - y \cdot \vec{\alpha}, y).
	\]
	Because $g$ is irreducible, one of the above factors factors must be an element of $\F$.
	Without loss of generality, we have $s(\vec{x} - y \cdot \vec{\alpha}, y) \in \F$, so it follows that $s(\vec{x}, y) \in \F$, which implies that $g(\vec{x} + y \cdot \vec{\alpha})$ is irreducible.
	Thus, if $g(\vec{x})$ is an irreducible factor of $f(\vec{x})$, then $g(\vec{x} + y \cdot \vec{\alpha})$ is an irreducible factor of $\hat{f}(\vec{x}, y)$.

	A symmetric argument shows that if $\hat{g}(\vec{x}, y)$ is an irreducible factor of $\hat{f}(\vec{x},y)$, then $\hat{g}(\vec{x} - y \cdot \vec{\alpha}, y)$ is an irreducible factor of $f(\vec{x})$.
	Because $f$ does not depend on $y$, this implies that there is some polynomial $g(\vec{x})$ such that $g(\vec{x}) = \hat{g}(\vec{x} - y \cdot \vec{\alpha}, y)$, so $\hat{g}$ has the form $\hat{g}(\vec{x}, y) = g(\vec{x} + y \cdot \vec{\alpha})$ as claimed.
\end{proof}

\subsubsection{Polynomial Identity Testing}

Polynomial identity testing (abbreviated as PIT) is the algorithmic problem of deciding if a given arithmetic circuit $C$ computes the identically zero polynomial.
In arithmetic complexity, polynomial identity testing is the central example of a problem that can be solved efficiently with the use of randomness, but for which an efficient deterministic algorithm is not known.
The following lemma, often referred to as the Schwartz--Zippel lemma, provides the basis for an efficient randomized algorithm to solve PIT.

\begin{lemma}[\cite{Schwartz80,Zippel79}] \label{lem:schwartz zippel}
	Let $f \in \F[\vec{x}]$ be a nonzero polynomial of degree $d$.
	Then for any finite $S \subseteq \F$, we have
	\[
		\Pr_{\vec{\alpha} \in S^n}[f(\vec{\alpha}) = 0] \le \frac{d}{|S|}.
	\]
\end{lemma}

To test if a circuit $C$ of degree $d$ computes the zero polynomial, choose a set $S \subseteq \F$ of size $2d$ and evaluate $C$ at a randomly chosen point $\vec{\alpha} \in S^n$. 
The algorithm reports that $C$ computes the zero polynomial if and only if $C(\vec{\alpha}) = 0$.
It is clear that this algorithm has one-sided error (it is always correct when $C$ computes the zero polynomial), and the Schwartz--Zippel lemma implies that this algorithm has error probability at most $1/2$, so PIT is in $\coRP$.

The design of deterministic algorithms for polynomial identity testing is a large, active area of research, with close connections to arithmetic circuit lower bounds \cite{KI04, KS19}.
Our focus here is not on the design of new algorithms for PIT, but rather on applications of existing and future algorithms.
For an introduction to identity testing, we refer the reader to \textcite{SY10}.

There is a search-to-decision reduction for PIT, analogous to the search-to-decision reduction for Boolean satisfiability.
Suppose $f(\vec{x})$ is a nonzero polynomial of degree $d$.
\cref{lem:schwartz zippel} implies that for any set $S = \set{\alpha_1,\ldots,\alpha_{d+1}} \subseteq \F$, there is a point $\vec{\beta} \in S^n$ such that $f(\vec{\beta}) \neq 0$.
In particular, for some $i \in [d+1]$, the polynomial $f(\alpha_i, x_2, \ldots, x_n)$ is nonzero, and we can find such an $i$ by solving PIT for the polynomials $f(\alpha_i, x_2, \ldots, x_n)$.
This produces $\beta_1 \in \F$ such that $f(\beta_1, x_2,\ldots,x_n) \neq 0$. 
By recursion, we can find the remaining values $\beta_2,\ldots,\beta_n \in \F$ such that $f(\vec{\beta}) \neq 0$.
We record this search-to-decision reduction for polynomial identity testing in the following lemma.

\begin{lemma} \label{lem:pit search to decision}
	Let $\mathcal{C}$ be a class of arithmetic circuits closed under substitution.
	Given a $\mathcal{C}$-circuit $C(x_1,\ldots,x_n)$ of degree $d$, we can either determine that $C$ computes the zero polynomial or find point $\vec{\beta} \in \F^n$ such that $C(\vec{\beta}) \neq 0$ using $O(nd)$ calls to an oracle that solves PIT for $\mathcal{C}$-circuits.
\end{lemma}

The breakthrough lower bounds for $\AC^0_\F$ circuits by \textcite{LST21a} resulted in a deterministic subexponential-time PIT algorithm for $\AC^0_\F$ circuits.
We will make use of this algorithm later to derandomize our multivariate algorithms in subexponential time.

\begin{theorem}[\cite{LST21a}] \label{thm:ac0 pit}
	Let $\F$ be a field of characteristic zero.
	For all $\eps > 0$, there is a deterministic algorithm that receives as input an $n$-variate arithmetic circuit $C$ of size $s$ and depth $\Delta \le o(\log \log \log n)$ and decides if $C$ computes the zero polynomial in time $(s^{\Delta + 1} n)^{O(n^\eps)}$.
\end{theorem}

\subsection{Multivariate Algorithms from Univariate Algorithms}

The algorithms of \cref{sec:operations on roots,sec:gcd and lcm} for the squarefree decomposition, GCD, and LCM were all stated and proved for univariate polynomials over a field $\F$.
If we have a multivariate polynomial $f \in \F[\vec{x}]$, we can always regard $f$ as a univariate in one variable, say $x_n$, whose coefficients are from the field $\K = \F(x_1,\ldots,x_{n-1})$.
Our earlier univariate algorithms make no assumptions about the underlying field, other than requiring the characteristic to be sufficiently large. 
Because $\ch(\K) = \ch(\F)$, we can apply these univariate algorithms to $f \in \K[x_n]$ without any modification to the algorithm.
This raises two questions that must be addressed.
The first is an issue of representation: the field $\K = \F(x_1, \ldots, x_{n-1})$ is more complicated than $\F$, and we need a succinct way to encode the $\K$-coefficients of a polynomial $f \in \F[x_1,\ldots,x_n]$.
The second deals with correctness: why does solving a problem such as the GCD over $\F(x_1, \ldots, x_{n-1})[x_n]$ yield the solution to the same problem over $\F[x_1,\ldots,x_{n-1}][x_n] \cong \F[x_1, \ldots, x_n]$?

The issue of representation is easily dealt with.
The natural way to represent the input $f$ is via an arithmetic circuit.
Writing
\[
	f(\vec{x}) = \sum_{i=0}^d f_i(x_1, \ldots, x_{n-1}) x_n^i,
\]
we see that the $\K$-coefficients of $f$ are exactly the polynomials $f_0,\ldots,f_d$.
An application of \cref{lem:polynomial interpolation} shows that the coefficients $f_i$ can be computed by circuits whose size and depth is comparable to that of the circuit for $f$.

To deal with correctness, \cref{cor:monic preserves factorization} guarantees that our reduction to the univariate case will produce the correct answer, provided the polynomial $f$ is monic in some variable.
We achieve this using \cref{lem:transform to monic}, which provides a linear change of variables that ensure $f$ is monic in a fresh variable $y$.
Because the resulting change of variables is invertible, we will be able to easily recover the solution that corresponds to the original input $f$.

We now implement the argument sketched above to compute the squarefree decomposition of a multivariate polynomial represented by an arithmetic circuit.
We focus on the classes $\AC^0_\F$ and $\NC^1_\F$, as it was previously known that $\NC^2_\F$ and $\VBP$ were closed under squarefree decomposition.

\begin{theorem} \label{thm:multivariate squarefree decomp}
	Let $\F$ be a field of characteristic zero or characteristic greater than $d$.
	Let $\mathcal{C} \in \set{\VAC^0, \VNC^1}$.
	Let $\mathcal{O}$ be an oracle that solves polynomial identity testing for $\mathcal{C}$-circuits.
	There is a deterministic, polynomial-time algorithm with oracle access to $\mathcal{O}$ that does the following.
	\begin{enumerate}
		\item
			The algorithm receives as input a $\mathcal{C}$-circuit that computes a polynomial $f \in \F[\vec{x}]$ of degree $d$.
		\item
			The algorithm outputs a collection of $\mathcal{C}$-circuits $C_1,\ldots,C_m$ such that $C_i$ computes $f_i$, where $(f_1, \ldots, f_m)$ is the squarefree decomposition of $f$.
	\end{enumerate}
\end{theorem}

\begin{proof}
	Let $f_{\mathrm{top}} \in \F[\vec{x}]$ be the top-degree homogeneous component of $f$.
	Applying \cref{lem:polynomial interpolation} to $f$, we obtain a $\mathcal{C}$-circuit that computes $f_{\mathrm{top}}$.
	By \cref{lem:pit search to decision}, we can find a point $\vec{\alpha} \in \F^n$ such that $f_{\mathrm{top}}(\vec{\alpha}) \neq 0$ using $O(nd)$ calls to the PIT oracle $\mathcal{O}$.
	Let $y$ be a fresh variable and define 
	\[
		\hat{f}(\vec{x},y) \coloneqq \frac{1}{f_{\mathrm{top}}(\vec{\alpha})} f(\vec{x} + y \cdot \vec{\alpha}).
	\]
	By \cref{lem:transform to monic}, the polynomial $\hat{f}$ is monic and the irreducible factors of $\hat{f}$ are in one-to-one correspondence with those of $f$.

	The correspondence between the irreducible factors of $f$ and $\hat{f}$ implies that the elements of the squarefree decomposition of $f$ and $\hat{f}$ are in one-to-one correspondence.
	To see this, write the factorization of $f$ into irreducibles as 
	\[
		f = \prod_{i=1}^r f_i^{d_i},
	\]
	where each $f_i \in \F[\vec{x}]$ is irreducible and $\gcd(f_i, f_j) = 1$ for $i \neq j$.
	Because the factors of $f$ and $\hat{f}$ are in one-to-one correspondence, the factorization of $\hat{f}$ into irreducibles is given by
	\[
		\hat{f} = \prod_{i=1}^r f_i(\vec{x} + y \cdot \vec{\alpha})^{d_i}.
	\]
	The squarefree decomposition of $\hat{f}$ consists of the polynomials
	\[
		\hat{g}_j \coloneqq \prod_{i : d_i = j} f_i(\vec{x} + y \cdot \vec{\alpha})
	\]
	for $j \in [d]$.
	Applying the change of variables $(\vec{x}, y) \mapsto (\vec{x} - y \cdot \vec{\alpha}, y)$ to $\hat{g}_j$, we obtain
	\[
		g_j \coloneqq \prod_{i : d_i = j} f_i(\vec{x}),
	\]
	which is precisely the $j$\ts{th} element of the squarefree decomposition of $f$.
	Thus, if we can compute the squarefree decomposition of $\hat{f}$, then the change of variables $(\vec{x}, y) \mapsto (\vec{x} - y \cdot \vec{\alpha}, y)$ yields the squarefree decomposition of $f$.

	It remains to compute the squarefree decomposition of $\hat{f}$.
	Write $\hat{f}(\vec{x}, y) = \sum_{i=0}^d \hat{f}_i(\vec{x}) y^i$.
	By \cref{lem:polynomial interpolation}, each $\hat{f}_i$ can be computed by a $\mathcal{C}$-circuit.
	Regarding $\hat{f}(\vec{x}, y) \in \F(\vec{x})[y]$ as a polynomial in $y$ with coefficients in $\F(\vec{x})$, the coefficients of $\hat{f}$ are precisely $\hat{f}_0, \hat{f}_1, \ldots, \hat{f}_d$.
	Applying the piecewise $\AC^0_\F$ algorithm (recall \cref{def:piecewise circuits}) of \cref{lem:squarefree decomp} to $\hat{f}$, using $\hat{f}_0, \hat{f}_1, \ldots, \hat{f}_d$ as the coefficients of $\hat{f}$, we obtain a piecewise arithmetic circuit that correctly computes the squarefree decomposition of $\hat{f}$.
	This circuit is the composition of a $\mathcal{C}$-circuit and an $\AC^0_\F$ circuit, which results in a $\mathcal{C}$-circuit.

	It remains to remove the selection gates from this circuit to obtain a standard arithmetic circuit that computes the squarefree decomposition of $\hat{f}$.
	We use the PIT oracle $\mathcal{O}$ to remove the selection gates one by one, proceeding in topological order.
	Let $v$ be the topologically first selection gate and let $w_1, \ldots, w_t$ be its children.
	Because $v$ is the first selection gate in topological order, the subcircuits rooted at $w_1, \ldots, w_t$ contain no selection gates and hence are standard arithmetic circuits.
	For $i \in [t]$, let $g_{i}(\vec{x})$ be the polynomial computed by gate $w_i$.
	Using the PIT oracle $\mathcal{O}$, we can compute the set $\set{i \in [t] : g_i(\vec{x}) \neq 0}$, as well as $i^\star \coloneqq \min\set{i \in [t] : g_i(\vec{x}) \neq 0}$ if this set is nonempty.
	The semantics of the selection gate imply that $v$ computes either $g_{i^\star}(\vec{x})$ if $i^\star$ is defined or the constant $0$ if $i^\star$ is not defined.
	By replacing $v$ with the gate $w_{i^\star}$ or zero as appropriate, we eliminate one selection gate while preserving the correctness of the circuit.
	Note that this procedure also preserves the size, depth, and fan-in of the circuit.
	Continuing iteratively, we eliminate all selection gates, resulting in a standard arithmetic circuit that outputs the squarefree decomposition of $\hat{f}$.
\end{proof}

Using an essentially identical argument, we can design an algorithm to compute the GCD and LCM of multiple polynomials given as arithmetic circuits.
As above, we focus on the classes $\VAC^0$ and $\VNC^1$, since $\VNC^2$ and $\VBP$ were known to be closed under taking GCD's and LCM's.

\begin{theorem} \label{thm:multivariate gcd and lcm}
	Let $\F$ be a field of characteristic zero or characteristic greater than $m^2 d$.
	Let $\mathcal{C} \in \set{\VAC^0, \VNC^1}$.
	Let $\mathcal{O}$ be an oracle that solves polynomial identity testing for $\mathcal{C}$-circuits.
	There is a deterministic, polynomial-time algorithm with oracle access to $\mathcal{O}$ that does the following.
	\begin{enumerate}
		\item
			The algorithm receives as input $m$ $\mathcal{C}$-circuits that compute polynomials $f_1,\ldots,f_m \in \F[\vec{x}]$ of degree at most $d$.
		\item
			The algorithm outputs $\mathcal{C}$-circuits $C_{\gcd}$ and $C_{\lcm}$ that compute $\gcd(f_1,\ldots,f_m)$ and $\lcm(f_1,\ldots,f_m)$, respectively.
	\end{enumerate}
\end{theorem}

\begin{proof}
	The proof is essentially the same as \cref{thm:multivariate squarefree decomp}, replacing the use of \cref{lem:squarefree decomp} with either \cref{thm:multigcd} or \cref{cor:multilcm} as appropriate.
	The only difference is the fact that we need to find a change of variables that makes all of the polynomials $f_1, \ldots, f_m$ simultaneously monic in a fresh variable $y$.

	To find this change of variables, let $d_i \coloneqq \deg(f_i)$.
	For $i \in [m]$, let $f_{i, \mathrm{top}} \in \F[\vec{x}]$ be the top-degree homogeneous component of $f_i$.
	Applying \cref{lem:polynomial interpolation} to $f_i$, we obtain a $\mathcal{C}$-circuit that computes $f_{i, \mathrm{top}}$.
	It follows that the product $\prod_{i=1}^m f_{i, \mathrm{top}}$ can be computed by a $\mathcal{C}$-circuit.
	\cref{lem:pit search to decision} implies that we can find a point $\vec{\alpha} \in \F^n$ such that $\prod_{i=1}^m f_{i, \mathrm{top}}(\vec{\alpha}) \neq 0$ using $O(nmd)$ calls to the PIT oracle $\mathcal{O}$.
	This implies that $f_{i, \mathrm{top}}(\vec{\alpha}) \neq 0$ for each $i \in [m]$.

	Let $y$ be a fresh variable and define
	\[
		\hat{f}_i(\vec{x}, y) \coloneqq \frac{1}{f_{i, \mathrm{top}}(\vec{\alpha})} f_i(\vec{x} + y \cdot \vec{\alpha}).
	\]
	Because $f_{i, \mathrm{top}}(\vec{\alpha}) \neq 0$ for each $i \in [m]$, \cref{lem:transform to monic} implies that $\hat{f}_i$ is monic in $y$ for each $i \in [m]$.
	At this point, the proof proceeds in the same manner as the proof of \cref{thm:multivariate squarefree decomp}.
\end{proof}

Finally, we observe that in the case of $\VAC^0$ circuits, the preceding algorithms for the multivariate squarefree decomposition, GCD, and LCM can all be derandomized in subexponential time.
This is done by implementing the PIT oracles in \cref{thm:multivariate squarefree decomp,thm:multivariate gcd and lcm} using \cref{thm:ac0 pit}.

\begin{corollary}
	In the case $\mathcal{C} = \VAC^0$, the algorithms of \cref{thm:multivariate squarefree decomp,thm:multivariate gcd and lcm} can be implemented in deterministic subexponential time.
\end{corollary}

\section{Conclusions and Open Problems} \label{sec:conclusion}

In this work, we introduced techniques for manipulating evaluations of symmetric polynomials $\set{e_k(\alpha_1,\ldots,\alpha_n) : k \in [n]}$ and $\set{e_k(\beta_1,\ldots,\beta_m) : k \in [m]}$ without having explicit access to the $\alpha_i$ and $\beta_j$.
We have seen how to change the points at which the $e_i$ are evaluated (\cref{sec:symmetric functions of roots}) and how to implement combinatorial operations like set difference, intersection, and union on the sets $\set{\alpha_1,\ldots,\alpha_n}$ and $\set{\beta_1,\ldots,\beta_m}$ (\cref{sec:operations on roots,sec:gcd and lcm}).
We also saw how to incorporate multiplicities into these set operations when $\set{\alpha_1,\ldots,\alpha_n}$ and $\set{\beta_1,\ldots,\beta_m}$ are regarded as multisets.
These techniques naturally led to efficient algorithms for symmetric functions (such as the discriminant) and bi-symmetric functions (such as the resultant and the GCD).
How far can these techniques be pushed, and what are their limitations?

A natural question along these lines is to understand which symmetric polynomials can be computed in $\AC^0_\F$.
\textcite{CKLMSS23} showed that certain Schur polynomials are at least as hard as the determinant, hence are not in $\AC^0_\F$.
\textcite{CLS22} showed that some families of monomial symmetric polynomials are $\VNP$-hard, i.e., a polynomial-size circuit for them would imply that the permanent has polynomial-size circuits.
As the permanent is known to not be in $\AC^0_\F$, the same lower bound holds for the monomial symmetric polynomials.
Can we hope to classify which families of symmetric polynomials are in $\AC^0_\F$?

A related question deals with the complexity of the fundamental theorem of symmetric polynomials.
If $f(x_1,\ldots,x_n)$ is symmetric, then there is a unique $g$ such that $f(\vec{x}) = g(e_1(\vec{x}), \ldots, e_n(\vec{x}))$.
How do the complexity of $f$ and $g$ relate to one another?
\textcite{BJ19} showed that $f$ and $g$ have comparable circuit complexity.
A variation on this for arithmetic branching programs is a key ingredient in the result of \textcite{CKLMSS23} mentioned above.
What can be said about the complexity of $f$ and $g$ when they are written as formulas, or as low-depth circuits?

Many problems in polynomial algebra can be solved in $\NC^2_\F$ by reduction to linear algebra.
As we have seen, these algorithms can be improved all the way to $\AC^0_\F$ for the GCD and related problems.
However, this only scratches the surface of linear algebra's applications to polynomial algebra.
The fact that these natural $\NC^2_\F$ algorithms can be improved suggests that there is more to say about the parallel complexity of algebraic problems.

One interesting problem left open by this work is to determine the parallel complexity of the extended Euclidean scheme.
The \emph{extended Euclidean scheme} of two polynomials $f, g \in \F[x]$ consists of the sequence of remainders and B\'{e}zout coefficients produced during the execution of the Euclidean algorithm on input $f$ and $g$.
Linear algebra provides an $\NC^2_\F$ algorithm to compute the extended Euclidean scheme \cite{vonzurGathen84}.
Can the techniques of this paper be used to design circuits of lower depth?
The extended Euclidean scheme has many applications, including rational interpolation, Pad\'{e} approximation, continued fraction expansion, counting real roots of polynomials in $\reals[x]$ (as a consequence of Sturm's theorem), and solving Toeplitz and Hankel systems of linear equations (see \cite[Chapters 4 and 5]{vzGG13} and \cite{BGY80} for details).
As a step towards understanding the complexity of the extended Euclidean scheme, can we find faster parallel algorithms for any of these related problems?

\printbibliography

\end{document}